\documentclass[a4paper,
               11pt
           twoside,
               ]{article}                  
       
\usepackage[T1]{fontenc}                    

\usepackage[utf8]{inputenc}                 

\usepackage[italian,english]{babel}         

\usepackage{sectsty}
\allsectionsfont{\boldmath}

\usepackage{indentfirst}                    

\usepackage{booktabs}                       

\usepackage{graphicx}                       

\usepackage{subfig}                         

\usepackage{caption}                        

\usepackage{listings}                       

\usepackage{amsmath,amssymb,amsthm}         

\usepackage[italian]{varioref}              

\usepackage{mparhack,fixltx2e,relsize}      

\usepackage[numbers,sort&compress]{natbib}

\usepackage[dvipsnames]{xcolor}             

\usepackage{hyperref}                       
\usepackage{bookmark}                       

\usepackage{comment}
\specialcomment{draft}{\begingroup\color{gray}}{\endgroup}

\normalsize

\numberwithin{equation}{section}

\newcommand{\numberset}{\mathbb} 
 
\newcommand{\R}{\numberset{R}}

\newcommand{\C}{\numberset{C}}  
\renewcommand{\Re}{\mathop{\mathrm{Re}}}
\renewcommand{\Im}{\mathop{\mathrm{Im}}}

\newcommand{\vl}{\guillemotleft}
\newcommand{\vr}{\guillemotright}
\newcommand{\spt}{\scriptscriptstyle}
\newcommand{\ssp}{\scriptstyle}

  {\left\lbrace\begin{array}{@{}l@{}}}%
  {\end{array}\right.}

\theoremstyle{definition}

\theoremstyle{plain}

\usepackage{color}
\usepackage{cases}
\usepackage{fancybox}
\usepackage{empheq}
\definecolor{shadowcolor}{rgb}{0,.5,.5}
\usepackage{braket}
\usepackage[makeroom]{cancel}

\definecolor{light-gray}{gray}{0.45}

\providecommand*{\iu}{\ensuremath{\mathrm{i}}}

\DeclareMathOperator{\erf}{erf}
\DeclareMathOperator*{\tr}{Tr}
\definecolor{shadecolor}{cmyk}{0,0,0.41,0}
\definecolor{light-blue}{rgb}{.9, .9, 5}
\newsavebox{\mysaveboxM} 
\newsavebox{\mysaveboxT} 

\newcommand\restr[2]{{
  \left.\kern-\nulldelimiterspace 
  #1 
  \vphantom{\big|} 
  \right|_{#2} 
  }}

\theoremstyle{plain}
\newtheorem{thm}{Theorem}[section]
\newtheorem{corth}{Corollary}[thm]
\newtheorem{lem}{Lemma}[section]

\newtheorem{corlem}{Corollary}[lem]

\newtheorem{cor}{Corollary}[section]
\newtheorem{conj}{Conjecture}[section]

\theoremstyle{definition}
\newtheorem{defask}{Definition}[section]
\newtheorem{defdp}{Definition}[defask]

\theoremstyle{remark}
\newtheorem{remark}{\textsc{remark}}

\newenvironment{lem*} 
  {\pushQED{\qed}\lem} 
  {\popQED\endlem}
\newenvironment{corlem*} 
  {\pushQED{\qed}\corlem} 
  {\popQED\endcorlem}

\newenvironment{thm*} 
  {\pushQED{\qed}\thm} 
  {\popQED\endthm}
\newenvironment{corth*} 
  {\pushQED{\qed}\corth} 
  {\popQED\endcorth}

\DeclarePairedDelimiter{\abs}{\lvert}{\rvert}

\newcommand{\e}{\mathrm{e}}
\newcommand{\nc}{\spt{\textsc{nc}}}
\newcommand{\cc}{\spt{\textsc{c}}}
\newcommand{\Z}{\numberset{Z}}

\newcommand{\Q}{\numberset{Q}}
\newcommand{\0}{\setminus \left\{{0}\right\}}

\usepackage{caption}
\usepackage{bm}
\usepackage{mathrsfs}
\usepackage{bbm}
\usepackage{enumitem}

\usepackage{tikz}
\usepackage{pgfplots,relsize}
\pgfplotsset{/pgf/number format/use comma}

\colorlet{circle edge}{blue!50}
\colorlet{circle area}{blue!20}

\tikzset{filled/.style={fill=circle area, draw=circle edge, thick},
    outline/.style={draw=circle edge, thick}}

\def\firstcircle{(0,0) circle (6.8cm)}
\def\secondcircle{(180:1.4cm) circle (2.9cm)}
\def\thirdcircle{(45:2cm) circle (2.3cm)}
\def\fourthcircle{(170:0.7cm) circle (1.35cm)}
\def\fifthcircle{(165:1.4cm) circle (0.23cm)}
\def\sixthcircle{(45:1.5cm) circle (4.2cm)}

\usepackage{preview}
\PreviewEnvironment{tikzpicture}
\usetikzlibrary[positioning]
\usetikzlibrary{patterns}
\newcommand{\daywidth}{8 em}
\tikzset{day/.style  = {draw, text centered,
                       minimum height = 1cm,
                        minimum width  = \daywidth,
                        anchor         = south west,
                         text width = 12.35 em}}
\tikzset{dayF/.style  = {draw, text centered,
                       minimum height = 1cm,
                        minimum width  = 3cm,
                        anchor         = south west,
                         text width = 14 em}}
\tikzset{title/.style  = {draw,
                        minimum height = 1 cm,
                        minimum width  = 1 cm,
                           fill =      orange!25,
                        text           = black,
                        text width    = 13 em}}
\tikzset{hour/.style = {draw, text centered,
                        minimum height = 1.5 cm,
                        minimum width  = 1.5 cm,
                        fill           = yellow!30,
                        anchor         = north east,
                        text width    = 13 em}}
\tikzset{hourF/.style = {draw, text centered,
                        minimum height = 2 cm,
                        minimum width  = 1.5cm,
                        fill           = yellow!30,
                        anchor         = north east,
                        text width    = 13 em}}
\tikzset{hours/.style = {draw, text centered,
                         minimum width = \daywidth,
                         anchor        = north west,
                         text width = 13 em}}
\tikzset{hoursF/.style = {draw, text centered,
                         minimum width = 3cm,
                         anchor        = north west,
                         text width = 14 em}}
\tikzset{hoursF2/.style = {draw, text centered,
                         minimum width = 4cm,
                         anchor        = north west,
                         text width = 18em}}
\tikzset{hoursF2m/.style = {draw, 
                         minimum width = 4cm,
                         anchor        = north west,
                         text width = 18em}}
\tikzset{hoursFM/.style = {draw, text centered,
                         minimum width = 3cm,
                         anchor        = north west,
                         text width = 13 em}}
\tikzset{1hour/.style ={hours, minimum height = 1.5cm}}

\tikzset{1hourF/.style ={hoursF, minimum height = 1.5cm}}

\tikzset{1hourF2/.style ={hoursF2, minimum height = 2cm}}
\tikzset{1hourF2m/.style ={hoursF2m, minimum height = 2cm}}

\tikzset{2hours/.style ={hours, minimum height = 1.5cm}}

\tikzset{3hours/.style ={hoursF, minimum height = 2cm}}

\tikzset{3hoursM/.style ={hoursFM, minimum height = 3.4cm}}
\tikzset{6B/.style    ={1hour, fill = blue!10}}
\tikzset{6BF/.style    ={1hourF, fill = blue!10}}
\tikzset{2A/.style    ={1hour, fill = green!10}}
\tikzset{2AF/.style    ={1hourF, fill = green!10}}
\tikzset{2AF2/.style    ={1hourF2, fill = green!10}}
\tikzset{2AF2m/.style    ={1hourF2m, fill = green!10}}

\tikzset{TESSP/.style ={1hour, fill = blue!20}}
\tikzset{TES/.style   ={1hour, fill = blue!10}}
\tikzset{PESSP/.style ={1hour, fill = magenta!50}}
\tikzset{Empty/.style ={1hour, fill = lightgray!30}}

\hypersetup{%
    hyperfootnotes=true,pdfpagelabels, 
    colorlinks=true, linktocpage=true, pdfstartpage=1, pdfstartview=FitV,%
    breaklinks=true, pdfpagemode=UseNone, pageanchor=true, pdfpagemode=UseOutlines,%
    plainpages=false, bookmarksnumbered, bookmarksopen=true, bookmarksopenlevel=1,%
    hypertexnames=true, pdfhighlight=/O,
    urlcolor=blue, linkcolor=RoyalBlue, citecolor=webgreen,
    pdfsubject={},%
    pdfkeywords={},%
    pdfcreator={pdfLaTeX},%
    pdfproducer={LaTeX with hyperref and ClassicThesis}%
}

\graphicspath{{Immagini/}} 

\definecolor{webgreen}{rgb}{0,.5,0}
\definecolor{webbrown}{rgb}{.6,0,0}

\begin{document}

\begin{titlepage}

\begin{center}
{ \textbf{\vspace{-2cm}\\
\fontsize{13.2}{15}\selectfont{\mbox{VORTEX ON A NON-COMMUTATIVE TORUS IN}\\
\vspace{.5em}
THE DUAL SUPERCONDUCTIVITY MODEL}}}\\ 
\vspace{1em}
{ \textsc{\fontsize{9.7}{15}\selectfont{a link between a cyclic $\tau\subset U(1)$, the color-electric charge of quarks and the vortex mass}}}

         \vspace{.7cm}
         {\scshape{Andrea Spirito}} 

\vspace{.15cm}

\textit{Dipartimento di Matematica e Fisica ``Ennio De Giorgi'' \\Via Arnesano, I-73100 Lecce, Italy}\\
\medskip
\textit{E-mail:} \href{mailto:andrea.spirito@le.infn.it}{\texttt{andrea.spirito@le.infn.it}}
     \vspace{.9cm}

      \textbf{Abstract}
 \end{center}
{\fontsize{9.5}{12}\selectfont {The model of \emph{dual superconductivity}, used to explain the static quark confinement, has been revisited considering the typical $U(1)$-gauged Ginzburg-Landau lagrangian density on a non-commutative torus $\mathcal{T}^{2}_{\nc}$, according to a new approach, we propose, in dealing with non-commutative space coordinates ($[\hat{x},\hat{y}]=\iu\theta$).\\ 
This led to consider a different set of the \emph{twist matrices} $\Omega_{\mu}$ relative to $\mathcal{T}^{2}_{\nc}$, since the corrisponding set, adopted in previous works, has resulted to be incompatible with the homogeneity of $\mathcal{T}^{2}_{\nc}$. Beyond this, we have found other differences which concern inter alia the observables of the model. Essentially, here, the index labelling the \emph{twists} $n\in \Z$ no longer has the usual physical role. In fact, the energy, suitably rewritten, of the minimum configurations at the point of Bogomolny and the quark color-electric charge $q_{e}$ depend on $r_{\spt{\theta}}=\theta_{\cc}/\theta_{\nc}\in\Q\0$, where $\theta_{\cc}$ ($\theta_{\nc}$) is the parameter that characterizes the commutative torus $\mathcal{T}^{2}_{\cc}$ ($\mathcal{T}^{2}_{\nc}$): here $\mathcal{T}^{2}_{\cc}$ is not characterized only by a $\theta$ null. The quantity $r_{\spt{\theta}}$ was then found to determine the order of a cyclic subgroup of $U(1)$ generated by $\e^{\iu\pi\theta_{\nc}}$. Therefore here $q_{e}$ can also be a fraction as well as a multiple of $g^{-1}$, where $g$ is the magnetic charge \mbox{of the scalar field $\phi$.} \\
Furthermore, the above energy was calculated only considering an appropriate class of the \emph{twisted boundary conditions} solutions without solving, as usual, any system of equations such as the BPS (Bogomolny-Prassad-Sommerfeld).\\
Another interesting novelty is the presence of a double Higgs mechanism regulated by $r_{\spt{\theta}}$. In this regard, it predicts, in general, a non null mass for the usual massless Goldstone boson and establishes a relation between the mass spectrum and $r_{\spt{\theta}}$ according to which a variation over time of $r_{\spt{\theta}}$ involves a quantized variation in time of such masses that will be null when $r_{\spt{\theta}}$ assumes values outside a certain interval of values.\\

}}

\end{titlepage}



\hrule
 \pdfbookmark{\contentsname}{tableofcontents}
\tableofcontents
\markboth{\contentsname}{\contentsname} 
\medskip\medskip\smallskip 
\hrule
\bigskip\bigskip





\section{Introduction}
\label{sec:intr}

\noindent According to the definition given by~\citep{Manton2004}, a vortex is a configuration of a scalar field $\phi$ in two dimensions characterized
by a \emph{core} of a finite size.\\
This particular configuration has attracted many interests and different types of generalizations~\citep{Vilenkin1994,Manton2004,EWeinberg2012}.\\
\indent Among these interests, in particular, there is the one concerning the confinement of static quarks.
Different models have been proposed to explain such phenomenon~\citep{Hooft2007}, among them there is the so-called \emph{dual superconductivity model}~\citep{Mandelstam1975,Hooft75,Parisi75}.  
In this context, $\phi$ represents a color-magnetic monopoles condensate and the \emph{core} is a bounded portion of the space between two color-electric charges (representing a quark-antiquark pair $\mathbf{q}\mathbf{\bar{q}}$), whose the edges are identified by the so-called \emph{twisted boundary conditions} (TBC)~\citep{Hooft1979}. 
These determine the appearance of a flux string in-between $\mathbf{q}$ and $\mathbf{\bar{q}}$ which produces the static potential responsible for the confinement (\emph{Cornell potential}). The above mentioned \emph{core} represents but a section of the space orthogonal to the direction of such flux string (see figure~\ref{fig:Scdual}).\\ 
\begin{figure}
\centering
\includegraphics[width=%
0.8\textwidth] {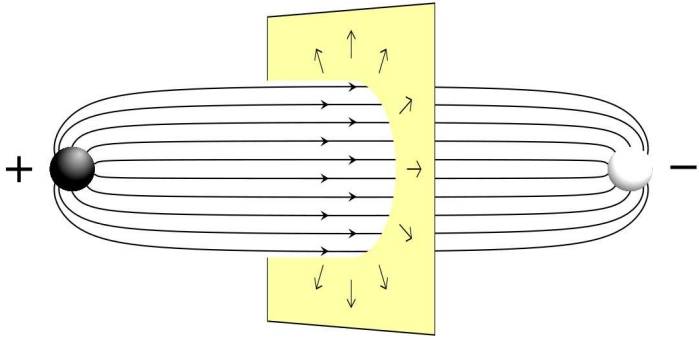} 
\caption{``Plaquette'' bidimensional that represents both the \emph{core} of a vortex $\phi$ and the orthogonal section of the flux string related to Cornell potential~\citep{Hooft2007}.}
\label{fig:Scdual}
\end{figure}
\hspace{-.45em}In particular, the \emph{abelian projection mechanism}~\citep{Hooft1981} on TBC allows the formation of the flux string. In fact this mechanism breaks the group $SU(3)$ of quarkless $\mathrm{QCD}$ into the maximal abelian subgroup and thus we can predicted the existence of a component of the gauge invariant field tensor $F_{\mu\nu}$, orthogonal to the \emph{core} and related to the flux string.\\
Furthermore, the twists around the \emph{core} determine the quantization of such flux and, ultimately, the quantization and sign of the color-electric charge of the quarks.
These twists, classified according to a winding number $n$, correspond to the elements of the fundamental homotopy group $\pi_{1}(SU(N)/Z_{N})\cong Z_{N}$~\citep{Baal82}. As a matter of fact, we can consider $SU(N)/Z_{N}$ as the effective gauge group since the gauge potential $A_{\mu}$ turns trivially under the $Z_{N}$ center of $SU(N)$.\\
\indent In this frame, we can understand how the system, although originally non-abelian, can be described by a U(1)-gauged Ginzburg-Landau lagrangian density $\mathcal{L}$ on a torus $\mathcal{T}^{2}$:
\begin{equation}
\label{eqn:dL}
\mathcal{L}=-\frac{1}{4}F_{\mu\nu}F^{\mu\nu}+(D_{\mu}\phi)^{\dagger}(D^{\mu}\phi)-\lambda(\phi^{\dagger}\phi-\phi^{2}_{0})^{2},
\end{equation}
to which corresponds the energy 
\begin{equation}
\label{eqn:energy'}
E=\int_{\mathcal{T}^{2}}\biggl(\frac{1}{4}F_{\mu\nu}F^{\mu\nu}+(D_{\mu}\phi)^{\dagger}(D^{\mu}\phi)+\lambda(\phi^{\dagger}\phi-\phi^{2}_{0})^{2}\biggr),
\end{equation}
where $D_{\mu}=\partial_{\mu}-\iu g A_{\mu}$, with $\mu=1,2$, is the covariant derivative with respect to $A_{\mu}$ ($U(1)$ gauge potential) and $\lambda, g$ are coupling constants.\\ 
Therefore, a phenomenology similar to that of the superconductors can be expected. The vacuum expectation value $\Braket{\phi}$ is the order parameter for the system, which makes a transition from a confined phase at $\mu^{2}=-2\phi^{2}_{0}\lambda<0 \Leftrightarrow\Braket{\phi}=\phi_{0}$, to a deconfined one for $\mu^{2}>0\Leftrightarrow\Braket{\phi}=0$. In the former case the typical lines of force are concentrated into a flux string, whereas in the latter they are distributed in the space like those of an electric dipole. \pagebreak \\
Later, it was observed~\citep{Bogomolny1976} that $E$ admitted a rewriting (called \textit{à la} Bogomolny), according to which, at a specific value of $\lambda$ (i.e.\,$\lambda=g^{2}/2$), $E$ presented a lower limit that
was achieved when these fields met the so-called Bogomolny-Prassad-Sommerfeld (BPS) two equations\footnote{They are first degree equations and we can easily show~\citep{Manton2004} that fields that meet these equations, automatically meet also the Ginzburg-Landau  second degree equations (GL eqs). We understand thus how the GL eqs represent ``weaker'' conditions compared to BPS eqs for the minimum-energy configurations.}(later they will be called I BPS eq and II BPS eq) and that depended on $n$.\\
From II BPS eq, a relation was later obtained~\citep{Bradlow90} according to which 
\begin{center}
$|\phi|^{2}\neq 0\Leftrightarrow \mathcal{A}-4\pi\, n>0$\quad (Bradlow's inequality),
\end{center}
with $\mathcal{A}$ area of $\mathcal{T}^{2}$ and $n$ that, in this context, appeared to determine the number of vortices. \\

The interest in considering the coordinates of $\mathcal{T}^{2}$ non-commutative came after a study by Gonz{\'a}lez-Arroyo and Ramos about solving BPS eqs.\\
In particular, it was proposed to solve them getting first from TBC some parameterizations for $A_{\mu}$ and $\phi$ and then expanding them in power series of \mbox{$\epsilon=1-4\pi n/\mathcal{A}$} (called Bradlow parameter), in such a way to obtain from BPS eqs certain recursion relations for the coefficients of the expansions. \\
\indent Such methodology was later revisited from a non-commutative perspective according to ``Fock space approach''~\citep{Lozano:2006xn}.\,\,This time, anyway, instead of $A_{\mu}$ and $\phi$ there were the operators $A_{\mu}(\hat{x},\hat{y})$ and $\phi(\hat{x},\hat{y})$ where $\hat{x}$ and $\hat{y}$ indicated the non-commutative coordinates such that $[\hat{x},\hat{y}]=\iu\theta\hat{\mathbbm{1}}$, with $\theta\in\R$ parameter of non-commutativity and $\hat{\mathbbm{1}}$ central element of non-commutative algebra of the coordinates. We obtained thus the operators $\hat{F}_{\mu\nu}=F_{\mu\nu}(\hat{x},\hat{y})$ and $|\hat{\phi}|^{2}=|\phi(\hat{x},\hat{y})|^{2}$ to which the quantities $F_{\mu\nu}(x,y)=\Omega^{-1}_{\spt{W}}\hat{F}_{\mu\nu}$ and $|\phi(x,y)|^{2}=\Omega^{-1}_{\spt{W}}|\hat{\phi}|^{2}$ were then associated with $\Omega_{\spt{W}}$ Weyl map~\citep{Weyl31}:\,\,isomorphism between the algebra of operators in ``Fock space'' and the algebra of functions on $\mathcal{T}^{2}$ multiplied with the non-commutative Gr$\mathrm{\ddot{o}}$newold-Moyal product~\citep{Gro46,Moyal49}.\\
\indent From these studies, both for the commutative case and for the non-commutative one no news have emerged as concerns the above-mentioned phenomenology of confinement. Anyway, in the latter case, they have predicted that the density of energy $\mathcal{E}$ will not have a domain of definition: in fact if the quantities \mbox{$(D_{\mu}\phi)^{\dagger}(D^{\mu}\phi)$} and \mbox{$(\phi^{\dagger}\phi-\phi^{2}_{0})^{2}$} are defined on $\mathcal{T}^{2}$, the quantity \mbox{$\frac{1}{4}F_{\mu\nu}F^{\mu\nu}$} is defined, instead, on $\tilde{\mathcal{T}}^{2}$ with periods scaled by a factor depending on $\theta$ in comparison with those of $\mathcal{T}^{2}$.
In addition, a rescaling of the coupling constant $g$ \mbox{by $\theta$ was expected.}\\ %
Moreover, having provided an explicit expression of $A_{\mu}$ and $\phi$ in terms of the coordinates in both contexts, it has favoured a direct control on the zeroes of the quantity $|\phi|^{2}$ (considered as the ``positions'' of the vortices). This fact, if on one side --- in the commutative case --- has led to get the metrics that describes the multivortex dynamics~\citep{Gonzalez07}, on the other side --- in the non-commutative case ---  it has shown how the non-commutativity of coordinates prevents $|\phi|^{2}$ from being null just in those points where in the commutative case, instead, such quantity is equal to zero.\\

In this paper, as we said in the Abstract, a new approach to non-commutativity will be introduced that will lead, inter alia, to redefine the role of $\theta$ in the \mbox{above-mentioned} \emph{dual superconductivity model} compared to~\citep{Forgacs2005,Lozano:2006xn}. \\
\indent Bearing this framework, this paper is organized as follows.\\ 
\indent In section~\ref{sec:twistc} we will recall~\citep{Baal82,Arroyo98} the typical equation (called \emph{consistency equation}) that characterizes the twist matrices $\Omega_{\mu}$ in the usual context in which the algebra of coordinates is commutative.\\
\indent In section~\ref{sec:twistnc}, it will be generalized by introducing the non-commutative coordinates $\hat{x}$ and $\hat{y}$ according to a new approach to the non-commutativity of space: the \emph{Weyl map} will be replaced by \emph{exponential mapping} typical of the studies of differential geometry~\citep{Helg78}, i.e.\,$(\hat{x},\hat{y})\equiv (x\bm{\hat{x}},y\bm{\hat{y}})$ with $x,y\in\R$, called \emph{normal coordinates}, and $\bm{\hat{x}},\bm{\hat{y}}$ base of vector space relative to $\mathcal{T}^{2}$. \\
It will lead to a systematic classification of the $\Omega_{\mu}$ and so of the fiber bundles \mbox{$\mathcal{T}^{2}\times SU(N)/Z_{N}$} for generic $N$ and for \mbox{$[\hat{\mu},\hat{\nu}]$} proportional or not to $\hat{\mathbbm{1}}$: in particular, this proportionality will depend on the homogeneity of $\mathcal{T}^{2}$.\\ 
The general definition of non-commutative torus $\mathcal{T}^{2}_{\nc}$ given in~\citep{Bondia2001} will be adapted to our context. In this regard, we will see how the difference between $\mathcal{T}^{2}_{\nc}$ (non-commutative torus) and $\mathcal{T}^{2}_{\cc}$ (commutative torus) lies in general in a ``widening'' of the set of twist matrices when we move from $\mathcal{T}^{2}_{\nc}$ to $\mathcal{T}^{2}_{\cc}$. \\
Three results will be the most remarkable: the first is that considering the $\mathcal{T}^{2}_{\cc}$ does not necessarily mean $[\hat{\mu},\hat{\nu}]=0$: we do not get $\mathcal{T}^{2}_{\cc}$ from $\mathcal{T}^{2}_{\nc}$ setting $\theta$ on zero as predicted in~\citep{Forgacs2005,Lozano:2006xn}; the second is that, for $\mathcal{T}^{2}_{\nc}$ and $\mathcal{T}^{2}_{\cc}$, in case $[\bm{\hat{x}},\bm{\hat{y}}]\propto\hat{\mathbbm{1}}$, i.e.  $[\bm{\hat{x}},\bm{\hat{y}}]=\iu\theta\hat{\mathbbm{1}}$, there is a different relation between \mbox{$\exp(\iu \theta\hat{\mathbbm{1}}I)$}, with $I$ the N-dimensional identity matrix, and $e$ (neutral element of $SU(N)$): while for $\mathcal{T}^{2}_{\cc}$ we have that \mbox{$\exp(\iu \theta\hat{\mathbbm{1}}I)=e$}, for $\mathcal{T}^{2}_{\nc}$ instead we have that \mbox{$\exp(\iu \theta\hat{\mathbbm{1}}I)\neq e$} (therefore we will have that $\mathcal{T}^{2}_{\cc}$ ($\mathcal{T}^{2}_{\nc}$) will be characterized by a specific non-commutativity parameter $\theta_{\cc}$ ($\theta_{\nc}$)). The third is that the $\Omega_{\mu}$'s proposed by~\citep{Forgacs2005,Lozano:2006xn} is not suitable for $\mathcal{T}^{2}_{\nc}$ (non-commutative torus) but only for $\mathcal{T}^{2}_{\cc}$: in fact they will result to belong to a class which is suitable to represent the twist matrices only after the above-mentioned ``widening''.\\ 
\indent In the end, in section~\ref{sec:lagrU1} we will go through the case $U(1)$ studying the lagrangian density~\eqref{eqn:dL} with $r_{\spt{\theta}}=\theta_{\cc}/\theta_{\nc}\in\Q\0$. Some analytical expressions will be shown for the fields $\phi$ and $A_{\mu}$ in terms of $\hat{x}$ and $\hat{y}$, with $[\bm{\hat{x}},\bm{\hat{y}}]=\iu\theta\hat{\mathbbm{1}}$, which are solutions of the TBC whose $\Omega_{\mu}$'s belong to a class found in section~\ref{sec:twistnc} compatible with the property of homogeneity of $\mathcal{T}^{2}_{\nc}$.
With such expressions for $\phi$ and $A_{\mu}$, we will show that $\mathcal{E}$ has $\mathcal{T}^{2}_{\nc}$ as a domain, differently from what happens in~\citep{Forgacs2005,Lozano:2006xn} in which, as we explained above, has not got any domain of definition. We will evaluate then $E$ rewriting it \textit{à la} Bogomolny~\citep{Bogomolny1976} for $\lambda=g^{2}/2$ (point of Bogomolny) and for minimum configurations without solving any system of equations (as it happens in~\citep{Manton2004,Arroyo04,Forgacs2005,Lozano:2006xn,EWeinberg2012}). In this regard, we will show how the II BPS eq is incompatible here with the toric geometry and we will find how  $E$ depends only on $r_{\spt{\theta}}$ and on the periods $(a_{1}, a_{2})$ of $\mathcal{T}^{2}_{\nc}$, so that \mbox{$E\rightarrow+\infty$} for \mbox{$a_{1},a_{2} \rightarrow+\infty$}, and not uniquely on $n$ (as it happens in refs.\,\citep{Baal82,Manton2004,Arroyo04,Forgacs2005,Lozano:2006xn,EWeinberg2012}).   
With $A_{\mu}$, we will show how no rescaling of $g$ of a factor depending on $\theta$ is predicted (as it happens in refs.\,\citep{Forgacs2005,Lozano:2006xn}) and we will calculate the flux $\mathcal{F}$ related to $q_{e}$ showing how it does not depend on $n$ --- as~\citep{Baal82,Manton2004,Arroyo04,Forgacs2005,Lozano:2006xn,EWeinberg2012} claim --- but uniquely on $r_{\spt{\theta}}$. We will have thus a fractionary $q_{e}$. We will show then how $r_{\spt{\theta}}$ determines the order of a cyclic $\tau_{r_{\spt{\theta}}}\subset U(1)$ generated by $\e^{\iu\pi\theta_{\nc}}$.\\
Finally, we will show how, instead of the vacuum expectation value $\Braket{\phi}$, $r_{\spt{\theta}}$ can be considered as an order parameter for the system phase transition. We will find thus a relation between $r_{\spt{\theta}}$ and the masses of the particle spectrum related to $\phi$ and $A_{\mu}$. In this regard, we will predict a double mechanism of Higgs according to which also the usual massless Goldstone boson acquires, in general, a non null mass depending on $r_{\spt{\theta}}$. Showing then how $r_{\spt{\theta}}$ can vary in time, we will conclude that such masses can vary in time in a quantized way. \\ 
\indent In the Appendix~\ref{sec:AppFintT} there is an explicit calculation of how starting from $\mathscr{F}$, i.e.\,a specific set of $\phi$, we can find some solutions of TBC related to $\phi$.\\
\indent In the Appendix~\ref{sec:AppErg} there is the proof of how a suitable class of fields comes to be minimum for $E$, with \mbox{$\lambda=g^{2}/2$}, rewritten \textit{à la} Bogomolny~\citep{Bogomolny1976}, without solving any system of equations.\\

\indent In the article, with the expression ``translations on $\mathcal{T}^{2}$'',
we refer to the elements of the group according to which $\mathcal{T}^{2}$ is homogeneous~\citep{Helg78}. Moreover, the differential calculus is set according to what is indicated in~\citep{Madore1995}.


\section{Abelian twist matrices: commutative case}
\label{sec:twistc}
\noindent The requirement of periodicity for the gauge invariant quantities of~\eqref{eqn:dL} provides the following definition\footnote{When a coordinate in the argument of a function is fixed to a specific value (for example 
$x_{\mu}=a_{\mu}$) the value of the other is understood to be arbitrary.} for TBC 
\begin{subequations}
\label{eqn:period}
\begin{align}
\phi(x_{\mu}=a_{\mu})&=\Omega_{\mu}\phi(x_{\mu}=0),\label{eqn:phi}\\
A_{\lambda}(x_{\mu}=a_{\mu})&=\Omega_{\mu}A_{\lambda}(x_{\mu}=0)\Omega^{-1}_{\mu}+\frac{i}{g}\Omega_{\mu}\partial_{\lambda}\Omega^{-1}_{\mu},\label{eqn:perigaug}
\end{align}
\end{subequations} 
with $0\leq x_{\mu}\leq a_{\mu}$, $\mu,\lambda=1,2$ ($a_{\mu}$'s indicate the periods of $\mathcal{T}^{2}$) and $\Omega_{\mu}$ with values in $SU(N)/Z_{N}$ (see Introduction). So, considering the two ways in which \mbox{$\phi(x_{\mu}=a_{\mu},x_{\nu}=a_{\nu})$} can be expressed in terms of \mbox{$\phi(x_{\mu}=0,x_{\nu}=0)$} through \mbox{$\phi(x_{\mu}=a_{\mu},x_{\nu}=0)$} and $\phi(x_{\mu}=0,x_{\nu}=a_{\nu})$, we get the equation that characterizes the $\Omega_{\mu}$'s:
\begin{equation}
\label{eqn:consist}
Z_{\mu\nu}=\Omega_{\mu}(x_{\nu}=a_{\nu})\Omega_{\nu}(x_{\mu}=0)\Omega^{-1}_{\mu}(x_{\nu}=0)\Omega^{-1}_{\nu}(x_{\mu}=a_{\mu})
\end{equation}
where $Z_{\mu\nu}$ is an element of the center $Z_{N}$ of $SU(N)$ of the form
\begin{equation}
\label{eqn:zmunug}
Z_{\mu\nu}=e^{\iu z_{\mu\nu}}, \quad \mathrm{with}\, z_{\mu\nu}=2\pi\frac{\eta_{\mu\nu}}{N}I,
\end{equation}
where $I$ is the N-dimensional identity matrix and $\eta_{\mu\nu}\in \Z(\mathrm{mod}\, N)$ \mbox{with $\eta_{\mu\nu}=-\eta_{\nu\mu}$.} This integer, independent of the coordinates, is also gauge invariant as it can be easily verified by noting that 
 $Z_{\mu\nu}$ is invariant under an arbitrary gauge transformation $\Omega$ of $\Omega_{\mu}$:
\begin{equation}
\label{eqn:Ogauge}
\Omega'_{\mu}=\Omega(x_{\mu}=a_{\mu})\Omega_{\mu}\Omega^{-1}(x_{\mu}=0).
\end{equation}
And also, considering that $\pi_{1}(SU(N)/Z_{N})\cong Z_{N}$, we can understand how $\eta_{\mu\nu}$ is an index that classifies the maps from a loop in $\mathcal{T}^{2}$ into a loop in $SU(N)/Z_{N}$. In particular, it specifies by which element of $Z_{N}$ a \emph{jump} occurs~\citep{Marmo2010} at $(a_{\mu},a_{\nu})$ in each class (element of the 
first homotopy group $\pi_{1}(SU(N)/Z_{N})$) of these maps.\\
A more accurate analysis then shows how these maps are defined by the transition functions of the principal bundle~\citep{Eguchi80} $\mathcal{T}^{2}\times SU(N)/Z_{N}$ and how the $\Omega_{\mu}$'s are but multiple transition functions~\citep{Baal82}.\\
\indent Then, by taking into account the \emph{consistency equation}~\eqref{eqn:consist}, it is proved~\citep{Baal82} that an arbitrary $\Omega_{\mu}$ assumes the following general expression
\begin{equation}
\label{eqn:prab}
\Omega_{\mu}=U(x_{\mu}=a_{\mu})\Omega^{(ab)}_{\mu}U^{-1}(x_{\mu}=0)
\end{equation}
where $\Omega^{(ab)}_{\mu}$ can be defined as 
\begin{equation}
\label{eqn:ab1} 
\Omega^{(ab)}_{\mu}=\exp\Bigl(\frac{\pi\iu}{N}\frac{\eta_{\mu\nu}x_{\nu}}{a_{\nu}}T_{1}\Bigr),
\end{equation}
$U$ is a gauge function on the edge of the \emph{core} and $T_{1}$ is a generator of $H=U(1)^{N-1}$, i.e.\,the maximal abelian (Cartan) subalgebra of $SU(N)$.  
We represent $T_{a}$, with $\tr(T_{a})=0$ for $a=1,2,\ldots, N-1$, as: $T_{a}=\mathrm{diag}(1,1, \ldots,1, -N+a,0,\ldots,0)$ where the first $N-a$ entries equal to $1$.\\  
However, if the purpose is to look for the minimum energy configurations satisfying the eqs.\,\eqref{eqn:period}, we will consider, as mentioned in the Introduction, only the \emph{abelian projection} $\Omega^{(ab)}_{\mu}$ of $\Omega_{\mu}$. 
References~\citep{Arroyo98,Arroyo04,Gonzalez07} can be seen in this regard where, in particular, \eqref{eqn:ab1} is adopted to represent the \emph{abelian projections} of $\Omega_{\mu}$'s. Anyway this representation of $\Omega^{(ab)}_{\mu}$'s could be adopted:
\begin{equation}
\label{eqn:abrev} 
\Omega^{(ab)}_{\mu}=\exp\Bigl(\frac{\pi\iu}{N}\eta_{\mu\nu}\Bigl(\frac{x_{\nu}}{a_{\nu}}+\frac{x_{\mu}}{a_{\mu}}\Bigr)T_{1}\Bigr).
\end{equation}
As matter of fact it is also a solution of~\eqref{eqn:consist} translated to any point $x\equiv(x_{1},x_{2})$ on $\mathcal{T}^{2}\colon$
\begin{equation}
\label{eqn:periodtr}
Z_{\mu\nu}=\Omega_{\mu}(a_{\nu}+x)\Omega_{\nu}(0+x)\Omega^{-1}_{\mu}(0+x)\Omega^{-1}_{\nu}(a_{\mu}+x).
\end{equation}
The difference, as can be noted, is that while the~\eqref{eqn:abrev} contemplates a dependence 
of $\Omega_{\mu}$ ($\Omega_{\nu}$) on translations along the axis $\mu$ (the axis $\nu$), the~\eqref{eqn:ab1} instead contemplates 
independence from them.\\
\indent However, an arbitrariness in the choice between~\eqref{eqn:abrev} and~\eqref{eqn:ab1} as a solution
of~\eqref{eqn:periodtr} is noted.
In this regard, in the next section, we will show how this arbitrariness  has its origin in a ``widening'' that presents the class of $\Omega_{\mu}$'s, which is suitable to represent the twist matrices for the non-commutative torus $\mathcal{T}^{2}_{\nc}$, when we move from $\mathcal{T}^{2}_{\nc}\rightarrow\mathcal{T}^{2}_{\cc}$ (commutative torus).


\section{Abelian twist matrices: non-commutative case}
\label{sec:twistnc}
\noindent We will present now a more general form for the equations of section~\ref{sec:twistc}.\\ 
To this regard, we should recall that, in general, with the expression $f(x)$ (with $x\in\mathbb{R}^{n}$
n-tuple of coordinated functions on a neighborhood $I_{O}$ of the origin $O$ of a manifold $M$) we actually indicate the 
composition $f\circ\psi^{-1}(x)$ for a generic homeomorphism $\psi$ on $I_{O}$ to a neighborhood of the origin of $\mathbb{R}^{n}$.\\ 
However we can use a different homeomorphism, that we call $\epsilon$, that maps the points of $I_{O}$ into a neighborhood $N_{O}$ of the origin $\hat{0}$ of the tangent vector space $M_{O}$ at $O$, whose elements (the \emph{tangent vectors}) are but derivations of the algebra of the functions on $I_{O}$~\citep{Helg78}. 
This is in complete analogy with the representation of a Lie group $G (\equiv M)$ that, via the \emph{exponential mapping}, can be put in correspondence with its algebra $\mathfrak{g}$ (for example cf.\,eq.\,\eqref{eqn:ab1} where, in particular, $G=H$ and the $T_{a}$ are a basis of $\mathfrak{g}$ of $H$).\\
\indent If therefore we consider $\mathcal{T}^{2}$ as $M$ and $I_{O}$ as a patch which covers the whole torus~\citep{Baal82,Arroyo98} --- such assumption is possible because $\mathcal{T}^{2}\equiv S^{1}\times S^{1}$ with $S^{1}$ is a unit circle --- we can use $\epsilon$ (the expression $f(x)$ is replaced by  $f(\hat{x})\equiv f\circ\epsilon^{-1}(\hat{x})$ with $\hat{x}\in N_{O}$) and thus the eq.\,\eqref{eqn:periodtr} admits the following rewriting 
\begin{equation}
\label{eqn:rev'}
Z_{\mu\nu}=\Omega_{\mu}(\hat{x}+\hat{a}_{\nu})\Omega_{\nu}(\hat{x}+\hat{0})\Omega^{-1}_{\mu}(\hat{x}+\hat{0})\Omega^{-1}_{\nu}(\hat{x}+\hat{a}_{\mu})
\end{equation} 
with $\hat{x}=(x_{\mu}\hat{\mu},x_{\nu}\hat{\nu})$, where $(x_{\mu},x_{\nu})$ indicate, in this case, the so-called \emph{normal coordinates}~\citep{Helg78} at $O$. 
As in section~\ref{sec:twistc}, they parametrize a ``translation'' to a generic point $p$ on $\mathcal{T}^{2}$ while $\hat{\mu}$ and $\hat{\nu}$ represent a basis in $N_{O}$. 
In this regard, we see that $\hat{\mu}$ and $\hat{\nu}$ can also be considered the two independent vectors that generate the lattice $\Lambda (\equiv N_{O})$ that ``fixes the quotient'' of the plan $\mathbb{R}^{2}$ in the definition of the torus: $\mathcal{T}^{2}=\mathbb{R}^{2}/\Lambda$.\\
The \emph{consistency equation} we have obtained, eq.\,\eqref{eqn:rev'}, is thus more general than~\eqref{eqn:periodtr}: 
\mbox{we can indeed consider non-commutative the product defined in $N_{O}$.}\pagebreak\\
\indent It's interesting to understand what novelties this choice brings to the study of fiber bundle $\mathcal{B}=\mathcal{T}^{2}\times SU(N)/Z_{N}$.\\
\indent Let's start with the analysis of the r.h.s.\,of eq.\,\eqref{eqn:rev'}. As, by construction, such r.h.s.\,\,is invariant under translations  (see section~\ref{sec:twistc}), then we begin to look for the conditions under which this can be represented invariant  for $\hat{x}$. In this regard, denoted by  $[[\Omega_{\mu},\Omega_{\nu}]]$ the quantity $\Omega_{\mu}(\hat{a}_{\nu})\Omega_{\nu}(\hat{0})\Omega^{-1}_{\mu}(\hat{0})\Omega^{-1}_{\nu}(\hat{a}_{\mu})$,
the following lemma is useful 
\begin{lem}
\label{lem:OeqOab}
Given an arbitrary $(\Omega_{\mu},\Omega_{\nu})$, 
\begin{equation*}
[[\Omega_{\mu},\Omega_{\nu}]] \,\,\mathit{is\,independent\,of}\,\, \hat{x} \Leftrightarrow [[\Omega^{(ab)}_{\mu},\Omega^{(ab)}_{\nu}]] \,\,\mathit{is\,independent\,of}\,\, \hat{x} 
\end{equation*}
with $(\Omega^{(ab)}_{\mu},\Omega^{(ab)}_{\nu})$ pair of abelian projections of $(\Omega_{\mu},\Omega_{\nu})$.
\end{lem}
\begin{proof}
It’s sufficient to note that being  $[[\Omega_{\mu},\Omega_{\nu}]]$ gauge invariant  (cf.\,eq.\,\eqref{eqn:Ogauge}), in particular, it is invariant under \emph{abelian projection} of  $\Omega_{\mu}$ (cf.\,eq.\,\eqref{eqn:prab}), \mbox{for which}
\begin{equation}
\label{eqn:OeqOab}
[[\Omega_{\mu},\Omega_{\nu}]] = [[\Omega^{(ab)}_{\mu},\Omega^{(ab)}_{\nu}]].
\end{equation}
\end{proof}
The focus shifts thus from $\mathcal{B}$ to the principal fiber bundle \mbox{$\mathcal{B}^{{\spt{H}}}=\mathcal{T}^{2}\times H$} and, in particular, to finding the conditions under which  $[[\Omega^{(ab)}_{\mu},\Omega^{(ab)}_{\nu}]]$ preserves the invariance under translations. In this regard, it is useful to redefine the set of the \emph{abelian projections} $\Omega^{(ab)}_{\mu}$'s $\mathscr{P}_{ab}=\{\Omega_{\mu}\colon \Omega_{\mu}=\exp\{\iu f^{a}_{\mu}(x)T_{a}\}$, with $T_{a}$ generators of $H$, in this way 
\begin{defask}
\label{def:Pab2}
$\mathscr{P}_{ab}=\{\Omega_{\mu}\colon [\Omega_{\mu}]_{ij}=\delta_{ij}e^{\iu f_{\mu,i}(\hat{x})}\}$, with $i,j=1, \ldots, N$,\\ \mbox{$\sum^{N}_{i=1}f_{\mu,i}(\hat{x})=\hat{0}$}, $f_{\mu,i}\colon N_{O}\rightarrow M_{O}$ and where $\delta$ is the Kronecker delta. 
\end{defask}
\noindent As we can see, the definitions are equivalent if we put $\sum^{N-1}_{a=1}[T_{a}]_{ii}f^{a}_{\mu}:=f_{\mu,i}$.\\
\indent Let's consider now the following subset of $\mathscr{P}_{ab}\times\mathscr{P}_{ab}$
\begin{defask}
\label{def:CinPxP}
$\mathscr{C}=\{(\Omega_{\mu},\Omega_{\nu})\colon (\Omega_{\mu},\Omega_{\nu})\in\mathscr{P}_{ab}\times\mathscr{P}_{ab}\,\,\mathrm{with}\,\, f_{\mu,i}\,\,\mathrm{and}\,\,f_{\nu,i}\,\,\mathrm{satisfying}\\
\mathrm{the\,\, equations}\,\, (\mathscr{C}_{1}), (\mathscr{C}_{2}), (\mathscr{C}_{3})\}$ below

\vspace{-1.15cm}

{\fontsize{8.9}{23}\selectfont
\begin{align}
\!\mathlarger{(\mathscr{C}_{1})} \,\,\,& f_{\mu,i}(\hat{x}+\hat{a}_{\nu})=\mathlarger{\sum}^{\infty}_{n=1}\sum^{n}_{k=1}\frac{(-1)^{k-1}}{k}\hspace{-.2cm}\sum_{\substack{
            (i_{1},j_{1}),\ldots,(i_{k},j_{k})\neq (0,0)\\
            i_{1}+j_{1}+\dots+i_{k}+j_{k}=n}}\frac{[f^{i_{1}}_{\mu,i}(\hat{a}_{\nu})f^{j_{1}}_{\mu,i}(\hat{x}) \cdots f^{i_{k}}_{\mu,i}(\hat{a}_{\nu})f^{j_{k}}_{\mu,i}(\hat{x})]}{n\cdot i_{1}!j_{1}!\cdots i_{k}!j_{k}!},\notag\\
\!\mathlarger{(\mathscr{C}_{2})} \,\,\,& f_{\nu,i}(\hat{x}+\hat{a}_{\mu})=\mathlarger{\sum}^{\infty}_{n=1}\sum^{n}_{k=1}\frac{(-1)^{k-1}}{k}\hspace{-.2cm}\sum_{\substack{
            (i_{1},j_{1}),\ldots,(i_{k},j_{k})\neq (0,0)\\
            i_{1}+j_{1}+\dots+i_{k}+j_{k}=n}}\frac{[f^{i_{1}}_{\nu,i}(\hat{a}_{\mu})f^{j_{1}}_{\nu,i}(\hat{x}) \cdots f^{i_{k}}_{\nu,i}(\hat{a}_{\mu})f^{j_{k}}_{\nu,i}(\hat{x})]}{n\cdot i_{1}!j_{1}!\cdots i_{k}!j_{k}!},\notag\\
\!\mathlarger{(\mathscr{C}_{3})} \,\,\,& [f_{\mu,i}(\hat{x}),f_{\nu,i}(\hat{x})]=0\notag
\end{align}}
\!\!where {\small{$[f^{i_{1}}_{\mu,i}(\hat{a}_{\nu})f^{j_{1}}_{\mu,i}(\hat{x}) \cdots f^{i_{k}}_{\mu,i}(\hat{a}_{\nu})f^{j_{k}}_{\mu,i}(\hat{x})]$}} indicates \mbox{the right nested commutator based}\\ on the
word {\small{$f^{i_{1}}_{\mu,i}(\hat{a}_{\nu})f^{j_{1}}_{\mu,i}(\hat{x}) \cdots f^{i_{k}}_{\mu,i}(\hat{a}_{\nu})f^{j_{k}}_{\mu,i}(\hat{x})$}}
(and similarly for {\small{$[f^{i_{1}}_{\nu,i}(\hat{a}_{\mu})f^{j_{1}}_{\nu,i}(\hat{x})\cdots f^{i_{k}}_{\nu,i}(\hat{a}_{\mu})f^{j_{k}}_{\nu,i}(\hat{x})]$}}). 
\end{defask}
\pagebreak
We find that 
\begin{lem}
\label{lem:defC}
\begin{equation*}
[[\Omega^{(ab)}_{\mu},\Omega^{(ab)}_{\nu}]] \,\,\mathit{is\,independent\,of}\,\, \hat{x}\Leftrightarrow \exists\; (\Omega^{(ab)}_{\mu},\Omega^{(ab)}_{\nu})\in\mathscr{C}.
\end{equation*} 
\end{lem}
\begin{proof}
It's sufficient to note that $[[\Omega^{(ab)}_{\mu},\Omega^{(ab)}_{\nu}]]$ is independent of $\hat{x}$ if{}f 
\begin{subequations}
\label{eqn:nnlinOab}
\begin{align}
&\Omega^{(ab)}_{\mu}(\hat{a}_{\nu}+\hat{x})=\Omega^{(ab)}_{\mu}(\hat{a}_{\nu})\Omega^{(ab)}_{\mu}(\hat{x}),\label{eqn:Omuab}\\
&\Omega^{(ab)}_{\nu}(\hat{a}_{\mu}+\hat{x})=\Omega^{(ab)}_{\nu}(\hat{a}_{\mu})\Omega^{(ab)}_{\nu}(\hat{x}),\label{eqn:Onuab}\\
&[\Omega^{(ab)}_{\mu}(\hat{x}),\Omega^{(ab)}_{\nu}(\hat{x})]=\hat{0}\label{eqn:cOmunuab}.
\end{align}
\end{subequations}
Representing now the $\Omega^{(ab)}_{\mu}$'s in terms of $f_{\mu,i}$ according to the Definition~\ref{def:Pab2} and applying the Campbell-Baker-Hausdorf{}f-Dynkin formula (C-B-H-D)~\citep{Dynkin47}, we obtain $(\mathscr{C}_{1}), (\mathscr{C}_{2})$ and $(\mathscr{C}_{3})$, respectively.
\end{proof}
We see, therefore, that to represent the r.h.s.\,\,of eq.\,\eqref{eqn:rev'} invariant under translations $\hat{x}$, it’s sufficient to find an element of $\mathscr{C}$.\\
We are thus interested in the search for solutions of  $(\mathscr{C}_{i})$'s. In this regard, we define $\mathscr{D}\subset\mathscr{P}_{ab}\times\mathscr{P}_{ab}$  
\begin{defask}
\label{def:D}
\mbox{$\mathscr{D}=\{(\Omega_{\mu},\Omega_{\nu})\colon (\Omega_{\mu},\Omega_{\nu})\in \mathscr{P}_{ab}\times\mathscr{P}_{ab}$} such that\\
$[\Omega_{\mu}(\hat{x},\hat{y})]_{ij}=\delta_{ij}e^{\iu\beta_{\mu,i}(x_{\nu})\hat{\nu}}e^{\iu\alpha_{\mu,i}(x_{\mu})\hat{\mu}}, [\Omega_{\nu}(\hat{x},\hat{y})]_{ij}=\delta_{ij}e^{\iu\alpha_{\nu,i}(x_{\mu})\hat{\mu}}e^{\iu\beta_{\nu,i}(x_{\nu})\hat{\nu}}\}$ with {\small{$i,j=1, \ldots, N$}}, {\small{$\alpha_{\mu,i},\beta_{\mu,i},\alpha_{\nu,i},\beta_{\nu,i}\colon\R\rightarrow\R$}}, 
\end{defask}
\hspace{-1.5em}and consider its subsets $\mathscr{D}_{\spt{\beta\alpha}}\,, \mathscr{D}_{\spt{\alpha\beta}}$ 
\begin{defdp}
\label{def:Dab}
$\mathscr{D}_{\spt{\beta\alpha}}=\{(\Omega_{\mu},\Omega_{\nu})\colon (\Omega_{\mu},\Omega_{\nu})\in \mathscr{D}\,\, \mathrm{with}\,\, \beta_{\mu,i},\alpha_{\nu,i}\neq 0\}$,\\ $\mathscr{D}_{\spt{\alpha\beta}}=\{(\Omega_{\mu},\Omega_{\nu})\colon (\Omega_{\mu},\Omega_{\nu})\in \nolinebreak[4]\mathscr{D}\,\, \mathrm{with}\,\, \alpha_{\mu,i},\beta_{\nu,i}\neq 0\}$.
\end{defdp}
Furthermore, we will continue the discussion for $N\geq{2}$. Moreover, we are implicitly under that hypothesis, since we consider non-zero the functions $\alpha_{\mu,i}$'s and $\beta_{\mu,i}$'s in $\mathscr{D}_{\spt{\beta\alpha}}$ and $\mathscr{D}_{\spt{\alpha\beta}}$.\footnote{\mbox{The case $N=1$ is characterized by null $\alpha_{\mu,i}$'s and $\beta_{\mu,i}$'s. We will have thus that $\mathscr{C}\cap\mathscr{D}=\mathscr{D}$.}}\\     
For instance, we look for which conditions $\mathscr{C}\cap\mathscr{D}_{\spt{\beta\alpha}}$ is different from the void set. 
But first let's define the following subset of $\mathscr{D}_{\spt{\beta\alpha}}$
\begin{defdp}
\label{def:D*ab}
$\mathscr{D}^{*}_{\spt{\beta\alpha}}=\{(\Omega_{\mu},\Omega_{\nu})\colon (\Omega_{\mu},\Omega_{\nu})\in \mathscr{D}_{\spt{\beta\alpha}}$ with \mbox{$\alpha_{\mu,i},\beta_{\mu,i},\alpha_{\nu,i},\beta_{\nu,i}$}\linebreak[4] linear
functions satisfying the equations $(\mathscr{D}^{*}_{{\spt{\beta\alpha}},1})\}$ below
\begin{equation}
\mathlarger{(\mathscr{D}^{*}_{{\spt{\beta\alpha}},1})} \,\,\, \alpha_{\mu,i}(x_{\mu})\beta_{\nu,i}(x_{\nu})-\beta_{\mu,i}(x_{\nu})\alpha_{\nu,i}(x_{\mu})=0,\,\forall i=1, \ldots, N.\notag
\end{equation}
\end{defdp}
\pagebreak
We find that\footnote{With \vl \,$\propto$\! \vr \, we mean ``proportional to''.}
\begin{lem}
\label{lem:D*ba}
$\mathscr{C}\cap\mathscr{D}_{\spt{\beta\alpha}}\neq\emptyset$ if{}f $\mathscr{D}^{*}_{\spt{\beta\alpha}}\neq\emptyset$ and on $\mathcal{T}^{2}$ the condition \emph{(}$\mathcal{B}^{\spt{H}}_{1}$\emph{)}\, $[\hat{\mu},\hat{\nu}]\propto\hat{\mathbbm{1}}$ holds,
where $\hat{\mathbbm{1}}\in M_{O}$ is such that $[\hat{a},\hat{\mathbbm{1}}]=0\;\; \forall\; \hat{a}\in N_{O}$.
\end{lem}
\begin{proof}
It's sufficient to apply the $(\mathscr{C}_{i})$'s to the representation of $f_{\mu,i}$'s given by the elements of $\mathscr{D}_{\spt{\beta\alpha}}$:
{\fontsize{9.5}{12.5}\selectfont
\begin{subequations}
\label{eqn:rapDf}
\begin{align}
\hspace{-1,5cm}f_{\mu,i}(\hat{x})&:=\mathlarger{\sum}^{\infty}_{n=1}\sum^{n}_{k=1}\frac{(-1)^{k-1}}{k}\hspace{-0,45cm}\sum_{\substack{
            (i_{1},j_{1}),\ldots,(i_{k},j_{k})\neq (0,0)\\
            i_{1}+j_{1}+\dots+i_{k}+j_{k}=n}}\hspace{-0,5cm}\frac{\iu^{n-1}[(\beta_{\mu,i}(x_{\nu})\hat{\nu})^{i_{1}}(\alpha_{\mu,i}(x_{\mu})\hat{\mu})^{j_{1}} \cdots (\beta_{\mu,i}(x_{\nu})\hat{\nu})^{i_{k}}(\alpha_{\mu,i}(x_{\mu})\hat{\mu})^{j_{k}}]}{n\cdot i_{1}!j_{1}!\cdots i_{k}!j_{k}!},\label{eqn:rapDf1}\\
\hspace{-1,5cm}f_{\nu,i}(\hat{x})&:=\mathlarger{\sum}^{\infty}_{n=1}\sum^{n}_{k=1}\frac{(-1)^{k-1}}{k}\hspace{-0,45cm}\sum_{\substack{
            (i_{1},j_{1}),\ldots,(i_{k},j_{k})\neq (0,0)\\
            i_{1}+j_{1}+\dots+i_{k}+j_{k}=n}}\hspace{-0,5cm}\frac{\iu^{n-1}[(\alpha_{\nu,i}(x_{\mu})\hat{\mu})^{i_{1}}(\beta_{\nu,i}(x_{\nu})\hat{\nu})^{j_{1}} \cdots (\alpha_{\nu,i}(x_{\mu})\hat{\mu})^{i_{k}}(\beta_{\nu,i}(x_{\nu})\hat{\nu})^{j_{k}}]}{n\cdot i_{1}!j_{1}!\cdots i_{k}!j_{k}!}\label{eqn:rapDf2}
\end{align}
\end{subequations}}
\!\!where (C-B-H-D) has been used. In particular, from $(\mathscr{C}_{1})$ and $(\mathscr{C}_{2})$, we obtain the request for linearity on $\alpha_{\mu,i}$'s and $\beta_{\mu,i}$'s. For example, since the comparison of both sides of $(\mathscr{C}_{1})$ to the first order, corresponding respectively to 
{\fontsize{10}{12.5}\selectfont
\begin{equation}
\label{eqn:fespl}
\begin{split}
\hspace{-1cm}f_{\mu,i}(\hat{x}+\hat{a}_{\nu})&= \alpha_{\mu,i}(x_{\mu})\hat{\mu}+\beta_{\mu,i}(x_{\nu}+a_{\nu})\hat{\nu} +\\
\hspace{-1cm}&\hspace{0,45cm}-\frac{\iu}{2} \alpha_{\mu,i}(x_{\mu})\beta_{\mu,i}(x_{\nu}+a_{\nu})[\hat{\mu},\hat{\nu}]+\\
\hspace{-1cm}&\hspace{0,45cm}-\frac{1}{12}\alpha^{2}_{\mu,i}(x_{\mu})\beta_{\mu,i}(x_{\nu}+a_{\nu})[\hat{\mu},[\hat{\mu},\hat{\nu}]]+\\
\hspace{-1cm}&\hspace{0,45cm}-\frac{1}{12}\alpha_{\mu,i}(x_{\mu})\beta^{2}_{\mu,i}(x_{\nu}+a_{\nu})[\hat{\nu},[\hat{\nu},\hat{\mu}]]+\ldots
\end{split}
\end{equation}}
\hspace{-.37em}and to
{\fontsize{8,2}{13}\selectfont
\begin{equation}
\label{eqn:fCBH}
\begin{split}
\hspace{-1,85cm}&\bigl(\alpha_{\mu,i}(0)+\alpha_{\mu,i}(x_{\mu})\bigr)\hat{\mu}+\bigl(\beta_{\mu,i}(x_{\nu})+\beta_{\mu,i}(a_{\nu})\bigr)\hat{\nu} +\\
\hspace{-1,85cm}&-\frac{\iu}{2} \bigl(\alpha_{\mu,i}(0)\beta_{\mu,i}(a_{\nu})+\alpha_{\mu,i}(x_{\mu})\beta_{\mu,i}(x_{\nu})+\alpha_{\mu,i}(x_{\mu})\beta_{\mu,i}(a_{\nu})-\alpha_{\mu,i}(0)\beta_{\mu,i}(x_{\nu})\bigr)[\hat{\mu},\hat{\nu}]+\\
\hspace{-1,85cm}&-\frac{1}{12}\Bigl(\alpha^{2}_{\mu,i}(0)\bigl(\beta_{\mu,i}(x_{\nu})+\beta_{\mu,i}(a_{\nu})\bigr)+\alpha^{2}_{\mu,i}(x_{\mu})\bigl(\beta_{\mu,i}(x_{\nu})+\beta_{\mu,i}(a_{\nu})\bigr)-2\alpha_{\mu,i}(0)\alpha_{\mu,i}(x_{\mu})\bigl(2\beta_{\mu,i}(x_{\nu})-\beta_{\mu,i}(a_{\nu})\bigr)\Bigr)[\hat{\mu},[\hat{\mu},\hat{\nu}]]+\\
\hspace{-1,85cm}&-\frac{1}{12}\Bigl(\alpha_{\mu,i}(0)\bigl(\beta^{2}_{\mu,i}(x_{\nu})+\beta^{2}_{\mu,i}(a_{\nu})-4\beta_{\mu,i}(x_{\nu})\beta_{\mu,i}(a_{\nu})\bigr)+\alpha_{\mu,i}(x_{\mu})\bigl(\beta^{2}_{\mu,i}(x_{\nu})+\beta^{2}_{\mu,i}(a_{\nu})+2\beta_{\mu,i}(x_{\nu})\beta_{\mu,i}(a_{\nu})\bigr)\Bigr)[\hat{\nu},[\hat{\nu},\hat{\mu}]]+\ldots , 
\end{split}
\end{equation}}
\!\!\!we notice that they are equal if{}f $\alpha_{\mu,i}$ and $\beta_{\mu,i}$ are linear functions in their respective arguments. The same consideration for $\alpha_{\nu,i}$ and $\beta_{\nu,i}$ whose request of linearity is obtained from $(\mathscr{C}_{2})$. \\
From $(\mathscr{C}_{3})$ we obtain the $(\mathscr{D}^{*}_{{\spt{\beta\alpha}},1})$ and the $(\mathcal{B}^{\spt{H}}_{1})$. As a matter of fact, considering
{\fontsize{8.5}{12.5}\selectfont
\begin{equation}
\label{eqn:[fmu,fnu]}
\begin{split}
\hspace{-0cm}[f_{\mu,i}(\hat{x}),f_{\nu,i}(\hat{x})]=&\bigl(\alpha_{\mu,i}(x_{\mu})\beta_{\nu,i}(x_{\nu})-\alpha_{\nu,i}(x_{\mu})\beta_{\mu,i}(x_{\nu})\bigr)[\hat{\mu},\hat{\nu}] +\\
\hspace{-0cm}&+\frac{\iu}{2}\alpha_{\mu,i}(x_{\mu})\alpha_{\nu,i}(x_{\mu})\bigl(\beta_{\mu,i}(x_{\nu})+\beta_{\nu,i}(x_{\nu})\bigr)[\hat{\mu},[\hat{\mu},\hat{\nu}]]+\\
\hspace{-0cm}&-\frac{\iu}{2}\beta_{\mu,i}(x_{\nu})\beta_{\nu,i}(x_{\nu})\bigl(\alpha_{\mu,i}(x_{\mu})+\alpha_{\nu,i}(x_{\mu})\bigr)[\hat{\nu},[\hat{\nu},\hat{\mu}]] +\\
\hspace{-0cm}&-\frac{\alpha_{\mu,i}(x_{\mu})\alpha^{2}_{\nu}(x_{\mu})\beta_{\nu}(x_{\nu})}{12}[\hat{\mu},[\hat{\mu},[\hat{\mu},\hat{\nu}]]]+\\
\hspace{-0cm}&+\frac{\alpha_{\nu,i}(x_{\mu})\beta_{\nu,i}(x_{\nu})}{12}\bigl(\alpha_{\mu,i}(x_{\mu})\beta_{\nu,i}(x_{\nu})-\alpha_{\nu,i}(x_{\mu})\beta_{\mu,i}(x_{\nu})\bigr)[\hat{\mu},[\hat{\nu},[\hat{\mu},\hat{\nu}]]] +\\
\hspace{-0cm}&-\frac{\alpha_{\nu,i}(x_{\mu})\beta_{\mu}(x_{\nu})\beta^{2}_{\nu}(x_{\nu})}{12}[\hat{\nu},[\hat{\nu},[\hat{\nu},\hat{\mu}]]] +\ldots ,
\end{split}
\end{equation}} 
\!\!from the linear independence of the various addends, we get that $[f_{\mu,i}(\hat{x}),f_{\nu,i}(\hat{x})]$ is null if{}f $(\mathscr{D}^{*}_{{\spt{\beta\alpha}},1})$ and $(\mathcal{B}^{\spt{H}}_{1})$ occur. 
\end{proof}

\begin{remark}
\label{rmk:5ord}
At first, we have tried to obtain $(\mathcal{B}^{\spt{H}}_{1})$ not from $(\mathscr{C}_{3})$ but testing the validity of $(\mathscr{C}_{1})$ in the various orders of~\eqref{eqn:rapDf1}. A validity of such equation was found up to the fifth order without getting thus any condition on $[\hat{\mu},\hat{\nu}]$!\footnote{In the calculation we have used the following identities:\,\scriptsize{$[\hat{\nu},[\hat{\mu},[\hat{\nu},\hat{\mu}]]]=[\hat{\mu},[\hat{\nu},[\hat{\nu},\hat{\mu}]]]$, $[[\hat{\mu},\hat{\nu}],[\hat{\mu},[\hat{\mu},\hat{\nu}]]]=[\hat{\mu},[\hat{\nu},[\hat{\mu},[\hat{\mu},\hat{\nu}]]]]-[\hat{\nu},[\hat{\mu},[\hat{\mu},[\hat{\mu},\hat{\nu}]]]]$ and $[[\hat{\mu},\hat{\nu}],[\hat{\nu},[\hat{\mu},\hat{\nu}]]]=-[\hat{\mu},[\hat{\nu},[\hat{\nu},[\hat{\nu},\hat{\mu}]]]]+[\hat{\nu},[\hat{\mu},[\hat{\nu},[\hat{\nu},\hat{\mu}]]]]$}.} 
\end{remark}
This leads to formulate the following conjecture that, in general, may happen to be interesting for the studies on (C-B-H-D):
\begin{conj}
\label{conj:nnlin}
Given \mbox{$\mathcal{T}^{2}$} with \mbox{$[\hat{\mu},\hat{\nu}]\neq 0$}, a class of solutions for the equations~\eqref{eqn:Omuab} and~\eqref{eqn:Onuab} is represented by the set $\mathscr{D}^{\ell}=\{(\Omega_{\mu},\Omega_{\nu})\colon (\Omega_{\mu},\Omega_{\nu})\in \mathscr{D}$\linebreak[4] with linear functions $\small{\alpha_{\mu,i},\beta_{\mu,i},\alpha_{\nu,i},\beta_{\nu,i}}\}$ with $\mathscr{D}$ as above.  
\end{conj} 
\noindent Such conjecture is easily proved only in some cases as, for example, for \mbox{$\mathscr{D}^{\ell}_{\spt{\alpha\alpha}}=\mathscr{D}^{\ell} \cap \mathscr{D}_{\spt{\alpha\alpha}}$} with
\begin{defdp}
\label{def:Daa}
$\mathscr{D}_{\spt{\alpha\alpha}}=\{(\Omega_{\mu},\Omega_{\nu})\colon (\Omega_{\mu},\Omega_{\nu})\in\mathscr{D}\,\, \mathrm{with}\,\, \alpha_{\mu,i}=\alpha_{\nu,i}=0\}$,
\end{defdp} 
\hspace{-.525cm}while for $\mathscr{D}^{\ell} \cap (\mathscr{D}_{\spt{\beta\alpha}}\cap \mathscr{D}_{\spt{\alpha\beta}})$ the proof, as we said, is less obvious.\\
\indent Before going on, we notice that for $\mathscr{C}\cap\mathscr{D}_{{\spt{\alpha\beta}}}$ we obtain the same conditions of $\mathscr{C}\cap\mathscr{D}_{\spt{\beta\alpha}}$. Moreover, for $N=2$ and $[\hat{\mu},\hat{\nu}]\propto\hat{\mathbbm{1}}$, the following sistem of equations\footnote{The eqs.\,\eqref{eqn:D*1D*2(2)}  and~\eqref{eqn:D*1D*2(3)} derive from the request that \mbox{$\sum^{N}_{i=1}f_{\mu,i}(\hat{x})=\sum^{N}_{i=1}f_{\nu,i}(\hat{x})=0$} (see Definition~\ref{def:Pab2}).}
\begin{subequations}
\label{eqn:D*1D*2}
\begin{align}
&\alpha_{\mu,i}(x_{\mu})\beta_{\nu,i}(x_{\nu})-\beta_{\mu,i}(x_{\nu})\alpha_{\nu,i}(x_{\mu})=0,\,\forall i=1, \ldots, N,\label{eqn:D*1D*2(1)}\\
&\sum^{N}_{i=1}\alpha_{\mu,i}(x_{\mu})\hat{\mu}+\beta_{\mu,i}(x_{\nu})\hat{\nu}-\frac{\iu}{2}\alpha_{\mu,i}(x_{\mu})\beta_{\mu,i}(x_{\nu})[\hat{\mu},\hat{\nu}]=0,\label{eqn:D*1D*2(2)}\\
&\sum^{N}_{i=1}\alpha_{\nu,i}(x_{\mu})\hat{\mu}+\beta_{\nu,i}(x_{\nu})\hat{\nu}+\frac{\iu}{2}\alpha_{\nu,i}(x_{\mu})\beta_{\nu,i}(x_{\nu})[\hat{\mu},\hat{\nu}]=0,\label{eqn:D*1D*2(3)}
\end{align}
\end{subequations}
does not admit solutions, differently from the case in which $N>2$ .\pagebreak\\
We can thus sum up that 
\begin{corlem}
\label{cor:dimH2}
$\mathscr{C}\cap\mathscr{D}_{\spt{\beta\alpha}}\neq\emptyset$ or $\mathscr{C}\cap\mathscr{D}_{{\spt{\alpha\beta}}}\neq\emptyset$ if{}f $[\hat{\mu},\hat{\nu}]\propto\hat{\mathbbm{1}}$ and $N>2$. 
\end{corlem}
Considering the Lemma~\ref{lem:D*ba}, we also find that
\begin{lem}
\label{lem:CcapD=D*ba}
Given $\mathcal{T}^{2}$ with $[\hat{\mu},\hat{\nu}]\propto\hat{\mathbbm{1}}$, we have
$\mathscr{C}\cap\mathscr{D}_{\spt{\beta\alpha}}=\mathscr{D}^{*}_{\spt{\beta\alpha}}$. 
\end{lem}
\begin{proof}
See the proof of the Lemma~\ref{lem:D*ba}.
\end{proof}
Now, in the case in which $[\hat{\mu},\hat{\nu}]$ is not proportional to $\hat{\mathbbm{1}}$ or $N=2$, we may consider 
$\mathscr{D}_{\spt{\alpha\alpha}}$.
Denoting by $\mathscr{D}^{\ell}_{\spt{\alpha\alpha}}=\mathscr{D}^{\ell}\cap\mathscr{D}_{\spt{\alpha\alpha}}$, we find that 
\begin{lem}
\label{lem:D*aa}
\!\!$\mathscr{C}\cap\mathscr{D}_{\spt{\alpha\alpha}}=\mathscr{D}^{\ell}_{\spt{\alpha\alpha}}$.
\end{lem}
\begin{proof}
We proceed similarly to the proof of the Lemma~\ref{lem:D*ba}. 
\end{proof}

We have thus found some representations of the twist matrices in different contexts: $[\hat{\mu},\hat{\nu}]$ proportional or not to $\hat{\mathbbm{1}}$ and for generic $N$, imposing as we have seen the invariance for $\hat{x}$ of r.h.s.\,of~\eqref{eqn:rev'}.\\
\indent Let's examine now such equation as a whole.
First let's define the set 
\begin{defdp}
\label{def:tD*}
$\mathscr{D}_{g}=\{(\Omega_{\mu},\Omega_{\nu})\colon (\Omega_{\mu},\Omega_{\nu})\in\mathscr{D}\,\, \mathrm{such\,\,that}\,\, g^{m+1}=\Omega_{\mu}(\hat{a}_{\nu}), \\g^{m}=\Omega_{\nu}(\hat{a}_{\mu}) \,\,\mathrm{for\,\,some}\,\, m\in \Z(\mathrm{mod}\, N)\}$ where $g$ is the generator of the center $Z_{N}$.
\end{defdp}
We find the following 
\begin{thm} 
\label{thm:rev3}
Given $\mathcal{T}^{2}$ with $[\hat{\mu},\hat{\nu}]=\iu\theta\hat{\mathbbm{1}}$, if we represent $[[\Omega_{\mu},\Omega_{\nu}]]$ according to \mbox{$(\Omega^{(ab)}_{\mu},\Omega^{(ab)}_{\nu})\in\mathscr{D}^{*}_{\spt{\beta\alpha}}$}, we have
{\fontsize{11}{18}\selectfont
\[
Z_{\mu\nu}=[[\Omega_{\mu},\Omega_{\nu}]] \Leftrightarrow \begin{dcases*}
                                                                             \hspace{-0,35cm}&{\small{\emph{(}$\mathscr{D}_{g,1}$\emph{)}}}\qquad$(\Omega^{(ab)}_{\mu},\Omega^{(ab)}_{\nu})\in\mathscr{D}_{g}$\\
\hspace{-0,35cm}&{\small{\emph{(}$\mathscr{D}_{g,2}$\emph{)}}}\qquad$\delta_{ij}\exp(\iu\beta_{\mu,i}(a_{\nu})\alpha_{\nu,i}(a_{\mu})\theta\hat{\mathbbm{1}})=[e]_{ij}$
                                                                       \end{dcases*}
\]}
\!\!with $e$ neutral element of $SU(N)$.
\end{thm} 
\begin{proof}
Considering that $Z_{\mu\nu}\equiv g^{n}$, and setting $\Omega_{\mu}\equiv\tilde{\Omega}^{n}_{\mu}$, $\Omega_{\nu}\equiv\tilde{\Omega}^{n}_{\nu}$, we have that the~\eqref{eqn:rev'} is satisfied if{}f 
\begin{equation}
\label{eqn:indn}
g^{n}=[[\tilde{\Omega}^{n}_{\mu},\tilde{\Omega}^{n}_{\nu}]]\quad\forall\, n\in \Z(\mathrm{mod}\, N)\, .
\end{equation}
The~\eqref{eqn:indn} is demonstrated by induction. For $n=1$ is true if{}f, considering the~\eqref{eqn:OeqOab}, $g=[[\tilde{\Omega}^{(ab)}_{\mu},\tilde{\Omega}^{(ab)}_{\nu}]]$ and it holds, taking into account that by hypothesis the \mbox{$(\tilde{\Omega}^{(ab)}_{\mu},\tilde{\Omega}^{(ab)}_{\nu})\in\mathscr{D}^{*}_{\spt{\beta\alpha}}$}, if{}f the {\small{($\mathscr{D}_{g,i}$)}}'s hold.\,\,In particular the {\small{($\mathscr{D}_{g,2}$)}} ensures that $Z_{N}$ is represented as a group, i.\,e.\,for example that $gg=g^{2}$. \\
Now, if it is supposed true for $n$, it is then proved easily true for $n+1$, too. 
\end{proof}
On the other hand, if we consider $\mathscr{C}\cap\mathscr{D}_{\spt{\beta\alpha}}=\emptyset$, i.e.\,in the case of a fiber bundle $\mathcal{B}=\mathcal{T}^{2}\times SU(N)/Z_{N}$ with $[\hat{\mu},\hat{\nu}]$ is not proportional to $\hat{\mathbbm{1}}$ or $N=2$ (see Corollary~\ref{cor:dimH2}), we have that
\begin{thm} 
\label{thm:revnn1}
If $\mathscr{C}\cap\mathscr{D}_{\spt{\beta\alpha}}=\emptyset$ and $[[\Omega_{\mu},\Omega_{\nu}]]$ is represented according to \mbox{$(\Omega^{(ab)}_{\mu},\Omega^{(ab)}_{\nu})\in\mathscr{D}^{\ell}_{\spt{\alpha\alpha}}$},
\[
Z_{\mu\nu}=[[\Omega_{\mu},\Omega_{\nu}]] \Leftrightarrow 
                                                                          (\Omega^{(ab)}_{\mu},\Omega^{(ab)}_{\nu})\in\mathscr{D}_{g}
\]
\end{thm} 
\begin{proof}
The proof is similar to the one of Theorem~\ref{thm:rev3} except that
here, being by hypothesis that $\alpha_{\nu}=0$, no condition is present on $[\hat{\mu},\hat{\nu}]$.
\end{proof}
Let’s now continue the discussion, focusing on~\eqref{eqn:indn}.\\  
This equation, showing that $Z_{\mu\nu}$ is generally different from the neutral element $e$ of $SU(N)/Z_{N}$, suggests to interpret it as
\begin{equation}
\label{eqn:defz'}
\phi(\hat{a}_{\mu},\hat{a}_{\nu})=Z_{\mu\nu}(\hat{a}_{\mu},\hat{a}_{\nu})\phi(\hat{a}_{\nu},\hat{a}_{\mu}).
\end{equation}
As matter of fact, the origin $O$, in the case of $\mathcal{T}^{2}$, can be ``labelled'', through $\epsilon$, with $(\hat{0},\hat{0})$ and with the ordered pair $(\hat{a}_{\mu},\hat{a}_{\nu})$ or $(\hat{a}_{\nu},\hat{a}_{\mu})\in N_{\hat{0}}\times N_{\hat{0}}$.  
This leads to reconstruct the r.h.s.\,\,of~\eqref{eqn:rev'} in the following way
\begin{equation}
\label{eqn:reimpcon}
\!\!\phi(\hat{a}_{\mu},\hat{a}_{\nu})=\Omega_{\mu}(\hat{a}_{\nu})Z_{\mu\nu}(0,\hat{a}_{\nu})\Omega_{\nu}(\hat{0})\Omega^{-1}_{\mu}(\hat{0})Z_{\mu\nu}(\hat{a}_{\mu},\hat{0})\Omega^{-1}_{\nu}(\hat{a}_{\mu})Z^{-1}_{\mu\nu}(\hat{a}_{\mu},\hat{a}_{\nu})\phi(\hat{a}_{\mu},\hat{a}_{\nu}),
\end{equation}
and so to have  
\begin{equation}
\label{eqn:revgen}
Z_{\mu\nu}(\hat{a}_{\mu},\hat{a}_{\nu})=\Omega_{\mu}(\hat{a}_{\nu})Z_{\mu\nu}(0,\hat{a}_{\nu})\Omega_{\nu}(\hat{0})\Omega^{-1}_{\mu}(\hat{0})Z_{\mu\nu}(\hat{a}_{\mu},\hat{0})\Omega^{-1}_{\nu}(\hat{a}_{\mu}),
\end{equation}
\hspace{-.05em}from which we find again the~\eqref{eqn:indn} considering that $Z_{\mu\nu}\in Z_{N}$ and that therefore 
\begin{equation} Z_{\mu\nu}(\hat{a}_{\mu},\hat{a}_{\nu})Z^{-1}_{\mu\nu}(0,\hat{a}_{\nu})Z^{-1}_{\mu\nu}(\hat{a}_{\mu},\hat{0})\equiv g^{n}. 
\end{equation}
\hspace{-.05em}From~\eqref{eqn:reimpcon}, the following \emph{twisted boundary conditions} derive
\begin{subequations}
\label{eqn:bordH}
\begin{align}
\phi(\hat{x}_{\mu}+\hat{a}_{\mu},\hat{x}_{\nu})&=(\tilde{\Omega}^{(ab)}_{\mu})^{n}(\hat{x}_{\mu},\hat{x}_{\nu})\phi(\hat{x}_{\mu},\hat{x}_{\nu}),\label{eqn:bordH1}\\
\phi(\hat{x}_{\nu}+\hat{a}_{\nu},\hat{x}_{\mu})&=(\tilde{\Omega}^{(ab)}_{\nu})^{n}(\hat{x}_{\nu},\hat{x}_{\mu})\phi(\hat{x}_{\nu},\hat{x}_{\mu})\label{eqn:bordH2}
\end{align}
\end{subequations}
and for the gauge potential
{\fontsize{9.8}{12}
 \selectfont
\begin{subequations}
\label{eqn:bordUA}
\begin{align}
\hspace{-2cm}A_{\lambda}(\hat{x}_{\mu}+\hat{a}_{\mu},\hat{x}_{\nu})&=(\tilde{\Omega}^{(ab)}_{\mu})^{n}(\hat{x}_{\mu},\hat{x}_{\nu})A_{\lambda}(\hat{x}_{\mu},\hat{x}_{\nu})(\tilde{\Omega}^{(ab)}_{\mu})^{-n}(\hat{x}_{\mu},\hat{x}_{\nu})+\frac{\iu}{g}(\tilde{\Omega}^{(ab)}_{\mu})^{n}(\hat{x}_{\mu},\hat{x}_{\nu})\hat{\partial}_{\lambda}(\tilde{\Omega}^{(ab)}_{\mu})^{-n}(\hat{x}_{\mu},\hat{x}_{\nu}),\label{eqn:bordU1A}\\
\hspace{-2cm}A_{\lambda}(\hat{x}_{\nu}+\hat{a}_{\nu},\hat{x}_{\mu})&=(\tilde{\Omega}^{(ab)}_{\nu})^{n}(\hat{x}_{\nu},\hat{x}_{\mu})A_{\lambda}(\hat{x}_{\nu},\hat{x}_{\mu})(\tilde{\Omega}^{(ab)}_{\nu})^{-n}(\hat{x}_{\nu},\hat{x}_{\mu})+\frac{\iu}{g}(\tilde{\Omega}^{(ab)}_{\nu})^{n}(\hat{x}_{\nu},\hat{x}_{\mu})\hat{\partial}_{\lambda}(\tilde{\Omega}^{(ab)}_{\nu})^{-n}(\hat{x}_{\nu},\hat{x}_{\mu})\label{eqn:bordU2A}
\end{align}
\end{subequations}} 
\hspace{-.35em}with $n\in \Z(\mathrm{mod}\, N)$, $(\hat{x}_{\mu},\hat{x}_{\nu})=(x_{\mu}\hat{\mu},x_{\nu}\hat{\nu})$ and $\hat{\partial}_{\lambda}$ defined as 
\begin{equation}
\label{eqn:partial}
\hat{\partial}_{\lambda}:=\frac{\iu}{a_{\varrho}}\varepsilon_{\lambda\varrho}[\bm{\hat{x}}_{\varrho},\,\,\,\,]\quad\quad\lambda, \varrho=1,2\,\, ,
\end{equation}
where $\varepsilon_{\lambda\varrho}$ is an antisymmetric tensor with $\varepsilon_{12}=1$, $\bm{\hat{x}}_{1}=\hat{\mu}$ and $\bm{\hat{x}}_{2}=\hat{\nu}$.\pagebreak\\
 
Let's conclude now the section illustrating a theorem that proves how $\mathcal{T}^{2}$ can be considered commutative despite $[\hat{\mu},\hat{\nu}]\neq 0$.\\
First we will give the following\footnote{The definition of non-commutative torus is similar to the one given in~\citep{Bondia2001}: the difference is that attention here is focused more on the twist matrices than on the function algebra that derives from them.}
\begin{defask}
The torus $\mathcal{T}^{2}$ is said \emph{commutative}, and is indicated with $\mathcal{T}^{2}_{\cc}$, if{}f
\begin{equation}
\label{eqn:defTc}
\Omega^{(ab)}_{\mu}\Omega^{(ab)}_{\nu}(\Omega^{(ab)}_{\nu}\Omega^{(ab)}_{\mu})^{-1}=e,
\end{equation}
where $e$ is the neutral element of $SU(N)$, otherwise it is said \emph{non-commutative} and it is indicated with $\mathcal{T}^{2}_{\nc}$.
\end{defask}
We note thus that for $\mathcal{T}^{2}_{\cc}$, to represent $[[\Omega_{\mu},\Omega_{\nu}]]$ invariant for $\hat{x}$, the condition~\eqref{eqn:cOmunuab}, and so $(\mathscr{C}_{3})$, does not exist. 
In other words, renaming $\mathscr{C}$ as $\mathscr{C}_{\spt{\mathcal{T}^{2}_{{\nc}}}}$, we will have that  
\begin{equation}
\label{eq:allarg}
\mathscr{C}_{\spt{\mathcal{T}^{2}_{{\nc}}}}\subset \mathscr{C}_{\spt{\mathcal{T}^{2}_{{\cc}}}} \,\,\, ,
\end{equation}
where $\mathscr{C}_{\spt{\mathcal{T}^{2}_{{\cc}}}}=\{(\Omega_{\mu},\Omega_{\nu})\colon (\Omega_{\mu},\Omega_{\nu})\in\mathscr{P}_{ab}\,\,\mathrm{with}\,\, f_{\mu,i}\,\,\mathrm{and}\,\,f_{\nu,i}$ satisfying the equations $(\mathscr{C}_{1}), (\mathscr{C}_{2})\}$.\\
As concerns the research for an explicit realization of an element of $\mathscr{C}_{\spt{\mathcal{T}^{2}_{{\cc}}}}$ we can state the following 
\begin{conj}
\label{conj:T2cdell}
Given \mbox{$\mathcal{T}^{2}_{\cc}$} with \mbox{$[\hat{\mu},\hat{\nu}]\neq 0$}, 
\begin{equation}
\mathscr{C}_{\spt{\mathcal{T}^{2}_{{\cc}}}}\cap\mathscr{D}=\mathscr{D}^{\ell} 
\end{equation}
with $\mathscr{D}$ and $\mathscr{D}^{\ell}$ as above.
\end{conj} 
\begin{remark}
\label{rmk:wide}
Such a conjecture is easily proved for the subsets of $\mathscr{D}$ as $\mathscr{D}_{\spt{\alpha\alpha}}$ while for the ones like $\mathscr{D}_{\spt{\beta\alpha}}\cap \mathscr{D}_{\spt{\alpha\beta}}$ with $[\hat{\mu},\hat{\nu}]$ not proportional to  $\hat{\mathbbm{1}}$, the proof is less obvious (see Remark~\ref{rmk:5ord}).\\
Moreover, as it was predictable (cf.\,eq.\,\eqref{eq:allarg}), in the case in which $[\hat{\mu},\hat{\nu}]\propto\hat{\mathbbm{1}}$, the subset of $\mathscr{D}_{\spt{\beta\alpha}}$ that represents $[[\Omega^{(ab)}_{\mu},\Omega^{(ab)}_{\nu}]]$ invariant for $\hat{x}$ has ``widened'' to include, for example, also the set\footnote{With $\sim$ we denote the set-theoretic operation of complementation.} $\tilde{\mathscr{D}}^{\ell}_{\spt{\beta\alpha}}\subset \mathscr{D}^{\ell}_{\spt{\beta\alpha}}\!\sim\!\mathscr{D}^{*}_{\spt{\beta\alpha}}$, with $\mathscr{D}^{\ell}_{\spt{\beta\alpha}}=\mathscr{D}^{\ell}\cap\mathscr{D}_{\spt{\beta\alpha}}$ and $\tilde{\mathscr{D}}^{\ell}_{\spt{\beta\alpha}}=\mathscr{D}^{\ell}\cap\tilde{\mathscr{D}}_{\spt{\beta\alpha}}$ where
\mbox{$\tilde{\mathscr{D}}_{\spt{\beta\alpha}}=\{(\Omega_{\mu},\Omega_{\nu})\colon (\Omega_{\mu},\Omega_{\nu})\in\mathscr{D}_{\spt{\beta\alpha}}\,\, \mathrm{with}\,\, \alpha_{\mu,i},\beta_{\nu,i}= 0\}$}, that before was excluded (see figure~\ref{fig:VennW2}). From this point of view, we can thus reinterpret the arbitrariness we found at the end of section~\ref{sec:twistc} in the choice between~\eqref{eqn:ab1} and~\eqref{eqn:abrev} as a consequence of this ``widening''  due to a transition $\mathcal{T}^{2}_{\nc}\rightarrow\mathcal{T}^{2}_{\cc}$ and ultimately, as we will see later, due to a difference in the value of $\theta$ between $\mathcal{T}^{2}_{\nc}$ and $\mathcal{T}^{2}_{\cc}$.
\end{remark}

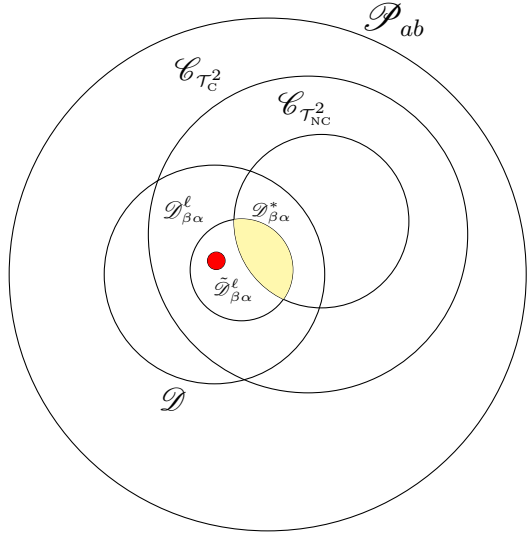
\begin{figure}
\centering
\begin{tikzpicture}[scale=0.5]
    \draw \firstcircle node at (3.3,6.75) {\fontsize{13.5}{12}\selectfont{$\mathscr{P}_{ab}$}};
    \draw \secondcircle node at (-2.5,-3.3)  {\fontsize{12}{12}\selectfont{$\mathscr{D}$}};
    \draw \thirdcircle node at (1,4.4) {\fontsize{12}{12}\selectfont{$\mathscr{C}_{\spt{\mathcal{T}^{2}_{{\nc}}}}$}};
    \draw \fourthcircle node at (-2.2,1.7) {\fontsize{9}{12}\selectfont{$\mathscr{D}^{\ell}_{\spt{\beta\alpha}}$}};
     \draw \fifthcircle node at (-0.9,-.43){\fontsize{6.8}{12}\selectfont{$\tilde{\mathscr{D}}^{\ell}_{\spt{\beta\alpha}}$}};
     \draw \sixthcircle node at (-1.8,5.3){\fontsize{13}{12}\selectfont{$\mathscr{C}_{\spt{\mathcal{T}^{2}_{\cc}}}$}};
       \node[anchor=south] at (.1,1.1) {\fontsize{7}{12}\selectfont{$\mathscr{D}^{*}_{\spt{\beta\alpha}}$}};
        \begin{scope}
      \clip \thirdcircle;
      \fill[yellow!40] \fourthcircle;
     \end{scope}
     \begin{scope}
     \fill[red] \fifthcircle;
    \end{scope}

\end{tikzpicture}
\caption{``Widening'' of $\mathscr{C}_{\spt{\mathcal{T}^{2}_{{\nc}}}}$ in $\mathscr{P}_{ab}$ after the transition $\mathcal{T}^{2}_{\nc}\rightarrow\mathcal{T}^{2}_{\cc}$ in the case \mbox{$[\hat{\mu},\hat{\nu}]\propto\hat{\mathbbm{1}}$}. The $(\Omega_{\mu},\Omega_{\nu})$ proposed in~\citep{Forgacs2005,Lozano:2006xn}, generalizations in a non-commutative context of~\eqref{eqn:ab1}, belong to $\tilde{\mathscr{D}}^{\ell}_{\spt{\beta\alpha}}$: they are not so suitable to represent the twist matrices for $\mathcal{T}^{2}_{\nc}$ but only for $\mathcal{T}^{2}_{\cc}$.}
\label{fig:VennW2}
\end{figure}
\pagebreak 
We explain now the above-mentioned theorem with which we opened this discussion about $\mathcal{T}^{2}_{\cc}$. We find that
\begin{thm} 
\label{teoth}
Given \mbox{$\mathcal{T}^{2}$} with \mbox{$[\hat{\mu},\hat{\nu}]=\iu\theta\hat{\mathbbm{1}}$}, where $\theta\in\R$, $\mathcal{T}^{2}$ is commutative if{}f
\begin{equation}
\label{eqn:th1e}
\exp(\iu \theta\hat{\mathbbm{1}}I)=e
\end{equation}
where $I$ is the N-dimensional identity matrix.
\end{thm}
\begin{proof}
The quantity {\small{$\Omega^{(ab)}_{\mu}\Omega^{(ab)}_{\nu}(\Omega^{(ab)}_{\nu}\Omega^{(ab)}_{\mu})^{-1}$}}, calculated for example according to $(\Omega_{\mu},\Omega_{\nu})$ in $\tilde{\mathscr{D}}^{\ell}_{\spt{\beta\alpha}}$ as above, is equal to:
\begin{equation}
\label{eqn:thDdba}
\bigl(\exp(\iu \theta\hat{\mathbbm{1}})\bigr)^{\alpha_{\nu,i}(x)\beta_{\mu,i}(y)}\quad \mathrm{for\,\,each}\,\,i=1, \ldots, N\,.\notag
\end{equation} 
As we can see, such quantity is equal to $[e]_{ii}$ if{}f~\eqref{eqn:th1e} holds.
\end{proof}
From this theorem it follows that  
\begin{corth}
\label{cor:thkerom}
In the hypotheses above,
\begin{center}
$\mathcal{T}^{2}$ is commutative if{}f  $\theta\in \ker(\pi_{\spt{C}})$ 
\end{center}
with $\pi_{\spt{C}}\in \mathscr{H}=\{\pi_{\spt{C}}\colon \pi_{\spt{C}}(t):=\exp(\iu t\hat{\mathbbm{1}}I)$ homomorphism where $t\in\R/C$ and\linebreak[4] $\pi_{\spt{C}}(t)\in H$ with $C\subset\R$, such that $\pi_{\spt{C}}(c)=e\, ,\,\,\forall c\in C\} $ .
\end{corth}
\begin{proof}
Setting $e^{\iu \theta\hat{\mathbbm{1}}}=e^{\iu\pi\theta}$ and considering that $e^{\iu\pi\theta}=\cos(\pi\theta)+\iu \sin(\pi\theta)$, from~\eqref{eqn:th1e} we will have that $\theta\in\ker(\pi_{\spt{[0]}})$ with $[0]\in\R/2\Z$.
\end{proof}
\noindent We can see explicitly how $\mathcal{T}^{2}_{\cc}$ is not necessarily characterized by a null value of $\theta$ !\pagebreak \\

As concerns the~\eqref{eqn:rev'} we can say that 
\begin{thm} 
\label{thm:rev3Tc}
Given \mbox{$\mathcal{T}^{2}_{\cc}$}, if we represent $[[\Omega_{\mu},\Omega_{\nu}]]$ according to \mbox{$(\Omega^{(ab)}_{\mu},\Omega^{(ab)}_{\nu})\in\mathscr{D}^{\ell}_{\beta\alpha}$},
{\fontsize{11}{18}\selectfont
\[
Z_{\mu\nu}=[[\Omega_{\mu},\Omega_{\nu}]] \Leftrightarrow \begin{dcases*}
                                                                             \hspace{-0,35cm}&{\small{\emph{(}$\mathscr{D}_{g,\cc_{1}}$\emph{)}}}\qquad$(\Omega^{(ab)}_{\mu},\Omega^{(ab)}_{\nu})\in\mathscr{D}_{g}$\\
\hspace{-0,35cm}&{\small{\emph{(}$\mathscr{D}_{g,\cc_{2}}$\emph{)}}}\qquad$[\hat{\mu},\hat{\nu}]\propto\hat{\mathbbm{1}}$
                                                                       \end{dcases*}
\]}
\end{thm} 
\begin{proof}
We proceed in the same way to the demonstration of the Theorem~\ref{thm:rev3} with the dif{}ference that here $gg=g^{2}$ if{}f 
\begin{equation}
\label{eqn:eTCDab}
\hspace{-1,2cm}[e]_{ii}=\exp\Bigl(\beta_{\mu,i}(a_{\nu})\alpha_{\nu,i}(a_{\mu})[\hat{\mu},\hat{\nu}]-\frac{\iu}{2}\beta_{\mu,i}(a_{\nu})\alpha_{\nu,i}(a_{\mu,i})\bigl(\alpha_{\nu,i}(a_{\mu})[\hat{\mu},[\hat{\mu},\hat{\nu}]]+\beta_{\mu,i}(a_{\nu})[\hat{\nu},[\hat{\nu},\hat{\mu}]]\bigr)+\ldots\Bigr)
\end{equation}
that is, considering the linear independence of the various terms in the exponent of l.h.s.\,of~\eqref{eqn:eTCDab} and the Theorem~\ref{teoth}, if{}f the {\small{($\mathscr{D}_{g,\cc_{2}}$)}} holds.
\end{proof}

On the other hand, 
\begin{thm} 
\label{thm:revn1Tc}
Given \mbox{$\mathcal{T}^{2}_{\cc}$} with $[\hat{\mu},\hat{\nu}]$ not proportional to $\hat{\mathbbm{1}}$, if $[[\Omega_{\mu},\Omega_{\nu}]]$ is represented according to $(\Omega^{(ab)}_{\mu},\Omega^{(ab)}_{\nu})\in\mathscr{D}^{\ell}_{\spt{\alpha\alpha}}$,
\[
Z_{\mu\nu}=[[\Omega_{\mu},\Omega_{\nu}]] \Leftrightarrow 
                                                                          (\Omega^{(ab)}_{\mu},\Omega^{(ab)}_{\nu})\in\mathscr{D}_{g} \,\, .
\]
\end{thm} 
\begin{proof}
The demonstration is similar to the one of the Theorem~\ref{thm:revnn1}. 
\end{proof}

This approach to the non-commutativity has allowed, we might say, a more systematic classification of the principal fiber bundles \mbox{$\mathcal{B}=\mathcal{T}^{2}\times SU(N)/Z_{N}\colon$} imposing the invariance under translations $\hat{x}$ of the r.h.s.\,\,of the~\eqref{eqn:rev'} (\emph{consistency equation}), both for $\mathcal{T}^{2}_{\nc}$ and for $\mathcal{T}^{2}_{\cc}$, a representation of the twist matrices has been given for generic $N$ and for $[\hat{\mu},\hat{\nu}]$ proportional or not to $\hat{\mathbbm{1}}$.\\
\indent The most remarkable result is that in case $[\hat{\mu},\hat{\nu}]\propto\hat{\mathbbm{1}}$, the~\eqref{eqn:ab1}, resumed by~\citep{Forgacs2005,Lozano:2006xn} in the non-commutative context, is not suitable to represent the twist matrices for $\mathcal{T}^{2}_{\nc}$: it belongs to $\mathscr{D}^{\ell}_{\spt{\beta\alpha}}\!\sim\!\mathscr{D}^{*}_{\spt{\beta\alpha}}$ (in fact, though it is an element of $\mathscr{D}^{\ell}_{\spt{\beta\alpha}}$, it does not satisfy the ($\mathscr{D}^{*}_{{\spt{\beta\alpha}},1}$)).
As we have seen, $\mathscr{D}^{*}_{\spt{\beta\alpha}}$ is the set of $(\Omega^{(ab)}_{\mu},\Omega^{(ab)}_{\nu})$ (generalizations of~\eqref{eqn:abrev}) whose elements are appropriate to represent the twist matrices for $\mathcal{T}^{2}_{\nc}\colon$ for $N>2$ and $[\hat{\mu},\hat{\nu}]\propto\hat{\mathbbm{1}}$, they represent {\small{$[[\Omega_{\mu},\Omega_{\nu}]]$}} invariant for $\hat{x}$ (see Corollary~\ref{cor:dimH2} and Lemma~\ref{lem:CcapD=D*ba}). The~\eqref{eqn:ab1} has resulted suitable for $\mathcal{T}^{2}_{\cc}$ (it belongs in fact to $\tilde{\mathscr{D}}^{\ell}_{\spt{\beta\alpha}}$, see Remark~\ref{rmk:wide} and figure~\ref{fig:VennW2}).\\
Moreover, we have understood that the position $[\hat{\mu},\hat{\nu}]\propto\hat{\mathbbm{1}}$ that in~\citep{Forgacs2005,Lozano:2006xn} is presented as arbitrary, may not be considered such \emph{a priori}: from the Lemma~\ref{lem:D*ba} we can see in fact how such position is necessary if, representig the twist matrices as elements of $\mathscr{D}_{\spt{\beta\alpha}}$, we want to preserve the invariance under \mbox{translations on $\mathcal{T}^{2}_{\nc}$.}\\
As concerns the cases in which $[\hat{\mu},\hat{\nu}]$ is not proportional to $\hat{\mathbbm{1}}$ or $N=2$, the $(\Omega_{\mu},\Omega_{\nu})$ to choose are those of $\mathscr{D}^{\ell}_{\spt{\alpha\alpha}}$: they can be useful when \emph{ab initio} $[\hat{\mu},\hat{\nu}]$ is set non proportional to $\hat{\mathbbm{1}}$ (see ref.\,\citep{Pachol:2013}).\pagebreak\\
\indent We have moved, then, to the study of~\eqref{eqn:rev'} as a whole.
We have understood that, for $\mathcal{T}^{2}_{\nc}$, the $(\Omega_{\mu},\Omega_{\nu})\in\mathscr{D}^{*}_{\spt{\beta\alpha}}$ that represent the twist matrices are more precisely the ones that belong also to $\mathscr{D}_{g}$ and that \mbox{$\exp(\iu \theta\hat{\mathbbm{1}}I)$} has a specific relation with $e$ (see Theorem~\ref{thm:rev3} and ($\mathscr{D}_{g,2}$)): but only in the case in which $N>2$ and $[\hat{\mu},\hat{\nu}]\propto\hat{\mathbbm{1}}$. In the other cases no relation has been found between \mbox{$\exp( [\hat{\mu},\hat{\nu}] I)$} and $e$.\\ 
\indent The~\eqref{eqn:indn} has led then to revisit the \emph{consistency equation}. So a different setting (cf.\;eqs.\,\eqref{eqn:bordH} and\,\,\eqref{eqn:bordUA}) has been laid out of typical TBC compared to that of refs.~\citep{Arroyo98,Arroyo04,Gonzalez07,Forgacs2005,Lozano:2006xn}.\\
\indent Finally, \eqref{eqn:th1e} and, more specifically the Corollary~\ref{cor:thkerom}, has proved that, differently from what~\citep{Forgacs2005,Lozano:2006xn} claim, \emph{considering $\hat{\mu}$ and $\hat{\nu}$ not commutative to each other does not necessarily mean considering $\mathcal{T}^{2}$ non-commutative}.
In other words, we do not get $\mathcal{T}^{2}_{\cc}$ from $\mathcal{T}^{2}_{\nc}$ setting $\theta$ to zero!\\
In this regard, differently from $\mathcal{T}^{2}_{\nc}$, we have found that \mbox{$\exp(\iu \theta\hat{\mathbbm{1}}I)=e$} while as for $\mathcal{T}^{2}_{\nc}$, we have learnt that the position $[\hat{\mu},\hat{\nu}]\propto\hat{\mathbbm{1}}$ may not be considered \emph{a priori} arbitrary: but this time it has been necessary to preserve the property of cyclicity of $Z_{N}$ and not the invariance under translations on $\mathcal{T}^{2}_{\cc}$ (see Theorem~\ref{thm:rev3Tc} and eq.\,\eqref{eqn:eTCDab}). 


\section{The Ginzburg-Landau lagrangian density}
\label{sec:lagrU1}

\noindent In this section we will study the typical lagrangian density of the \emph{dual superconductivity model} on the fiber bundle $\mathcal{B}_{\spt{U(1)}}=\mathcal{T}^{2}\times U(1)$ in non-commutative context according to the approach set out in section~\ref{sec:twistnc}.\\
\indent The quantity $r_{\spt{\theta}}$, that is the ratio between $\theta$ related to $\mathcal{T}^{2}_{\cc}$ and the one related to $\mathcal{T}^{2}_{\nc}$,
will show to be important. It will be considered belonging to $\Q\0$ and demonstrated to be bound to the order of a cyclic subgroup $\tau_{r_{\spt{\theta}}}$ of $U(1)$. \\
\indent The section will be divided into two paragraphs.\\
\indent  In \S\,\ref{ssec:TBCphi1T2nc}, taking advantage of solutions of TBC for the fields $\phi$ and $A_{\mu}$, we will face the question of the domain of the density of energy $\mathcal{E}$ and we will evaluate $E$ at the point of Bogomolny \mbox{(i.e.\,$\lambda=g^{2}/2$)} for minimum-energy configurations rewriting it \textit{à la} Bogomolny~\citep{Bogomolny1976} and without having to resolve, as usual, any system of equations. In this regard, we will find that $E$ depends on $r_{\spt{\theta}}$ and on $(a_{1}, a_{2})$ so that \mbox{$E\rightarrow+\infty$} for \mbox{$a_{1},a_{2} \rightarrow+\infty$}: the fact that it does not depend on $n$ will \mbox{be seen later in \S\,\ref{ssec:FandM(th)}.}\\
\indent In \S\,\ref{ssec:FandM(th)}, we will examine the influence of $\theta$ on the other observables of the model. No rescaling of $g$ by $\theta$ will be predicted and we will show the relations between $r_{\spt{\theta}}$ with $\mathcal{F}$ and the masses of the particle spectrum obtained after the $U(1)$ symmetry breaking: in particular we will see how despite the symmetry breaking we will not get in general any massless Goldstone boson.\pagebreak\\  
We will have that the quantization of $q_{e}$ will depend exclusively on $r_{\spt{\theta}}$ and not as usual on $n$. We will show moreover how $r_{\spt{\theta}}$ can vary in time and then how the abovementioned masses can vary in time in a quantized way.\\

We will consider the case $[\hat{\mu},\hat{\nu}]=\iu\theta\hat{\mathbbm{1}}$, and we will represent the twist matrices as elements of the set $\mathscr{D}^{*}_{\spt{\beta\alpha,U(1)}}$ (the analogous for $U(1)$ of $\mathscr{D}^{*}_{\spt{\beta\alpha}}$).\footnote{The previous analysis on r.h.s.\,of the~\eqref{eqn:rev'} for $SU(N)/Z_{N}$ can be easily adapted to the case $U(1)$:\,\,the only difference is that the condition $\tr(T_{a})=0$ is absent.}
Moreover, as the group center coincides with $U(1)$ itself, the index $n$ belongs to the whole $\Z$ and it describes only a winding number in the $(x,y)$-plane being ``jumps'' absent (see section~\ref{sec:twistc}).\\
From now on we will consider hemitian $\hat{\mu}$ and $\hat{\nu}$, i.e.\,$\hat{\mu}^{\dagger}=\hat{\mu}$ and $\hat{\nu}^{\dagger}=\hat{\nu}$, and we will use the following notation $\hat{\mu}\equiv\bm{\hat{x}}$ and $\hat{\nu}\equiv\bm{\hat{y}}$. Moreover, with the expression $\mathcal{T}^{2}$, we will indicate either $\mathcal{T}^{2}_{\cc}$ or $\mathcal{T}^{2}_{\nc}$ and with $\mathcal{B}^{\nc}_{\spt{U(1)}}$ ($\mathcal{B}^{\cc}_{\spt{U(1)}}$) we will indicate \mbox{$\mathcal{T}^{2}_{\nc}\times U(1)$} (\mbox{$\mathcal{T}^{2}_{\cc}\times U(1)$}).\\

Let's start with the following
\begin{defask}
\label{def:(BU(1),pi)}
The pair $(\mathcal{B}_{\spt{U(1)}},\rho)$ is $\mathcal{B}_{\spt{U(1)}}$ where the quantities \mbox{$\e^{\iu\pi (g_{i}(x_{j})\bm{\hat{x}}_{k}+g_{l}(x_{m})\bm{\hat{x}}_{n})}$}, with $g_{i}$ linear and homogeneous functions and $i,j,k,l,m,n=1,2$, \mbox{$x_{1}=x, x_{2}=y$}, are represented according to the linear and homogeneous map
\begin{equation}
\label{eqn:defrho}
\rho (\e^{\iu\pi (g_{i}(x_{j})\bm{\hat{x}}_{k}+g_{l}(x_{m})\bm{\hat{x}}_{n})})=\e^{\iu\pi (g_{i}(x_{j})+g_{l}(x_{m}))}
\end{equation}
and the $g_{i}$ are represented as $g_{i}(x_{j})=g_{i}(a_{j})\frac{x_{j}}{a_{j}}$. 
\end{defask}
\begin{remark}
\label{rmk:rhox2}
We can notice how, using (C-B-H-D), from~\eqref{eqn:defrho} we can easily get that 
\begin{equation}
\label{eqn:pi(prod)}
\rho (\e^{\iu\pi g_{i}(x_{j})\bm{\hat{x}}_{k}}\e^{\iu\pi g_{l}(x_{m})\bm{\hat{x}}_{n}})=\rho (\e^{\iu\pi g_{i}(x_{j})\bm{\hat{x}}_{k}})\rho (\e^{\iu\pi g_{l}(x_{m})\bm{\hat{x}}_{n}})\e^{-\frac{\pi^{2}}{2}g_{i}(x_{j})g_{l}(x_{m})[\bm{\hat{x}}_{k},\bm{\hat{x}}_{n}]}.
\end{equation}
From~\eqref{eqn:pi(prod)} we get then that 
{\fontsize{9.5}{10}
 \selectfont
\begin{equation}
\hspace{-0cm}\rho \Bigl(\sum^{\infty}_{r=0}\frac{(\iu \pi g_{i}(x_{j})\bm{\hat{x}}_{k})^{r}}{r!}\e^{\iu\pi g_{l}(x_{m})\bm{\hat{x}}_{n}} \Bigr)=\sum^{\infty}_{p,q=0}\frac{(\iu \pi g_{i}(x_{j}))^{p+q}}{2^{q}p!q!}(\iu\pi g_{l}(x_{m})[\bm{\hat{x}}_{k},\bm{\hat{x}}_{n}])^{q}\e^{\iu\pi g_{l}(x_{m})}
\end{equation}}
\hspace{-.31em}from which, considering the terms up to the second order (enough in this context) we will have
{\fontsize{9.5}{12.3}
 \selectfont
\begin{subequations}
\label{eqn:rho(x,x2)}
\begin{align}
\rho (\bm{\hat{x}}_{k}\,\e^{\iu\pi g_{l}(x_{m})\bm{\hat{x}}_{n}})=&\rho(\e^{\iu\pi g_{l}(x_{m})\bm{\hat{x}}_{n}})+\frac{1}{2}\rho([\bm{\hat{x}}_{k},\e^{\iu\pi g_{l}(x_{m})\bm{\hat{x}}_{n}}]),\\
\rho (\bm{\hat{x}}^{2}_{k}\,\e^{\iu\pi g_{l}(x_{m})\bm{\hat{x}}_{n}})=&\rho(\e^{\iu\pi g_{l}(x_{m})\bm{\hat{x}}_{n}})+\rho([\bm{\hat{x}}_{k},\e^{\iu\pi g_{l}(x_{m})\bm{\hat{x}}_{n}}])+\frac{1}{4}\rho([\bm{\hat{x}}_{k},[\bm{\hat{x}}_{k},\e^{\iu\pi g_{l}(x_{m})\bm{\hat{x}}_{n}}]]).
\end{align} 
\end{subequations}}
\end{remark}
\pagebreak Now, considering\footnote{We will call  $\mathcal{A}_{\theta_{\cc}}$($\mathcal{A}_{\theta_{\nc}}$) the space of functions $\mathcal{A}_{\theta}$ when it's referred to  $\mathcal{T}^{2}_{\cc}$($\mathcal{T}^{2}_{\nc}$).}
\begin{defask}
\label{def:Ath}
$\displaystyle\mathcal{A}_{\theta}=\{f\colon f \,\,\mathrm{of\,\, periods}\,\, (a_{1}, a_{2})\}$, where $f(\hat{x},\hat{y})\colon M_{O}\rightarrow \C$ with \mbox{$M=\mathcal{T}^{2}$,}
\end{defask}
\hspace{-1.5em}and replacing the integral operator with the following trace operator\footnote{The definition of trace is derived from the one given in refs.\,\citep{Forgacs2005,Lozano:2006xn}.} 
\begin{defask}
\label{def:tr}
\textsc{(trace)}\\
Given $(\mathcal{B}_{\spt{U(1)}},\rho)$, $\tr f(\hat{x},\hat{y})=\int_{\spt{\mathcal{T}^{2}}} \rho(f(\hat{x},\hat{y}))$  where $f\in\mathcal{A}_{\theta}$,  
\end{defask}
\hspace{-1.5em}we get the following
\begin{thm}
\label{thm:EBogo}
Given $(\mathcal{B}^{\nc}_{\spt{U(1)}},\rho)$, if $\tr(\phi^{\dagger}\phi)=\tr(\phi\phi^{\dagger})$, $\tr(\phi^{\dagger}\phi)^{2}=\tr(\phi\phi^{\dagger})^{2}$ and \mbox{$A^{\dagger}_{i}=A_{i}$} then $E$ can be rewritten as 
{\fontsize{9}{12,5}
  \selectfont
\begin{subequations} 
\label{eqn:erg'l}
\begin{align}
\hspace{-1.85cm}E=\tr\biggl(\frac{1}{2}(\hat{D}_{i}\phi-\iu\varepsilon_{ij}\hat{D}_{j}\phi)^{\dagger}(\hat{D}_{i}\phi-\iu\varepsilon_{ij}\hat{D}_{j}\phi)+\frac{1}{4}&\Bigl(F_{ij}-g\varepsilon_{ij}(\phi\phi^{\dagger}-\phi^{2}_{0})\Bigr)^{2}+\Bigl(\lambda -\frac{g^{2}}{2}\Bigr)(\phi^{\dagger}\phi-\phi^{2}_{0})^{2}-\frac{g}{2}\phi^{2}_{0}\varepsilon_{ij}F_{ij}+\notag\\
&+\frac{\iu}{2}\varepsilon_{ij}\bigl(\hat{\partial}_{i}(\phi^{\dagger}\hat{D}_{j}\phi)-\hat{\partial}_{j}(\phi^{\dagger}\hat{D}_{i}\phi)\bigr)+\label{eqn:erg1}\\
&-\frac{g}{2}\varepsilon_{ij}\phi^{\dagger}F^{c}_{ij}\phi +\frac{\iu}{2}g^{2}\varepsilon_{ij}\phi^{\dagger}[A_{i},A_{j}]\phi\,+\label{eqn:erg2}\\
&+\frac{g}{2}\varepsilon_{ij}F^{c}_{ij}\phi\phi^{\dagger}-\frac{\iu}{2}g^{2}\varepsilon_{ij}[A_{i},A_{j}]\phi\phi^{\dagger}\biggr)\label{eqn:erg3}
\end{align} 
\end{subequations}}
\!\!where 
\begin{equation}
\hat{D}_{i}\phi=\hat{\partial}_{i}\phi-\iu g A_{i}\phi
\end{equation}
and 
\begin{equation}
\label{eqn:[A1,A2]}
F_{ij}=F^{c}_{ij}-\iu g[A_{i},A_{j}]
\end{equation}
with $F^{c}_{ij}=\hat{\partial}_{i} A_{j}-\hat{\partial}_{j} A_{i}$.
\end{thm}
\begin{proof}
Cf.\,ref.\,\citep{Bogomolny1976}: in the present context the coordinates are non-commutative by hypothesis and we have to take them into account when we consider the products between $\phi,\phi^{\dagger}$ and $A_{i}$. 
\end{proof}

\begin{remark}
\label{rmk:[A1,A2]eTr}
We have to notice how in~\eqref{eqn:[A1,A2]} the term $[A_{1},A_{2}]$ is typical of the non-commutativity of coordinates and not, as usual, of the non-abelianity of the gauge group: we are in fact considering here the group $U(1)$!
\end{remark}

The point is, are there on $\mathcal{T}^{2}_{\nc}$ any solutions to TBC that are compatible with a rewriting \textit{à la} Bogomolny of $E$?\\
In other words: are there on $\mathcal{T}^{2}_{\nc}$ any solutions to TBC that meet the hypotheses of the abovementioned theorem?\\
In the refs.\,\citep{Forgacs2005,Lozano:2006xn} the answer is negative because, as we said in the Introduction, $\mathcal{E}$ has not got a domain of definition. Despite this, $E$ is anyway evaluated, with \mbox{$\lambda=g^{2}/2$}, rewriting it \textit{à la} Bogomolny,\footnote{Each quantity that composes $\mathcal{E}$ is integrated on its natural domain of periodicity.} finding out that it depends uniquely on $n$ as in the commutative case (see refs.\,\citep{Baal82,Manton2004,Arroyo04,EWeinberg2012}). As concerns then the other hypotheses of the theorem, generalizing the expression of $\mathcal{E}$ and introducing $\phi_{1}$ and $\phi_{2}$ (built starting from $\phi$ and that meet the same TBC), they only demonstrate that (see ref.\,\citep{Forgacs2005})   
\begin{equation}
\mathrm{Tr}_{\spt{\mathcal{T}^{2}_{\nc}}}(\phi^{\dagger}_{1}\phi_{2})=\mathrm{Tr}_{\spt{\tilde{\mathcal{T}}^{2}_{\nc}}}(\phi_{2}\phi^{\dagger}_{1})
\end{equation}
where the quantity $\phi_{2}\phi^{\dagger}_{1}$ is defined there on $\tilde{\mathcal{T}}^{2}_{\nc}$ whose periods, as we have explained, are scaled by a factor depending on $\theta$ in comparison with those of $\mathcal{T}^{2}_{\nc}$. There is not any demonstration about the equality between \mbox{$\mathrm{Tr}_{\spt{\mathcal{T}^{2}_{\nc}}}(\phi^{\dagger}_{1}\phi_{2})^{2}$} and \mbox{$\mathrm{Tr}_{\spt{\tilde{\mathcal{T}}^{2}_{\nc}}}(\phi_{2}\phi^{\dagger}_{1})^{2}$}. As concerns $A_{i}$, they find a hermitian representation of it on $\tilde{\mathcal{T}}^{2}_{\nc}$.\\
\indent In the next paragraph, we will show how here the answer to the above question is affirmative. Starting from the set $\mathscr{F}_{r_{\spt{\theta}}}$ of solutions to TBC related to $\phi$, we will demonstrate that $\mathcal{E}$ has $\mathcal{T}^{2}_{\nc}$ as domain\footnote{At the origin of the discordance on the domain of $\mathcal{E}$ with the refs.\,\citep{Forgacs2005,Lozano:2006xn} there is essentially the fact that here the quantity $\hat{x}+a\hat{\mathbbm{1}}$ with $a\in\R$ is not interpreted as a translation $t_{a}$ in the direction $\hat{x}$ of a quantity $a$: in fact in the present approach $t_{a}$ means $x\bm{\hat{x}}+a\bm{\hat{x}}=(x+a)\bm{\hat{x}}$ \mbox{with $x\in\R$.}\label{fn:disc}} and that $E$ can be rewritten \textit{à la} Bogomolny. In the end we will evaluate $E$, thus rewritten, at the point of Bogomolny for minimum-energy configurations without resolving, as already stated, any system of equations.

\subsection{$\mathcal{T}^{2}_{\nc}$ as domain of $\mathcal{E}$ and the Bogomolny rewriting}
\label{ssec:TBCphi1T2nc}
\noindent In this paragraph we will consider the following generalization of $E$  
\begin{equation}
\label{eqn:tenergy'12}
E_{\spt{12}}=\int\biggl(\frac{1}{4}F_{ij}F^{ij}+(\hat{D}_{i}\phi_{1})^{\dagger}(\hat{D}^{i}\phi_{2})+\lambda(\phi^{\dagger}_{1}\phi_{2}-\phi^{2}_{0})^{2}\biggr),
\end{equation}
and the following 
\begin{thm}
\label{thm:EBogorev}
Given $(\mathcal{B}^{\nc}_{\spt{U(1)}},\rho)$, if $\tr(\phi^{\dagger}_{1}\phi_{2})=\tr(\phi_{2}\phi^{\dagger}_{1})$, $\tr(\phi^{\dagger}_{1}\phi_{2})^{2}=\tr(\phi_{2}\phi^{\dagger}_{1})^{2}$ and \mbox{$A^{\dagger}_{i}=A_{i}$} then $E_{\spt{12}}$ can be rewritten as
{\fontsize{9}{12,5}
  \selectfont
\begin{subequations} 
\label{eqn:erg'12}
\begin{align}
\hspace{-2,75cm}E_{\spt{12}}=\tr\biggl(\frac{1}{2}(\hat{D}_{i}\phi_{1}-\iu\varepsilon_{ij}\hat{D}_{j}\phi_{1})^{\dagger}(\hat{D}_{i}\phi_{2}-\iu\varepsilon_{ij}\hat{D}_{j}\phi_{2})+\frac{1}{4}&\Bigl(F_{ij}-g\varepsilon_{ij}(\phi_{2}\phi^{\dagger}_{1}-\phi^{2}_{0})\Bigr)^{2}+\Bigl(\lambda -\frac{g^{2}}{2}\Bigr)(\phi^{\dagger}_{1}\phi_{2}-\phi^{2}_{0})^{2}-\frac{g}{2}\phi^{2}_{0}\varepsilon_{ij}F_{ij}+\notag\\
&+\frac{\iu}{2}\varepsilon_{ij}\bigl(\hat{\partial}_{i}(\phi^{\dagger}_{1}\hat{D}_{j}\phi_{2})-\hat{\partial}_{j}(\phi^{\dagger}_{1}\hat{D}_{i}\phi_{2})\bigr)+\label{eqn:erg1gen}\\
&-\frac{g}{2}\varepsilon_{ij}\phi^{\dagger}_{1}F^{c}_{ij}\phi_{2} +\frac{\iu}{2}g^{2}\varepsilon_{ij}\phi^{\dagger}_{1}[A_{i},A_{j}]\phi_{2}\,+\label{eqn:erg2gen}\\
&+\frac{g}{2}\varepsilon_{ij}F^{c}_{ij}\phi_{2}\phi^{\dagger}_{1}-\frac{\iu}{2}g^{2}\varepsilon_{ij}[A_{i},A_{j}]\phi_{2}\phi^{\dagger}_{1}\biggr).\label{eqn:erg3gen}
\end{align} 
\end{subequations}}

\end{thm}

\begin{proof}
The demonstration is similar to the one of Theorem~\ref{thm:EBogo}.  
\end{proof}
We will show up the set $\mathscr{F}_{r_{\spt{\theta}}}$ of solutions to TBC related to $\phi$. We will see how $\mathscr{F}_{r_{\spt{\theta}}}\neq\emptyset$ if{}f $r_{\spt{\theta}}\in\Q$ and how $r_{\spt{\theta}}$ is bound to the order of a cyclic subgroup $\tau_{r_{\spt{\theta}}}$ of $U(1)$ (see \S\,\ref{ssec:TBCphi}). From $\mathscr{F}_{r_{\spt{\theta}}}$, we will see how the three addends of $\mathcal{E}_{\spt{12}}$, the density of energy $E_{\spt{12}}$, can only belong to a unique $\mathcal{A}_{\theta_{\nc}}$: no $\tilde{\mathcal{T}}^{2}_{\nc}$ will be thus involved (see \S\,\ref{ssec:p1dp2ep2p1d}). We will evaluate eventually the $E_{\spt{12}}$ at the point of Bogomolny for minimum-energy configurations rewriting \textit{à la} Bogomolny (see \S\,\ref{ssec:domE}). In this regard we will find that $E_{\spt{12}}$ depends on $r_{\spt{\theta}}$ and on $(a_{1}, a_{2})$ so that \mbox{$E\rightarrow+\infty$} for \mbox{$a_{1},a_{2} \rightarrow+\infty$}: the fact that it does not depend on $n$ will be seen in \S\,\ref{ssec:FandM(th)}.

\subsubsection{Solutions to TBC related to $\phi$}
\label{ssec:TBCphi}
\noindent We will show now some solutions to TBC related to $\phi$.
In this regard, let's consider the set $\mathscr{T}$ defined as\footnote{The $(\mathscr{T}_{\spt{2}})$ derives from~\eqref{eqn:bordH2} considering the~\eqref{eqn:defz'}.\label{fn:T2}} 
\begin{defask}
\label{def:T}
$\!\mathscr{T}=\{\phi\colon \phi$ meets the following TBC,
\mbox{$\mathlarger{(\mathscr{T}_{\spt{1}})}\hspace{0.35em}\phi(\hat{x}+\hat{a}_{1},\hat{y})=\Omega'_{1}(\hat{x},\hat{y})\phi(\hat{x},\hat{y})$} and $\mathlarger{(\mathscr{T}_{\spt{2}})}\hspace{0.35em}\phi(\hat{x},\hat{y}+\hat{a}_{2})=\Omega'_{2}(\hat{x},\hat{y})\phi(\hat{x},\hat{y})$, where $\Omega'_{1}(\Omega'_{2})$ are equal to $\Omega_{1}(\Omega_{2})$ modulo a gauge transformation $\Omega_{x}(\Omega_{y})$ with $(\Omega_{1},\Omega_{2})\in\mathscr{D}^{*}_{\spt{\beta\alpha,U(1)}}\}$
\end{defask}
\hspace{-1.6em}and the set\footnote{The elements of $\mathscr{F}$ are inspired by Theta functions $\vartheta_{4}$~\citep{Watson1969}. In Definition~\ref{def:Fphi} we factorized out the term of $f^{\spt{\omega}}_{q,k}(\hat{x},\hat{y},\theta\hat{\mathbbm{1}})$ with the exponent proportional to $\hat{\mathbbm{1}}$. }
\begin{defask}
\label{def:Fphi}
$\mathscr{F}=\{\phi\colon \phi(\hat{x},\hat{y})=\sum^{\infty}_{q,k=-\infty}f^{\spt{\omega}}_{q,k}(\hat{x},\hat{y})\}$ with
\begin{equation}
f^{\spt{\omega}}_{q,k}(\hat{x},\hat{y})=(-)^{q+k}\e^{I_{q,k}(\hat{z}_{1},\hat{z}_{2},\hat{\tau}_{1},\hat{\tau}_{2})}\e^{\omega_{q,k}(x,y)}\e^{-\iu\pi q^{2}\hat{\tau}_{1}-\iu\pi k^{2}\hat{\tau}_{2}-2\iu q\hat{z}_{1}-2\iu k\hat{z}_{2}}
\end{equation}
where
{\fontsize{9}{12,5}
 \selectfont
\begin{equation}
\hspace{-1,5cm}I_{q,k}(\hat{z}_{1},\hat{z}_{2},\hat{\tau}_{1},\hat{\tau}_{2})=-2q(k+1)[\hat{z}_{1},\hat{z}_{2}]-\frac{\pi^{2}}{2}q^{2}k^{2}[\hat{\tau}_{1},\hat{\tau}_{2}]+\frac{\pi}{3} q^{3}[\hat{\tau}_{1},\hat{z}_{1}]-\pi q^{2}(k+1)[\hat{\tau}_{1},\hat{z}_{2}]+\pi qk^{2}[\hat{\tau}_{2},\hat{z}_{1}]+\frac{\pi}{3}k^{3}[\hat{\tau}_{2},\hat{z}_{2}]
\end{equation}}
\!\!and
{\fontsize{8}{12,5}
 \selectfont
\begin{equation}
\hspace{-0cm}\omega_{q,k}(x,y)=\pi\biggl(\frac{\Im^{2}(\omega_{1}(x))}{\Im(\omega_{1}(a_{1}))}+\frac{\Im^{2}(\omega_{2}(y))}{\Im(\omega_{2}(a_{2}))}\biggr)-\iu\pi \bigl(q^{2}\omega_{1}(a_{1})+k^{2}\omega_{2}(a_{2})+2q\omega_{1}(x)+2k\omega_{2}(y)\bigr),
\end{equation}}
\hspace{-.45em}with $\hat{z}_{1}=\pi\alpha_{1}(x)\bm{\hat{x}}+\pi\beta_{1}(y)\bm{\hat{y}}$, $\hat{z}_{2}=\pi\alpha_{2}(x)\bm{\hat{x}}+\pi\beta_{2}(y)\bm{\hat{y}}$, $\hat{\tau}_{1}=\alpha_{1}(a_{1})\bm{\hat{x}}$ and \mbox{$\hat{\tau}_{2}=\beta_{2}(a_{2})\bm{\hat{y}}$} where $\alpha_{i}$'s, $\beta_{i}$'s and $\omega_{i}$'s are linear and homogeneous functions with \mbox{$\omega_{i}\colon\R\rightarrow\C$} and \mbox{$\Im(\omega_{i}(a_{i}))<0$}.
\end{defask}
\begin{remark}
\label{rmk:Imo<0}
The requirement $\Im(\omega_{i}(a_{i}))<0$ in the Definition~\ref{def:Fphi} aims to obtain $\rho(\phi(\hat{x},\hat{y}))$ as a convergent function series. In fact \mbox{$\rho(\phi(\hat{x},\hat{y}))$} is a series of analytic functions uniformly convergent as
\begin{equation}
\hspace{-0cm}\abs{\rho(\phi(\hat{x},\hat{y}))}\leq\sum^{\infty}_{q,k=-\infty}\e^{\pi q^{2}\Im(\omega_{1}(a_{1}))}\e^{\pi k^{2}\Im(\omega_{2}(a_{2}))}\e^{2\pi (q\abs{\omega_{1}(a_{1})}+k\abs{\omega_{2}(a_{2})})}
\end{equation}
that is convergent because d'Alembert ratios are in the form
\begin{equation}
\e^{\pi (2q+1)\Im(\omega_{1}(a_{1}))}\e^{\pi (2k+1)\Im(\omega_{2}(a_{2}))}\e^{2\pi (\abs{\omega_{1}(a_{1})}+\abs{\omega_{2}(a_{2})})}
\end{equation}
and tend to zero as $q,k\rightarrow\infty$.
In particular we will have thus that $\rho(\phi(\hat{x},\hat{y}))$ is an integral function.
\end{remark}
Let's take into account that  
\begin{lem}
\label{lem:e^[x,y]=+-e}
Considering the non constant function $g_{1}(x)$ \emph{(}$g_{2}(y)$\emph{)} on $\mathcal{T}^{2}_{\nc}$ with $g_{1}(\bar{x})\neq 0$ \emph{(}$g_{2}(\bar{y})\neq 0$\emph{)} where $\bar{x}$ \emph{(}$\bar{y}$\emph{)} is a specific value of $x$ \emph{(}$y$\emph{)}, we have that
\begin{center}
\mbox{$\e^{\iu \pi g_{1}(\bar{x})\bm{\hat{x}}}=\pm e$\quad\emph{(}$\e^{\iu \pi g_{2}(\bar{y})\bm{\hat{y}}}=\pm e$\emph{)} $\Leftrightarrow$ $\e^{\iu \pi g_{1}(\bar{x})\theta\hat{\mathbbm{1}}}=e$\quad \emph{(}$\e^{\iu \pi g_{2}(\bar{y})\theta\hat{\mathbbm{1}}}=e$\emph{)}}.
\end{center} 
\end{lem}
\begin{proof}
It's sufficient to consider that by hypothesis $\mathcal{T}^{2}$ is non-commutative and so $\e^{\iu\pi\theta\hat{\mathbbm{1}}}\neq e$ (see\,Theorem~\ref{teoth}), and that 
\begin{equation}
\label{eqn:de^[x,y]=e}
\e^{\iu\pi g_{1}(\bar{x})\bm{\hat{x}}}\e^{\iu g_{2}(y)\bm{\hat{y}}}(\e^{\iu g_{2}(y)\bm{\hat{y}}}\e^{\iu\pi g_{1}(\bar{x})\bm{\hat{x}}})^{-1}=\e^{-\iu\pi g_{1}(\bar{x})g_{2}(y)\theta\hat{\mathbbm{1}}}.
\end{equation}
The thesis derives considering that $\e^{\iu \pi g_{1}(\bar{x})g_{2}(y)\theta\hat{\mathbbm{1}}}=(\e^{\iu \pi g_{1}(\bar{x})\theta\hat{\mathbbm{1}}})^{g_{2}(y)}$.
Similar demonstration for $\e^{\iu \pi g_{2}(\bar{y})\bm{\hat{y}}}$.
\end{proof}
Denoted by $\theta_{\cc}$ ($\theta_{\nc}$) the parameter $\theta$ for $\mathcal{T}^{2}_{\cc}$ ($\mathcal{T}^{2}_{\nc}$), with \mbox{$\theta_{\cc}\in [0]\0$} where \mbox{$[0]\in\R/2\Z$} and \mbox{$\theta_{\nc}\in\R\setminus [0]$} (see Corollary~\ref{cor:thkerom}), the Lemma~\ref{lem:e^[x,y]=+-e} implies the following  
\begin{corlem}
\label{cor:rth}
In the hypotheses above, 
\begin{center}
\mbox{$\e^{\iu \pi g_{1}(\bar{x})\bm{\hat{x}}}=\pm e$\quad\emph{(}$\e^{\iu \pi g_{2}(\bar{y})\bm{\hat{y}}}=\pm e$\emph{)} $\Leftrightarrow$ $g_{1}(\bar{x})=k_{1}r_{\spt{\theta}}$\quad \emph{(}$g_{2}(\bar{y})=k_{2}r_{\spt{\theta}}$\emph{)}} 
\end{center}
with $k_{1},k_{2}\in\Z\0 $ and $r_{\spt{\theta}}=\theta_{\cc}/\theta_{\nc}$.
\end{corlem}
\begin{proof}
It's sufficient to consider that \mbox{$\exp(\iu \pi k_{1}\theta_{\cc}\hat{\mathbbm{1}})=e$} \mbox{($\exp(\iu \pi k_{2}\theta_{\cc}\hat{\mathbbm{1}})=e$)} (see Theorem~\ref{teoth}) 
and so that, taking into account the Lemma~\ref{lem:e^[x,y]=+-e}, $g_{1}(\bar{x})\theta_{\nc}=k_{1}\theta_{\cc}$ (\mbox{$g_{2}(\bar{y})\theta_{\nc}=k_{2}\theta_{\cc}$}).
\end{proof}
Now, considering the subset of $\mathscr{F}$ 
\begin{defask}
$\mathscr{F}_{r_{\spt{\theta}}}=\{\phi\colon \phi\in\mathscr{F}\,\,\mathrm{with}\,\,\alpha_{i}\mathrm{'s\,\,and}\,\,\beta_{i}\mathrm{'s\,\,that\,\,meet\,\,the}\,\,\mathrm{(}\mathscr{D}^{*}_{{\spt{\beta\alpha}},1}\mathrm{)}$ such that $\e^{\iu \frac{\pi}{3}  \hat{\tau}'_{i}}$ are equal to $\pm e$, where $\hat{\tau}'_{1}=\beta_{1}(a_{2})\bm{\hat{y}}$ and $\hat{\tau}'_{2}=\alpha_{2}(a_{1})\bm{\hat{x}}\}$,
\end{defask} 
\hspace{-1.55em}we find that
\begin{thm}
\label{thm:FcapT=F*T2nc}
Given $\mathcal{B}^{\nc}_{\spt{U(1)}}$, $\mathscr{F}\cap\,\mathscr{T}=\mathscr{F}_{r_{\spt{\theta}}}$ .
\end{thm}
\begin{proof}
Given a $\phi\in\mathscr{F}$ on $\mathcal{T}^{2}_{\nc}$, calculating the quantities \mbox{$\phi(\hat{x}+\hat{a}_{1},\hat{y})$} and\linebreak[4] \mbox{$\phi(\hat{x},\hat{y}+\hat{a}_{2})$}, we find that
\begin{subequations}
\label{eqn:phitrasl}
\begin{align}
\phi(\hat{x}+\hat{a}_{1},\hat{y})&=\Omega'_{1}(\hat{x},\hat{y})\phi(\hat{x},\hat{y})\label{eqn:O'1}\\
\phi(\hat{x},\hat{y}+\hat{a}_{2})&=\Omega'_{2}(\hat{x},\hat{y})\phi(\hat{x},\hat{y})\label{eqn:O'2}
\end{align} 
\end{subequations}
with
{\fontsize{9.5}{14}
 \selectfont 
\begin{subequations}
\label{eqn:O'1O'2}
\begin{align}
\hspace{-0em}\Omega'_{1}(\hat{x},\hat{y})&= -\e^{2\iu\hat{z}_{1}+\iu\pi\hat{\tau}_{1}}\e^{2[\hat{z}_{1},\hat{z}_{2}]}\e^{-\frac{\pi}{3}[\hat{\tau}_{1},\hat{z}_{1}]}\e^{\pi[\hat{\tau}_{1},\hat{z}_{2}]}\e^{2\iu\pi\Re(\omega_{1}(x))+\iu\pi\Re(\omega_{1}(a_{1}))},\label{eqn:O'1}\\
\hspace{-0em}\Omega'_{2}(\hat{x},\hat{y})&= -\e^{2\iu \hat{z}_{2}+\iu\pi \hat{\tau}_{2}}\e^{-\frac{\pi}{3}[\hat{\tau}_{2},\hat{z}_{2}]}\e^{2\iu\pi\Re(\omega_{2}(y))+\iu\pi\Re(\omega_{2}(a_{2}))}\label{eqn:O'2}
\end{align} 
\end{subequations}}
\!\!and equal respectively to $\Omega_{1}$ and $\Omega_{2}\in\mathscr{D}^{*}_{\spt{\beta\alpha,U(1)}}$ modulo a gauge transformation, if{}f the $\alpha_{i}$'s and $\beta_{i}$'s meet the $\mathrm{(}\mathscr{D}^{*}_{{\spt{\beta\alpha}},1}\mathrm{)}$ and $\e^{\iu \frac{\pi}{3}  \hat{\tau}'_{i}}$ are equal to $\pm e$. 
In fact, using the Lemma~\ref{lem:e^[x,y]=+-e}, we have that (for more details see Appendix~\ref{sec:AppFintT})
{\fontsize{9,5}{12,5}
 \selectfont 
\begin{equation}
\hspace{0em}\phi(\hat{x}+\hat{a}_{1},\hat{y})\equiv\phi(\hat{z}_{1}+\pi\hat{\tau}_{1},\hat{z}_{2}+\pi\hat{\tau}'_{2})=\Omega'_{1}(\hat{x},\hat{y})\!\!\!\sum^{\infty}_{q,k=-\infty}\hspace{-.5em}f^{\spt{\omega}}_{q+1,k}(\hat{x},\hat{y},\theta\hat{\mathbbm{1}})=\Omega'_{1}(\hat{x},\hat{y})\phi(\hat{x},\hat{y})\label{eqn:phitrx}
\end{equation}}
\!\!if{}f $\e^{\iu \frac{\pi}{3}\hat{\tau}'_{2}}=\pm e$, and moreover we have also that 
{\fontsize{9,5}{12,5}
 \selectfont 
\begin{equation}
\hspace{0em}\phi(\hat{x},\hat{y}+\hat{a}_{2})\equiv\phi(\hat{z}_{1}+\pi\hat{\tau}'_{1},\hat{z}_{2}+\pi\hat{\tau}_{2})=\Omega'_{2}(\hat{x},\hat{y})\!\!\!\sum^{\infty}_{q,k=-\infty}\hspace{-.5em}f^{\spt{\omega}}_{q,k+1}(\hat{x},\hat{y},\theta\hat{\mathbbm{1}})=\Omega'_{2}(\hat{x},\hat{y})\phi(\hat{x},\hat{y})\label{eqn:phitry}
\end{equation}}
\!\!if{}f $\e^{\iu \frac{\pi}{3}  \hat{\tau}'_{1}}$ is equal to $\pm e$ and $[\hat{z}_{1},\hat{z}_{2}]=0$, i.e.\,$\alpha_{i},\beta_{i}$'s meet the  $\mathrm{(}\mathscr{D}^{*}_{{\spt{\beta\alpha}},1}\mathrm{)}$.\\
We find besides that $\Omega'_{1}$($\Omega'_{2}$) is equal to $\Omega_{1}$($\Omega_{2}$) in $\mathscr{D}^{*}_{\spt{\beta\alpha,U(1)}}$ modulo a gauge transformation $\Omega_{x}$($\Omega_{y}$) where
\begin{subequations}
\label{eqn:O1O2}
\begin{align}
\Omega_{1}(\hat{x},\hat{y})&=\e^{2\pi\iu\beta_{1}(y)\bm{\hat{y}}}\e^{2\pi\iu\alpha_{1}(x)\bm{\hat{x}}},\\
\Omega_{2}(\hat{x},\hat{y})&=\e^{2\pi\iu\alpha_{2}(x)\bm{\hat{x}}}\e^{2\pi\iu\beta_{2}(y)\bm{\hat{y}}},
\end{align} 
\end{subequations}
with $\alpha_{i}$'s and $\beta_{i}$'s that meet the $\mathrm{(}\mathscr{D}^{*}_{{\spt{\beta\alpha}},1}\mathrm{)}$ such that $\e^{\iu \frac{\pi}{3}  \hat{\tau}'_{i}}$ are equal to $\pm e$, and
\begin{subequations}
\begin{align}
\Omega_{x}(\hat{x},\hat{y}) &=\e^{\iu\pi\alpha_{1}(x)\bm{\hat{x}}}\e^{\frac{2}{3}\pi^{2}\alpha_{1}(x)\beta_{1}(y)[\bm{\hat{x}},\bm{\hat{y}}]}\e^{\iu\pi\bigl(\gamma_{1}(x)+\frac{\Re^{2}(\omega_{1}(x))}{\Re(\omega_{1}(a_{1}))}\bigr)},\label{eqn:Ox2}\\
\Omega_{y}(\hat{x},\hat{y}) &=\e^{\iu\pi\beta_{2}(y)\bm{\hat{y}}}\e^{-\frac{2}{3}\pi^{2}\alpha_{2}(x)\beta_{2}(y)[\bm{\hat{x}},\bm{\hat{y}}]}\e^{\iu\pi\bigl(\gamma_{2}(y)+\frac{\Re^{2}(\omega_{2}(y))}{\Re(\omega_{2}(a_{2}))}\bigr)},\label{eqn:Oy2}
\end{align}
\end{subequations}
with $\gamma_{i}$  linear functions such that \mbox{$\e^{\iu\pi\gamma_{i}(a_{i})}=-e$}, i.e.\,$\gamma_{i}(a_{i})$ are odd numbers.\\
We have that
\begin{subequations}
\label{eqn:O'1O'2def}
\begin{align}
\Omega'_{1}(\hat{x},\hat{y})&=\Omega_{x}(\hat{x}+\hat{a}_{1},\hat{y})\Omega_{1}(\hat{x},\hat{y})\Omega_{x}^{-1}(\hat{x},\hat{y}),\label{eqn:gaugeO1}\\
\Omega'_{2}(\hat{x},\hat{y})&=\Omega_{y}(\hat{x},\hat{y}+\hat{a}_{2})\Omega_{2}(\hat{x},\hat{y})\Omega_{y}^{-1}(\hat{x},\hat{y}),\label{eqn:gaugeO2}
\end{align}
\end{subequations}
where, in the~\eqref{eqn:gaugeO1}, we consider that $[\hat{z}_{1},\hat{z}_{2}]=0$.
\end{proof}
Furthermore, we obtain that 
\begin{thm}
\label{thm:(BUnc(1),pi)}
Given $(\mathcal{B}^{\nc}_{\spt{U(1)}},\rho)$, we have that
\begin{center}
$\mathscr{F}_{r_{\spt{\theta}}}\neq\emptyset$ if{}f $r_{\spt{\theta}}\in\Q$.
\end{center}
\end{thm}
\begin{proof}
It's sufficient to consider the Corollary~\ref{cor:rth}.
\end{proof}
We will define $(\mathcal{B}^{*\nc }_{\spt{U(1)}},\rho)$ as
\begin{defask}
$(\mathcal{B}^{*\nc }_{\spt{U(1)}},\rho)=(\mathcal{B}^{\nc}_{\spt{U(1)}},\rho)$ with $r_{\spt{\theta}}\in\Q$.
\end{defask}
\pagebreak 
Now, the Theorem~\ref{thm:(BUnc(1),pi)} is followed by
\begin{corth}
\label{corth:thetag}
Given $(\mathcal{B}^{*\nc }_{\spt{U(1)}},\rho)$, $e^{\iu\pi \theta\hat{\mathbbm{1}}}$ is the generator of a cyclic subgroup $\tau_{r_{\spt{\theta}}}$ of $U(1)$ so that
\begin{enumerate}[label=\emph{(\roman*)}]
\item  if $r_{\spt{\theta}}\in\Z$ then the order of $\tau_{r_{\spt{\theta}}}$ is equal to $\abs{r_{\spt{\theta}}}$, otherwise \label{item:(ir1)} 
\item if $r_{\spt{\theta}}\in\Q\sim\Z$, representing $r_{\spt{\theta}}$ as $p_{\spt{r_{\spt{\theta}}}}/q_{\spt{r_{\spt{\theta}}}}$ with $p_{\spt{r_{\spt{\theta}}}},q_{\spt{r_{\spt{\theta}}}}$ relatively prime integers, then the order of $\tau_{r_{\spt{\theta}}}$ is equal to $\abs{p_{\spt{r_{\spt{\theta}}}}}$. \label{item:(ir2)}
\end{enumerate}
\end{corth}
\begin{proof}
Let's consider the elements of $\mathscr{F}_{r_{\spt{\theta}}}$ and  Corollary~\ref{cor:rth}.
\end{proof}

\subsubsection{The belonging of $\mathcal{E}$ to $\mathcal{A}_{\theta_{\nc}}$}
\label{ssec:p1dp2ep2p1d}
\noindent The elements of $\mathscr{F}_{r_{\spt{\theta}}}$ will thus be ``bricks'' by which we will build the three quantities that make $\mathcal{E}_{\spt{12}}$.\\
\indent We will analyze now the conditions according to which the three addends of $\mathcal{E}_{\spt{12}}$ can be represented biperiodic and we will show how in such conditions $\mathcal{E}_{\spt{12}}$ belongs to $\mathcal{A}_{\theta_{\nc}}$.\\
\indent Let's see first when the quantities \mbox{$\phi^{\dagger}_{1}\phi_{2}$} and \mbox{$\phi_{2}\phi^{\dagger}_{1}$} are biperiodic. In this regard, let's define
\begin{defask}
\label{def:UAath}
$\mathcal{A}=\bigcup\limits_{\alpha\in P}\mathcal{A}^{\alpha}_{\theta}$ with $P$ a non-void index set and $\mathcal{A}^{\alpha}_{\theta}$ related 
to a 2-dimensional torus of periods $(a_{1;\alpha},a_{2;\alpha})$,
\end{defask}
\hspace{-1.45em}and the subsets of $\mathscr{F}_{r_{\spt{\theta}}}\times\mathscr{F}_{r_{\spt{\theta}}}$ 
\begin{defask}
\label{def:frackFw}
$\mathfrak{F}_{r_{\spt{\theta}}}=\{(\phi_{1},\phi_{2})\colon (\phi_{1},\phi_{2})\in\mathscr{F}_{r_{\spt{\theta}}}\times\mathscr{F}_{r_{\spt{\theta}}}$ with \mbox{$\Re(\omega_{i;j}(x_{i}))=\Re(\omega_{i;k}(x_{i}))$} and $\hat{z}_{i;j}=\hat{z}_{i;k}\}$,
\end{defask}
\begin{defask}
\label{def:tfrackFw}
$\tilde{\mathfrak{F}}_{r_{\spt{\theta}};1}=\{(\phi_{1},\phi_{2})\colon (\phi_{1},\phi_{2})\in\mathscr{F}_{r_{\spt{\theta}}}\times\mathscr{F}_{r_{\spt{\theta}}}$ with \mbox{$[\hat{z}_{i;j},\hat{z}_{k;l}]=0$} and one of the pairs of conditions arising from $\{\e^{2\pi\iu\hat{\tau}_{1;1}}\!=\!\pm e,\e^{2\pi\iu\hat{\tau}_{1;2}}\!=\!\pm e\}\!\times\!\{\e^{2\pi\iu\hat{\tau}_{2;1}}\!=\!\pm e,\e^{2\pi\iu\hat{\tau}_{2;2}}\!=\!\pm e\}\}$,
\end{defask}
\hspace{-1.45em}and the subset of $\tilde{\mathfrak{F}}_{r_{\spt{\theta}};1}$
\begin{defask}
\label{def:tfrackFw0}
$\tilde{\mathfrak{F}}_{r_{\spt{\theta}};0}=\{(\phi_{1},\phi_{2})\colon (\phi_{1},\phi_{2})\in\tilde{\mathfrak{F}}_{r_{\spt{\theta}};1}$ with one of the pairs of conditions arising from $\{\e^{\pi\iu\hat{\tau}_{1;1}}\!=\!\pm e,\e^{\pi\iu\hat{\tau}_{1;2}}\!=\!\pm e\}\!\times\!\{\e^{\pi\iu\hat{\tau}_{2;1}}\!=\!\pm e,\e^{\pi\iu\hat{\tau}_{2;2}}\!=\!\pm e\}\}$.
\end{defask}
We can prove the following 
\begin{thm}
\label{thm:O12O21period}
Given $(\mathcal{B}^{*\nc }_{\spt{U(1)}},\rho)$ with $(\phi_{1},\phi_{2})\in \mathscr{F}_{r_{\spt{\theta}}}\times\mathscr{F}_{r_{\spt{\theta}}}$, we have that 
{\fontsize{10,4}{12.5}
 \selectfont
\begin{enumerate}[label=\emph{(\roman*)}]
\item  $\phi^{\dagger}_{1}\phi_{2}\in\mathcal{A}$ if{}f \mbox{$(\phi_{1},\phi_{2})\in \mathfrak{F}_{r_{\spt{\theta}}}$}, \label{item:(i)}  
\item $\phi_{2}\phi^{\dagger}_{1}\in\mathcal{A}$ if{}f \mbox{$(\phi_{1},\phi_{2})\in \mathfrak{F}_{r_{\spt{\theta}}}\cap\tilde{\mathfrak{F}}_{r_{\spt{\theta}};1}$}. 
\end{enumerate}}
\end{thm}

\begin{proof}
Given a $(\phi_{1},\phi_{2})\in\mathscr{F}_{r_{\spt{\theta}}}\times\mathscr{F}_{r_{\spt{\theta}}}$, we indicate with $g^{\spt{\omega};\,\spt{1,2}}_{q,q',k,k'}$ and $h^{\spt{\omega};\,\spt{2,1}}_{q,q',k,k'}$ respectively the terms of the development of the quantity $\phi_{1}^{\dagger}\phi_{2}$ and $\phi_{2}\phi_{1}^{\dagger}$, i.e.
\begin{subequations}
\label{eqn:g12h21}
\begin{align}
\phi_{1}^{\dagger}\phi_{2}(\hat{x},\hat{y})&=\hspace{-1em}\sum^{+\infty}_{q,q',k,k'=-\infty}\hspace{-1em}g^{\spt{\omega};\,\spt{1,2}}_{q,q',k,k'}(\hat{x},\hat{y}),\label{eqn:g12}\\
\phi_{2}\phi_{1}^{\dagger}(\hat{x},\hat{y})&=\hspace{-1em}\sum^{+\infty}_{q,q',k,k'=-\infty}\hspace{-1em}h^{\spt{\omega};\,\spt{2,1}}_{q,q',k,k'}(\hat{x},\hat{y}),\label{eqn:h21}
\end{align} 
\end{subequations}
and with $g^{\spt{\omega};\,\spt{1,2,a_{i}}}_{q,q',k,k'}$ and $h^{\spt{\omega};\,\spt{2,1,a_{i}}}_{q,q',k,k'}$ the terms of the development of the quantities\footnote{Let's omit the explicit expression of $g^{\spt{\omega};\,\spt{1,2}}_{q,q',k,k'}, h^{\spt{\omega};\,\spt{2,1}}_{q,q',k,k'}, g^{\spt{\omega};\,\spt{1,2,a_{i}}}_{q,q',k,k'}$ and $h^{\spt{\omega};\,\spt{2,1,a_{i}}}_{q,q',k,k'}$ for want of room.}    
\begin{subequations}
\label{eqn:phi1dU1U2phi2}
\begin{align}
\phi^{\dagger}_{1}\phi_{2}(\hat{x}+\hat{a}_{1},\hat{y})&=\phi^{\dagger}_{1}(\hat{x},\hat{y})\Omega'^{\dagger}_{1;1}(\hat{x},\hat{y})\Omega'_{1;2}(\hat{x},\hat{y})\phi_{2}(\hat{x},\hat{y})=\hspace{-1em}\sum^{+\infty}_{q,q',k,k'=-\infty}\hspace{-1em}g^{\spt{\omega};\,\spt{1,2,a_{1}}}_{q,q',k,k'}(\hat{x},\hat{y}),\label{eqn:bp12x}\\
\phi^{\dagger}_{1}\phi_{2}(\hat{x},\hat{y}+\hat{a}_{2})&=\phi^{\dagger}_{1}(\hat{x},\hat{y})\Omega'^{\dagger}_{2;1}(\hat{x},\hat{y})\Omega'_{2;2}(\hat{x},\hat{y})\phi_{2}(\hat{x},\hat{y})=\hspace{-1em}\sum^{+\infty}_{q,q',k,k'=-\infty}\hspace{-1em}g^{\spt{\omega};\,\spt{1,2,a_{2}}}_{q,q',k,k'}(\hat{x},\hat{y})\label{eqn:bp12y},\\
\phi_{2}\phi^{\dagger}_{1}(\hat{x}+\hat{a}_{1},\hat{y})&=\Omega'_{1;2}(\hat{x},\hat{y})\phi_{2}(\hat{x},\hat{y})\phi^{\dagger}_{1}(\hat{x},\hat{y})\Omega'^{\dagger}_{1;1}(\hat{x},\hat{y})=\hspace{-1em}\sum^{+\infty}_{q,q',k,k'=-\infty}\hspace{-1em}h^{\spt{\omega};\,\spt{2,1,a_{1}}}_{q,q',k,k'}(\hat{x},\hat{y}),\label{eqn:bp21x}\\
\phi_{2}\phi^{\dagger}_{1}(\hat{x},\hat{y}+\hat{a}_{2})&=\Omega'_{2;2}(\hat{x},\hat{y})\phi_{2}(\hat{x},\hat{y})\phi^{\dagger}_{1}(\hat{x},\hat{y})\Omega'^{\dagger}_{2;1}(\hat{x},\hat{y})=\hspace{-1em}\sum^{+\infty}_{q,q',k,k'=-\infty}\hspace{-1em}h^{\spt{\omega};\,\spt{2,1,a_{2}}}_{q,q',k,k'}(\hat{x},\hat{y}),\label{eqn:bp21y}
\end{align} 
\end{subequations}
where $\Omega'_{i;j}$ with $i,j=1,2$ indicates the $\Omega'_{i}$ related to $\phi_{j}$ (cf.\,eqs.\,\eqref{eqn:O'1O'2}). \\
We find that 
\begin{subequations}
\label{eqn:phi1dphi2q-bq}
\begin{align}
&g^{\spt{\omega};\,\spt{1,2,a_{1}}}_{q,q',k,k'}(\hat{x},\hat{y})=g^{\spt{\omega};\,\spt{1,2}}_{q-\bar{q},q'-\bar{q}',k,k'}(\hat{x},\hat{y}),\\
&g^{\spt{\omega};\,\spt{1,2,a_{2}}}_{q,q',k,k'}(\hat{x},\hat{y})=g^{\spt{\omega};\,\spt{1,2}}_{q,q',k-\bar{k},k'-\bar{k}'}(\hat{x},\hat{y}),
\end{align} 
\end{subequations}
with $\bar{q},\bar{q}',\bar{k},\bar{k}'\in\Z$, if{}f $\bar{q},\bar{q}',\bar{k},\bar{k}'$ are null and \mbox{$\Re(\omega_{i;j}(x_{i}))=\Re(\omega_{i;k}(x_{i}))$}, \mbox{$\hat{z}_{i;j}=\hat{z}_{i;k}$}.\\
\indent As concerns $\phi_{2}\phi^{\dagger}_{1}$, considering the Lemma~\ref{lem:e^[x,y]=+-e} and that the condition ($\mathscr{D}^{*}_{{\spt{\beta\alpha}},1}$) is met by hypothesis, we find that
\begin{subequations}
\label{eqn:phi2phi1dq-bq}
\begin{align}
&h^{\spt{\omega};\,\spt{2,1,a_{1}}}_{q,q',k,k'}(\hat{x},\hat{y})=h^{\spt{\omega};\,\spt{2,1}}_{q-\bar{q},q'-\bar{q}',k,k'}(\hat{x},\hat{y}),\\
&h^{\spt{\omega};\,\spt{2,1,a_{2}}}_{q,q',k,k'}(\hat{x},\hat{y})=h^{\spt{\omega};\,\spt{2,1}}_{q,q',k-\bar{k},k'-\bar{k}'}(\hat{x},\hat{y})
\end{align} 
\end{subequations}
if{}f $\bar{q},\bar{q}',\bar{k},\bar{k}'$ are null and \mbox{$\Re(\omega_{i;j}(x_{i}))=\Re(\omega_{i;k}(x_{i}))$}, \mbox{$\hat{z}_{i;j}=\hat{z}_{i;k}$} with \mbox{$\e^{\iu 2\pi \hat{\tau}_{i}}=\pm e$}. 
\end{proof}

Futhermore, we will have that 
\begin{thm}
\label{thm:AiffAth}
In the hypotheses above,
\[
\phi^{\dagger}_{1}\phi_{2},\phi_{2}\phi^{\dagger}_{1}\in\mathcal{A}\Rightarrow\phi^{\dagger}_{1}\phi_{2},\phi_{2}\phi^{\dagger}_{1}\in\mathcal{A}_{\theta_{\nc}} .
\]
\end{thm}
\begin{proof}
It's sufficient to retrace the proof of the Theorem~\ref{thm:O12O21period} and note that, taking as an example just the translation in $\hat{x}$ of $\hat{a}_{1}$ (the same is valid also for the translation in $\hat{y}$ of $\hat{a}_{2}$), $\phi^{\dagger}_{1}\phi_{2}(\hat{x}+\hat{a}_{1},\hat{y})$ $\bigl(\phi_{2}\phi^{\dagger}_{1}(\hat{x}+\hat{a}_{1},\hat{y})\bigr)$ can only be equal to $\phi^{\dagger}_{1}\phi_{2}(\hat{x},\hat{y})$ $\bigl(\phi_{2}\phi^{\dagger}_{1}(\hat{x},\hat{y})\bigr)$ and not to $\phi^{\dagger}_{1}\phi_{2}(\hat{x}+\hat{a}'_{1},\hat{y})$ $\bigl(\phi_{2}\phi^{\dagger}_{1}(\hat{x}+\hat{a}'_{1},\hat{y})\bigr)$ with \mbox{$(a'_{1},a'_{2})\neq (a_{1},a_{2})$} periods related to some $\mathcal{T}'^{2}_{\nc}\neq\mathcal{T}^{2}_{\nc}$. 
\end{proof}

We are now ready to analyze the conditions according to which $\frac{1}{4}F_{ij}F^{ij}$, \mbox{$\bigl((\hat{D}_{i}\phi_{1})^{\dagger}(\hat{D}^{i}\phi_{2})\bigr)(\hat{x},\hat{y})$} and \mbox{$(\phi^{\dagger}_{1}\phi_{2}-\phi^{2}_{0})^{2}$} can be represented biperiodic. 
\begin{thm}
\label{thm:3bper}
Given $(\mathcal{B}^{*\nc }_{\spt{U(1)}},\rho)$ with $(\phi_{1},\phi_{2})\in \mathscr{F}_{r_{\spt{\theta}}}\times\mathscr{F}_{r_{\spt{\theta}}}$ and\footnote{We can notice how the~\eqref{eqn:solA} are solutions of the~\eqref{eqn:bordUA}:\,here the $\Omega_{\mu}$'s are dependent on $\lambda$ and in the place of~\eqref{eqn:bordU2A}, taking into account the $(\mathscr{T}_{\spt{2}})$ (see Definition~\ref{def:T} and the footnote~\ref{fn:T2}), we have considered: 
\begin{equation}
A_{i}(\hat{x},\hat{y}+\hat{a}_{2})=\tilde{\Omega}_{2;2,i}(\hat{x},\hat{y})A_{i}(\hat{x},\hat{y})\tilde{\Omega}^{\dagger}_{2;1,i}(\hat{x},\hat{y})+\frac{\iu}{g}\tilde{\Omega}_{2;2,i}(\hat{x},\hat{y})\hat{\partial}_{i}\tilde{\Omega}^{\dagger}_{2;1,i}(\hat{x},\hat{y}),\notag
\end{equation}
where $\tilde{\Omega}_{2;1,i}$ ($\tilde{\Omega}_{2;2,i}$) indicates the $\Omega'_{2}$ related to $\tilde{\phi}_{1;i}$ ($\tilde{\phi}_{2;i}$) (cf.\,eq.\,\eqref{eqn:O'2}).} 
\begin{equation}
\label{eqn:solA}
A_{i}(\hat{x},\hat{y})=\tilde{A}_{i}(\hat{x},\hat{y})+\varepsilon_{ij}\frac{1}{g a_{j}}\bm{\hat{x}}_{j}
\end{equation}
where $\tilde{A}_{i}(\hat{x},\hat{y})=c_{i}\tilde{\phi}_{2;i}\tilde{\phi}^{\dagger}_{1;i}$, with constant $c_{i}$\footnote{Indicating with $[M]$ the dimension of the mass, in the system of natural units $\hbar=c=1$, we have that $c_{i}$ have dimension $[M]^{\spt{-1/2}}$.} and $(\tilde{\phi}_{1;i},\tilde{\phi}_{2;i})\in\mathfrak{F}_{r_{\spt{\theta}}}$, for
\begin{enumerate}[label=\emph{(\alph*)}]
\item $\frac{1}{4}F_{ij}F^{ij}$ we have that\label{item:a}
\begin{enumerate} [label=\emph{(\alph{enumi}.\roman*)}]
\item if $\hat{\tilde{z}}_{k;i}\neq\hat{\tilde{z}}_{k;j}$ then $\frac{1}{4}F_{ij}F^{ij}\in\mathcal{A}$ if{}f \mbox{$(\tilde{\phi}_{1;i},\tilde{\phi}_{2;i})\in\tilde{\mathfrak{F}}_{r_{\spt{\theta}};1}$},\label{item:ai}
\item if $\hat{\tilde{z}}_{k;i}=\hat{\tilde{z}}_{k;j}$ then $\frac{1}{4}F_{ij}F^{ij}\in\mathcal{A}_{\theta_{\nc}}$, while for \label{item:aii} 
\end{enumerate}
\item $\bigl((\hat{D}_{i}\phi_{1})^{\dagger}(\hat{D}^{i}\phi_{2})\bigr)(\hat{x},\hat{y})$ we have that $\bigl((\hat{D}_{i}\phi_{1})^{\dagger}(\hat{D}^{i}\phi_{2})\bigr)(\hat{x},\hat{y})\in\mathcal{A}$ if{}f \mbox{$(\phi_{i},\tilde{\phi}_{j;k})\in\mathfrak{F}_{r_{\spt{\theta}}}$}, and eventually we have that
\item$(\phi^{\dagger}_{1}\phi_{2}-\phi^{2}_{0})^{2}\in\mathcal{A}$ if{}f \mbox{$(\phi_{1},\phi_{2})\in \mathfrak{F}_{r_{\spt{\theta}}}$}. \label{item:c} 
\end{enumerate}
\end{thm}
\begin{proof}
As concerns the term $\frac{1}{4}F_{ij}F^{ij}$, it's sufficient to notice that, calculating $F_{ij}$ according to the eq.\,\eqref{eqn:solA}, we obtain that 
\begin{equation}
\label{eqn:Fij}
F_{ij}(\hat{x},\hat{y})=-\varepsilon_{ij}\frac{\theta}{g a_{i}a_{j}}-\iu g [\tilde{A}_{i},\tilde{A}_{j}](\hat{x},\hat{y}).
\end{equation}
\hspace{-.02em}Considering 
\begin{equation}
\label{eqn:[Ai,Aj]p}
[\tilde{A}_{i},\tilde{A}_{j}](\hat{x},\hat{y})=\hspace{-1cm}\sum^{+\infty}_{q,q',k,k',r,r',s,s'=-\infty}\hspace{-.15cm}\tilde{A}^{i,j}_{\spt{\substack{q,q',k,k'\\r,r',s,s'}}}(\hat{x},\hat{y})
\end{equation}
and
\begin{subequations}
\label{eqn:[Ai,Aj]a1a2p}
\begin{align}
\hspace{-0cm}[\tilde{A}_{i},\tilde{A}_{j}](\hat{x}+\hat{a}_{1},\hat{y}) =&\hspace{-1cm}\sum^{+\infty}_{q,q',k,k',r,r',s,s'=-\infty}\hspace{-.15cm}\tilde{A}^{i,j;\,a_{1}}_{\spt{\substack{q,q',k,k'\\r,r',s,s'}}}(\hat{x},\hat{y}),\\
\hspace{-0cm}[\tilde{A}_{i},\tilde{A}_{j}](\hat{x},\hat{y}+\hat{a}_{2}) =&\hspace{-1cm}\sum^{+\infty}_{q,q',k,k',r,r',s,s'=-\infty}\hspace{-.15cm}\tilde{A}^{i,j;\,a_{2}}_{\spt{\substack{q,q',k,k'\\r,r',s,s'}}}(\hat{x},\hat{y}),
\end{align}
\end{subequations}
if $\hat{\tilde{z}}_{k;i}\neq\hat{\tilde{z}}_{k;j}$ then, taking into consideration the Lemma~\ref{lem:e^[x,y]=+-e} with the  hypothesis that $(\tilde{\phi}_{1;i},\tilde{\phi}_{2;i})\in\mathfrak{F}_{r_{\spt{\theta}}}$, we find that
{\fontsize{10}{12.5}
 \selectfont
\begin{subequations}
\begin{align}
\tilde{A}^{i,j;\,a_{1}}_{\spt{\substack{q,q',k,k'\\r,r',s,s'}}}(\hat{x},\hat{y})=&\tilde{A}^{i,j}_{\spt{\substack{q-\bar{q},q'-\bar{q}',k,k'\\r-\bar{r},r'-\bar{r}',s,s'}}}(\hat{x},\hat{y}),\\
\tilde{A}^{i,j;\,a_{2}}_{\spt{\substack{q,q',k,k'\\r,r',s,s'}}}(\hat{x},\hat{y})=&\tilde{A}^{i,j}_{\spt{\substack{q,q',k-\bar{k},k'-\bar{k}'\\r,r',s-\bar{s},s'-\bar{s}'}}}(\hat{x},\hat{y}),
\end{align}
\end{subequations}} 
\hspace{-.15cm}if{}f $\bar{q},\bar{q}',\bar{r},\bar{r}'$ ($\bar{k},\bar{k}',\bar{s},\bar{s}'$) are null and $\e^{\iu 2\pi \hat{\tau}_{i;j}}=\pm e$. 
On the other hand, in the case in which $\hat{\tilde{z}}_{k;i}=\hat{\tilde{z}}_{k;j}$, we obtain that $[\tilde{A}_{i},\tilde{A}_{j}](\hat{x},\hat{y})$ is null and so that  \mbox{$\frac{1}{4}F_{ij}F^{ij}\in\mathcal{A}_{\theta_{\nc}}$} without any condition on $\e^{\iu 2\pi \hat{\tau}_{i;j}}$.\\
Calculating explicitly the quantities \mbox{$\bigl((\hat{D}_{i}\phi_{1})^{\dagger}(\hat{D}^{i}\phi_{2})\bigr)(\hat{x}+\hat{a}_{1},\hat{y})$} and \mbox{$\bigl((\hat{D}_{i}\phi_{1})^{\dagger}(\hat{D}^{i}\phi_{2})\bigr)(\hat{x},\hat{y}+\hat{a}_{2})$}, we can easily notice that they are equal to $\bigl((\hat{D}_{i}\phi_{1})^{\dagger}(\hat{D}^{i}\phi_{2})\bigr)(\hat{x},\hat{y})$ if{}f \mbox{$(\phi_{i},\tilde{\phi}_{j;k})\in\mathfrak{F}_{r_{\spt{\theta}}}$}. 
In the end, the case\emph{~\ref{item:c}} is obtained from the case\emph{~\ref{item:(i)}} of the Theorem~\ref{thm:O12O21period}.
\end{proof}

In particular we have that 
\begin{thm} 
\label{thm:3inAth}
In the hypotheses above, 
\begin{itemize}
\item[$\mathrm{(i)}$] if $\frac{1}{4}F_{ij}F^{ij}\in\mathcal{A}$ then $\frac{1}{4}F_{ij}F^{ij}\in\mathcal{A}_{\theta_{\nc}}$,
\item[$\mathrm{(ii)}$] if $\bigl((\hat{D}_{i}\phi_{1})^{\dagger}(\hat{D}^{i}\phi_{2})\bigr)(\hat{x},\hat{y})\in\mathcal{A}$ then $\bigl((\hat{D}_{i}\phi_{1})^{\dagger}(\hat{D}^{i}\phi_{2})\bigr)(\hat{x},\hat{y})\in\mathcal{A}_{\theta_{\nc}}$,
\item[$\mathrm{(iii)}$] if $(\phi^{\dagger}_{1}\phi_{2}-\phi^{2}_{0})^{2}\in\mathcal{A}$ then $(\phi^{\dagger}_{1}\phi_{2}-\phi^{2}_{0})^{2}\in\mathcal{A}_{\theta_{\nc}}$.
\end{itemize}
\end{thm}
\begin{proof}
For the terms $\frac{1}{4}F_{ij}F^{ij}$ and $\bigl((\hat{D}_{i}\phi_{1})^{\dagger}(\hat{D}^{i}\phi_{2})\bigr)(\hat{x},\hat{y})$ it's sufficient to note respectively that \mbox{$[\tilde{A}_{i},\tilde{A}_{j}](\hat{x}+\hat{a}_{1},\hat{y})$} \mbox{$\Bigl([\tilde{A}_{i},\tilde{A}_{j}](\hat{x},\hat{y}+\hat{a}_{2})\Bigr)$} and \mbox{$\bigl((\hat{D}_{i}\phi_{1})^{\dagger}(\hat{D}^{i}\phi_{2})\bigr)(\hat{x}+\hat{a}_{1},\hat{y})$} \mbox{$\Bigl(\bigl((\hat{D}_{i}\phi_{1})^{\dagger}(\hat{D}^{i}\phi_{2})\bigr)(\hat{x},\hat{y}+\hat{a}_{2})\Bigr)$} can only be equal to \mbox{$[\tilde{A}_{i},\tilde{A}_{j}](\hat{x},\hat{y})$} and \mbox{$\bigl((\hat{D}_{i}\phi_{1})^{\dagger}(\hat{D}^{i}\phi_{2})\bigr)(\hat{x},\hat{y})$} and not to \mbox{$[\tilde{A}_{i},\tilde{A}_{j}](\hat{x}+\hat{a}'_{1},\hat{y})$} \mbox{$\Bigl([\tilde{A}_{i},\tilde{A}_{j}](\hat{x},\hat{y}+\hat{a}'_{2})\Bigr)$} and \mbox{$\bigl((\hat{D}_{i}\phi_{1})^{\dagger}(\hat{D}^{i}\phi_{2})\bigr)(\hat{x}+\hat{a}'_{1},\hat{y})$} \mbox{$\Bigl(\bigl((\hat{D}_{i}\phi_{1})^{\dagger}(\hat{D}^{i}\phi_{2})\bigr)(\hat{x},\hat{y}+\hat{a}'_{2})\Bigr)$} with \mbox{$(a'_{1},a'_{2})\neq (a_{1},a_{2})$} periods \mbox{related to some $\mathcal{T}'^{2}_{\nc}\neq\mathcal{T}^{2}_{\nc}$}.\\
For the term \mbox{$(\phi^{\dagger}_{1}\phi_{2}-\phi^{2}_{0})^{2}$} it's sufficient to consider the Theorem~\ref{thm:AiffAth}.
\end{proof}
To the Theorems~\ref{thm:3bper} and~\ref{thm:3inAth} it follows that
\begin{cor}
\label{cor:3bpand3inAth}
In the hypotheses above, \\
$\frac{1}{4}F_{ij}F^{ij}, \bigl((\hat{D}_{i}\phi_{1})^{\dagger}(\hat{D}^{i}\phi_{2})\bigr)(\hat{x},\hat{y}),(\phi^{\dagger}_{1}\phi_{2}-\phi^{2}_{0})^{2}\in\mathcal{A} \Rightarrow\mathcal{E}_{\spt{12}}\in\mathcal{A}_{\theta_{\nc}}$.
\end{cor}
\begin{proof}
It's enough to notice that when the three quantities that make $\mathcal{E}_{\spt{12}}$ belong contemporarily to $\mathcal{A}$ (i.e.\,for \mbox{$(\phi_{i},\tilde{\phi}_{j;k})\in\mathfrak{F}_{r_{\spt{\theta}}}$} (see Theorem~\ref{thm:3bper})), they belong in particular to $\mathcal{A}_{\theta_{\nc}}$ (see Theorem~\ref{thm:3inAth}). 
\end{proof}
Thus, starting from $(\mathcal{B}^{*\nc }_{\spt{U(1)}},\rho)$ and $\mathscr{F}\times\mathscr{F}$, we have that $\mathcal{E}$ has $\mathcal{T}^{2}_{\nc}$ as domain. In the refs.\,\citep{Forgacs2005,Lozano:2006xn}, instead, the Corollary~\ref{cor:3bpand3inAth} does not hold, because the belonging to $\mathcal{A}$ of the three addends does not imply the belonging of $\mathcal{E}$ to $\mathcal{A}_{\theta_{\nc}}$. In fact, they have that while $\frac{1}{4}F_{ij}F^{ij}$ belongs to $\tilde{\mathcal{A}}_{\theta_{\nc}}$ related to $\tilde{\mathcal{T}}^{2}_{\nc}$, \mbox{$\bigl((\hat{D}_{i}\phi)^{\dagger}(\hat{D}^{i}\phi)\bigr)(\hat{x},\hat{y})$} and $(\phi^{\dagger}\phi-\phi^{2}_{0})^{2}$ belong to $\mathcal{A}_{\theta_{\nc}}$ so that \mbox{$\mathcal{E}\notin\mathcal{A}$}. \\
Thus, differently from the refs.\,\citep{Forgacs2005,Lozano:2006xn}, we can define the following pair on which \mbox{$\mathcal{E}\in\mathcal{A}_{\theta_{\nc}}$} is defined: 
\begin{defask}
\label{def:B}
$(\mathcal{B}^{\nc}_{\spt{\mathfrak{F}_{r_{\spt{\theta}}}}},\rho)=(\mathcal{B}^{*\nc }_{\spt{U(1)}},\rho)$  with $(\phi_{1},\phi_{2})\in\mathscr{F}_{r_{\spt{\theta}}}\times\mathscr{F}_{r_{\spt{\theta}}}$ and  
\begin{equation}
\label{eqn:solAB}
A_{i}(\hat{x},\hat{y})=\tilde{A}_{i}(\hat{x},\hat{y})+\varepsilon_{ij}\frac{1}{g a_{j}}\bm{\hat{x}}_{j}
\end{equation}
\hspace{-.02em}where $\tilde{A}_{i}(\hat{x},\hat{y})=c_{i}\tilde{\phi}_{2;i}\tilde{\phi}^{\dagger}_{1;i}$ with constant $c_{i}$ and \mbox{$(\phi_{i},\tilde{\phi}_{j;k})\in\mathfrak{F}_{r_{\spt{\theta}}}$}. 
\end{defask}

\subsubsection{Evaluation \textit{à la} Bogomolny of $E$ on $\mathcal{T}^{2}_{\nc}$} 
\label{ssec:domE}
\noindent We will evaluate $E_{\spt{12}}$ at the point of Bogomolny for minimum-energy configurations rewriting it \textit{à la} Bogomolny.\\
\indent Hereinafter, with $k_{\spt{\alpha_{1}}}, k_{\spt{\beta_{2}}}$ $(k_{\spt{\alpha_{2}}}, k_{\spt{\beta_{1}}})$ we will indicate respectively the coefficient of proportionality between $2\alpha_{1}(a_{1}), 2\beta_{2}(a_{2})$ $(\alpha_{2}(a_{1})/3, \beta_{1}(a_{2})/3)$ and $r_{\spt{\theta}}$.\\
\indent Let's consider the following
\begin{lem}
\label{lem:p2p1d=p1p2d1v}
\textsc{(hermitianity)}\\
Given $(\mathcal{B}^{*\nc }_{\spt{U(1)}},\rho)$, considering $(\phi_{1},\phi_{2})\in\mathscr{F}_{r_{\spt{\theta}}}\times\mathscr{F}_{r_{\spt{\theta}}}$, we have that
\begin{center}
$\phi^{\dagger}_{1}\phi_{2}=\phi^{\dagger}_{2}\phi_{1}$ or $\phi_{2}\phi^{\dagger}_{1}=\phi_{1}\phi^{\dagger}_{2}$ if{}f $\phi_{1}=\phi_{2}$.
\end{center}
\end{lem}
\begin{proof}
It's sufficient to consider that the terms \mbox{$g^{\spt{\omega};\,\spt{1,2}}_{q,q',k,k'}(\hat{x},\hat{y},\theta\hat{\mathbbm{1}})$} and \mbox{$h^{\spt{\omega};\,\spt{2,1}}_{q,q',k,k'}(\hat{x},\hat{y},\theta\hat{\mathbbm{1}})$} (see eqs.\,\eqref{eqn:g12h21}) are such that
\begin{subequations}
\label{eqn:phi1dphi2vphi2phi1d}
\begin{align}
g^{\spt{\omega};\,\spt{1,2}}_{q,q',k,k'}(\hat{x},\hat{y})&=g^{\spt{\omega};\,\spt{2,1}}_{q,q',k,k'}(\hat{x},\hat{y}),\\
h^{\spt{\omega};\,\spt{2,1}}_{q,q',k,k'}(\hat{x},\hat{y})&=h^{\spt{\omega};\,\spt{1,2}}_{q,q',k,k'}(\hat{x},\hat{y})
\end{align}
\end{subequations}
if{}f $\phi_{1}=\phi_{2}$.
\end{proof}
\indent Let's also consider the following
\begin{lem}
\label{lem:p1dp2=p2p1d}
Given $(\mathcal{B}^{*\nc }_{\spt{U(1)}},\rho)$, considering $(\phi_{1},\phi_{2})\in\mathscr{F}_{r_{\spt{\theta}}}\times\mathscr{F}_{r_{\spt{\theta}}}$, we have that 
\begin{center}
$\phi^{\dagger}_{1}\phi_{2}=\phi_{2}\phi^{\dagger}_{1}$ if{}f \mbox{$(\phi_{1},\phi_{2})\in \tilde{\mathfrak{F}}_{r_{\spt{\theta}};1}$}.
\end{center} 
\end{lem} 
\begin{proof}
It's sufficient to consider the terms $g^{\spt{\omega};\,\spt{1,2}}_{q,q',k,k'}(\hat{x},\hat{y})$ and $h^{\spt{\omega};\,\spt{2,1}}_{q,q',k,k'}(\hat{x},\hat{y})$ respectively of the development of $\phi^{\dagger}_{1}\phi_{2}$ and $\phi_{2}\phi^{\dagger}_{1}$ (see eqs.\,\eqref{eqn:g12h21}) and notice that
\begin{equation}
\label{eqn:h=gIw}
h^{\spt{\omega};\,\spt{2,1}}_{q,q',k,k'}(\hat{x},\hat{y})=e^{\Delta(\hat{z}_{i;j},\hat{\tau}_{i;j})}g^{\spt{\omega};\,\spt{1,2}}_{q,q',k,k'}(\hat{x},\hat{y})
\end{equation}
with 
{\fontsize{9}{21}
 \selectfont
\begin{equation}
\label{eqn:Delta(h,g)Iw}
\begin{split}
\hspace{-.25cm}\Delta(\hat{z}_{i;j},\hat{\tau}_{i;j})=&2\pi(q'k^{2}[\hat{\tau}_{2;2},\hat{z}_{1;1}]-qk'^{2}[\hat{\tau}_{2;1},\hat{z}_{1;2},])+\pi^{2}(q^{2}k'^{2}[\hat{\tau}_{1;2},\hat{\tau}_{2;1}]-q'^{2}k^{2}[\hat{\tau}_{1;1},\hat{\tau}_{2;2}])+\\
\hspace{-.25cm}&+4(qk'[\hat{z}_{1;2},\hat{z}_{2;1}]-q'k[\hat{z}_{1;1},\hat{z}_{2;2}])+2\pi(q^{2}k'[\hat{\tau}_{1;2},\hat{z}_{2;1}]-q'^{2}k[\hat{\tau}_{1;1},\hat{z}_{2;2}])+\\
\hspace{-.25cm}&+2\pi (q^{2}q'[\hat{\tau}_{1;2},\hat{z}_{1;1}]-q'^{2}q[\hat{\tau}_{1;1},\hat{z}_{1;2}])+2\pi (k'k^{2}[\hat{\tau}_{2;2},\hat{z}_{2;1}]-k'^{2}k[\hat{\tau}_{2;1},\hat{z}_{2;2}])+\\
\hspace{-.25cm}&-4q'q[\hat{z}_{1;1},\hat{z}_{1;2}]-4k'k[\hat{z}_{2;1},\hat{z}_{2;2}]-\pi^{2}q'^{2}q^{2}[\hat{\tau}_{1;1},\hat{\tau}_{1;2}]-\pi^{2}k'^{2}k^{2}[\hat{\tau}_{2;1},\hat{\tau}_{2;2}],
\end{split}
\end{equation}}
\hspace{-.3em}from which we have to consider that, for $i=k$, \mbox{$[\hat{\tau}_{i;j},\hat{\tau}_{k;l}]=0$}. We have thus the equality of $h^{\spt{\omega};\,\spt{2,1}}_{q,q',k,k'}(\hat{x},\hat{y})$ with $g^{\spt{\omega};\,\spt{1,2}}_{q,q',k,k'}(\hat{x},\hat{y})$ if{}f\footnote{The request that $[\hat{z}_{i;j},\hat{z}_{k;l}]=0$ implies that $[\hat{\tau}_{i;j},\hat{z}_{i;k}]=[\hat{\tau}_{i;k},\hat{z}_{i;j}]$ and the equality to $e$ of the terms $\e^{\pi[\hat{\tau}_{i;j},\hat{z}_{k;l}]}$ with $i\neq k$: if \mbox{$[\hat{z}_{i;j},\hat{z}_{k;l}]=0$} then, for $i\neq k$, $\e^{\pi[\hat{\tau}_{i;j},\hat{z}_{k;l}]}=\e^{\pi[\hat{\tau}'_{k;l},\hat{z}_{i;j}]}$ but, from the Lemma~\ref{lem:e^[x,y]=+-e} with the hypothesis that $\phi_{i}\in\mathscr{F}_{r_{\spt{\theta}}}$, we have that $\e^{\pi[\hat{\tau}'_{k;l},\hat{z}_{i;j}]}$ is equal to $e$.\label{fn:e[t1,t2]=e}}  \mbox{$(\phi_{1},\phi_{2})\in \tilde{\mathfrak{F}}_{r_{\spt{\theta}};1}$}.
\end{proof}
\begin{remark}
It's interesting to notice that the equality between $\phi_{1}^{\dagger}\phi_{2}$ and $\phi_{2}\phi_{1}^{\dagger}$ though granted for $\mathcal{T}^{2}_{\cc}$, it is not for $\mathcal{T}^{2}_{\nc}$: starting from $\mathscr{F}\times\mathscr{F}$, while for $\mathcal{T}^{2}_{\cc}$ 
the equality is obvious, for $\mathcal{T}^{2}_{\nc}$ we can see it is valid only if \mbox{$(\phi_{1},\phi_{2})\in \tilde{\mathfrak{F}}_{r_{\spt{\theta}};1}$}.
\end{remark}
Let's consider now the following two Lemmas
\begin{lem}
\label{lem:EalàBogo'}
Given $(\mathcal{B}^{\nc}_{\spt{\mathfrak{F}_{r_{\spt{\theta}}}}},\rho)$, 
\begin{center}
\hspace{-.8cm}$E_{\spt{12}}$ can be rewritten \textit{à la} Bogomolny if{}f $\tilde{\phi}_{2;i}=\tilde{\phi}_{1;i}$ and \mbox{$(\phi_{1},\phi_{2})\in \tilde{\mathfrak{F}}_{r_{\spt{\theta}};1}$}.
\end{center}
\end{lem}

\begin{proof}
The hypotheses of the Theorem~\ref{thm:EBogorev} are met if{}f $\tilde{\phi}_{2;i}=\tilde{\phi}_{1;i}$ and \mbox{$(\phi_{1},\phi_{2})\in\tilde{\mathfrak{F}}_{r_{\spt{\theta}};1}$}. In fact $A_{i}$ is hermitian if{}f \mbox{$\tilde{\phi}_{2;i}=\tilde{\phi}_{1;i}$} (see Lemma~\ref{lem:p2p1d=p1p2d1v}) while \mbox{$\tr(\phi^{\dagger}_{1}\phi_{2})=\tr(\phi_{2}\phi^{\dagger}_{1})$} and  \mbox{$\tr(\phi^{\dagger}_{1}\phi_{2})^{2}=\tr(\phi_{2}\phi^{\dagger}_{1})^{2}$} if{}f  \mbox{$\phi^{\dagger}_{1}\phi_{2},\phi_{2}\phi^{\dagger}_{1}\in\mathcal{A}$} and so if{}f \mbox{$(\phi_{1},\phi_{2})\in\tilde{\mathfrak{F}}_{r_{\spt{\theta}};1}$} (see Theorem~\ref{thm:O12O21period}): once \mbox{$\phi^{\dagger}_{1}\phi_{2},\phi_{2}\phi^{\dagger}_{1}\in\mathcal{A}$} and that thus \mbox{$\phi^{\dagger}_{1}\phi_{2},\phi_{2}\phi^{\dagger}_{1}\in\mathcal{A}_{\theta}$} (see Theorem~\ref{thm:AiffAth}), we have that \mbox{$\tr(\phi^{\dagger}_{1}\phi_{2})=\tr(\phi_{2}\phi^{\dagger}_{1})$} and \mbox{$\tr(\phi^{\dagger}_{1}\phi_{2})^{2}=\tr(\phi_{2}\phi^{\dagger}_{1})^{2}$} (see Lemma~\ref{lem:p1dp2=p2p1d}).
\end{proof}

And moreover  
\begin{lem}
\label{lem:Trphid1phi2}
Given $(\mathcal{B}^{*\nc }_{\spt{U(1)}},\rho)$ with \mbox{$(\phi_{1},\phi_{2})\in\mathscr{F}_{r_{\spt{\theta}}}\times\mathscr{F}_{r_{\spt{\theta}}}$} and $\phi^{\dagger}_{1}\phi_{2}$ hermitian, we have that 
\[
\tr(\phi^{\dagger}_{1}\phi_{2})=0.
\]
\end{lem}

\begin{proof}
Let's notice that, according to the hypotheses, \mbox{$\phi^{\dagger}_{1}\phi_{2}\in\mathcal{A}_{\theta}$} (see the Lemma~\ref{lem:p2p1d=p1p2d1v} and the case\emph{~\ref{item:(i)}} of the Theorem~\ref{thm:O12O21period}).\\ 
Indicating $\phi_{1}$ with $\phi$ (see the Lemma~\ref{lem:p2p1d=p1p2d1v}) and defining  \mbox{$l=q'+q$}, \mbox{$l'=q'-q$}, \mbox{$m=k'+k$}, \mbox{$m'=k'-k$}, with $q',k'$ indexes referred to $\phi^{\dagger}$, we find that
\begin{equation}
\label{eqn:phidphig'lm}
\phi^{\dagger}\phi(\hat{x},\hat{y})=\hspace{-1em}\sum^{+\infty}_{l,l',m,m'=-\infty}\hspace{-1em}g'^{\spt{\omega}}_{l,l',m,m'}(\hat{x},\hat{y})
\end{equation}
where
\begin{equation}
\label{eqn:g'lm}
\hspace{-.5cm}g'^{\spt{\omega}}_{l,l',m,m'}(\hat{x},\hat{y})=(-)^{l+m}\e^{I^{\spt{g'}}_{l,l',m,m'}(\hat{z}_{1},\hat{z}_{2},\hat{\tau}_{1},\hat{\tau}_{2})}\e^{\omega'_{l,l',m,m'}(x,y)}\e^{\iu\pi mm'\hat{\tau}_{2}+\iu\pi ll'\hat{\tau}_{1}+2\iu l'\hat{z}_{1}+2\iu m'\hat{z}_{2}}
\end{equation}
with
{\fontsize{9.5}{15}
 \selectfont
\begin{equation}
\begin{split}
\hspace{-.85cm}I^{\spt{g'}}_{l,l',m,m'}(\hat{z}_{1},\hat{z}_{2},\hat{\tau}_{1},\hat{\tau}_{2})=&-\frac{\pi}{2}l'(m^{2}+m'^{2})[\hat{\tau}_{2},\hat{z}_{1}]+\frac{\pi^{2}}{4}ll'(m^{2}+m'^{2})[\hat{\tau}_{1},\hat{\tau}_{2}]+2l'(l+m)[\hat{z}_{1},\hat{z}_{2}]+\\
&+\pi l'l(l+m)[\hat{\tau}_{1},\hat{z}_{2}]-\frac{\pi}{3}l'^{3}[\hat{\tau}_{1},\hat{z}_{1}]-\frac{\pi}{3}m'^{3}[\hat{\tau}_{2},\hat{z}_{2}],
\end{split}
\end{equation}}
\hspace{-.3em}and \mbox{$\omega'_{l,l',m,m'}(x,y)=\omega'_{l,l'}(x/a_{1},\omega_{1}(a_{1}))+\omega'_{m,m'}(y/a_{2},\omega_{2}(a_{2}))$} with
{\fontsize{10}{15}
 \selectfont
\begin{equation}
\hspace{-.7cm}\omega'_{l,l'}(x/a_{1},\omega_{1}(a_{1}))= 2\pi\Im(\omega_{1}(a_{1}))\biggl(\Bigr(\,\frac{x}{a_{1}}+l\,\Bigr)\frac{x}{a_{1}}+\frac{(l'^{2}+l^{2})}{4}\biggr)+2\pi\iu\Re(\omega_{1}(a_{1}))l'\Bigl( \frac{x}{a_{1}}+\frac{l}{2}\Bigr).
\end{equation}}
\indent Now, according to the hypothesis, $\alpha_{i}$'s and $\beta_{i}$'s meet the $\mathrm{(}\mathscr{D}^{*}_{{\spt{\beta\alpha}},1}\mathrm{)}$ so that \mbox{$[\hat{z}_{1},\hat{z}_{2}]=0$} and then, considering what we stated in the footnote~\ref{fn:e[t1,t2]=e}, \mbox{we have that}
\begin{subequations}
\label{eqn:conDab}
\begin{align}
&\e^{-\frac{\pi}{2}l'(m^{2}+m'^{2})[\hat{\tau}_{2},\hat{z}_{1}]}=(\e^{\pi [\hat{\tau}'_{1},\hat{z}_{2}]})^{-\frac{l'(m^{2}+m'^{2})}{2}}=e\\
&\e^{\frac{\pi^{2}}{4}ll'(m^{2}+m'^{2})[\hat{\tau}_{1},\hat{\tau}_{2}]}=(\e^{\pi [\hat{\tau}'_{2},\hat{\tau}'_{1}]})^{\pi\frac{ll'(m^{2}+m'^{2})}{4}}=e,\\
&\e^{\pi l'l(l+m)[\hat{\tau}_{1},\hat{z}_{2}]}=(\e^{\pi [\hat{\tau}'_{2},\hat{z}_{1}]})^{l'l(1+m)}=e.
\end{align} 
\end{subequations}
Taking into account~\eqref{eqn:conDab}, we have that  
{\fontsize{10}{15}
 \selectfont
\begin{equation}
\hspace{-1.5cm}\rho(g'^{\spt{\omega}}_{l,l',m,m'}(\hat{x},\hat{y}))=G'^{\spt{\omega}}_{l,l',m'}(x/a_{1},\theta,k_{\spt{\alpha_{1}}},k_{\spt{\alpha_{2}}},k_{\spt{\beta_{2}}},\omega_{1}(a_{1}))G'^{\spt{\omega}}_{m,m',l'}(y/a_{2},\theta,k_{\spt{\beta_{2}}},k_{\spt{\beta_{1}}},-k_{\spt{\alpha_{1}}},\omega_{2}(a_{2}))
\end{equation}}
\!\!with
{\fontsize{9.5}{15}
 \selectfont
\begin{equation}
\hspace{-2.5cm}G'^{\spt{\omega}}_{l,l',m'}(x/a_{1},\theta,k_{\spt{\alpha_{1}}},k_{\spt{\alpha_{2}}},k_{\spt{\beta_{2}}},\omega_{1}(a_{1}))=(-)^{l}\e^{\Re(\omega'_{l,l'}(x/a_{1},\omega_{1}(a_{1})))}\e^{\iu \frac{\pi}{2} ll'\bigl(k_{\spt{\alpha_{1}}}r_{\spt{\theta}}+2\Re(\omega_{1}(a_{1}))\bigr)}\e^{\iu K^{\spt{\omega}}_{l',m'}(k_{\spt{\alpha_{1}}},k_{\spt{\alpha_{2}}},k_{\spt{\beta_{2}}},\omega_{1}(a_{1}))\frac{x}{a_{1}}}
\end{equation}}
\!\!where 
{\fontsize{9.25}{15}
 \selectfont
\begin{equation}
K^{\spt{\omega}}_{l',m'}(k_{\spt{\alpha_{1}}},k_{\spt{\alpha_{2}}},k_{\spt{\beta_{2}}},\omega_{1}(a_{1}))=\pi\Bigl(2l'\Re(\omega_{1}(a_{1}))+\frac{\pi}{2}m'^{3}k_{\spt{\alpha_{2}}}k_{\spt{\beta_{2}}}r^{2}_{\spt{\theta}}\theta +(l'k_{\spt{\alpha_{1}}}+6m'k_{\spt{\alpha_{2}}})r_{\spt{\theta}}\Bigr).
\end{equation}}
\indent Now, considering that \mbox{$\rho(\phi^{\dagger}_{1}\phi_{2}(\hat{x},\hat{y}))$} is a continuos function (see Remark~\ref{rmk:Imo<0}), we can invoke Fubini's Theorem and conclude that \mbox{$\int_{\spt{\mathcal{T}^{2}}}\rho(\phi^{\dagger}_{1}\phi_{2}(\hat{x},\hat{y}))$} is null if \mbox{$\int^{a_{1}}_{0}\rho_{y}(\phi^{\dagger}_{1}\phi_{2})(x)\,dx$} (\mbox{$\int^{a_{2}}_{0}\rho_{x}(\phi^{\dagger}_{1}\phi_{2})(y)\,dy$}) is null, where \mbox{$\rho_{y}(\phi^{\dagger}_{1}\phi_{2})(x)=\rho(\phi^{\dagger}_{1}\phi_{2}(\hat{x},\hat{y}))$} (\mbox{$\rho_{x}(\phi^{\dagger}_{1}\phi_{2})(y)=\rho(\phi^{\dagger}_{1}\phi_{2}(\hat{x},\hat{y}))$}). We will show thus that \mbox{$\int^{a_{1}}_{0}\rho_{y}(\phi^{\dagger}_{1}\phi_{2})(x)\,dx$} is null (analogous demonstration holds for \mbox{$\int^{a_{2}}_{0}\rho_{x}(\phi^{\dagger}_{1}\phi_{2})(y)\,dy$}). Considering that \mbox{$\rho(\phi^{\dagger}_{1}\phi_{2}(\hat{x},\hat{y}))$} converges uniformly (see Remark~\ref{rmk:Imo<0}), we can integrate term to term and have that
{\fontsize{9.5}{12.5}
 \selectfont
\begin{equation}
\label{eqn:Intrhok(x)}
\hspace{-2cm}\int^{a_{1}}_{0}\rho_{y}(\phi^{\dagger}_{1}\phi_{2})(x)\,dx=\hspace{-.45cm}\sum^{+\infty}_{l,l',m,m'=-\infty}\hspace{-.45cm}G'^{\spt{\omega}}_{m,m',l'}(y/a_{2},\theta,k_{\spt{\beta_{2}}},k_{\spt{\beta_{1}}},-k_{\spt{\alpha_{1}}},\omega_{2}(a_{2}))\int^{a_{1}}_{0}G'^{\spt{\omega}}_{l,l',m'}(x/a_{1},\theta,k_{\spt{\alpha_{1}}},k_{\spt{\alpha_{2}}},k_{\spt{\beta_{2}}},\omega_{1}(a_{1}))dx
\end{equation}}
\!\!where
{\fontsize{8.5}{20}
 \selectfont 
\begin{equation}
\label{eqn:Inta1G'}
\hspace{-2cm}\int^{a_{1}}_{0}G'^{\spt{\omega}}_{l,l',m'}(x/a_{1},\theta,k_{\spt{\alpha_{1}}},k_{\spt{\alpha_{2}}},k_{\spt{\beta_{2}}},\omega_{1}(a_{1}))=A_{l,l',m'}\biggl(\erf \Bigl(\frac{2\pi\, l\Im(\omega_{1}(a_{1}))+\iu K^{\spt{\omega}}_{l',m'}}{2\sqrt{2\pi\abs{\Im(\omega_{1}(a_{1}))}}}\Bigr)-\erf\Bigl(\frac{2\pi (l+2)\Im(\omega_{1}(a_{1}))+\iu K^{\spt{\omega}}_{l',m'}}{2\sqrt{2\pi\abs{\Im(\omega_{1}(a_{1}))}}}\Bigr)\biggr)
\end{equation}}
\!\!with 
\begin{equation}
\label{eqn:defAll1m1}
\hspace{-.5cm}A_{l,l',m'}=\frac{(-)^{l}a_{1}}{2\sqrt{2 \abs{\Im(\omega_{1}(a_{1}))}}}\e^{\frac{(K^{\spt{\omega}}_{l',m'})^{2}}{8\pi\Im(\omega_{1}(a_{1}))}+\frac{\pi}{2}l'^{2}\Im(\omega_{1}(a_{1}))}\e^{\iu \frac{  l}{2}(\pi l'k_{\spt{\alpha_{1}}}r_{\spt{\theta}}+2\pi l'\Re(\omega_{1}(a_{1}))-K^{\spt{\omega ;x}}_{l',m'})}. 
\end{equation}
Considering that \mbox{$\e^{\frac{\pi^{2}}{3}m'^{3}[\hat{\tau}_{2},\hat{\tau}'_{2}]+6\pi m'k_{\spt{\alpha_{2}}}r_{\spt{\theta}}}$} is equal to $e$ (see Lemma~\ref{lem:e^[x,y]=+-e} and Corollary~\ref{cor:rth}), we can easily get that \mbox{$\int^{a_{1}}_{0}\rho_{y}(\phi^{\dagger}_{1}\phi_{2})(x)\,dx$} is null.
\end{proof}
We are ready to evaluate $E_{\spt{12}}$.
\begin{thm}
\label{thm:E12Bogo}
Given $(\mathcal{B}^{\nc}_{\spt{\mathfrak{F}_{r_{\spt{\theta}}}}},\rho)$ with $A_{i}$ and $\phi^{\dagger}_{1}\phi_{2}$ hermitian and \mbox{$(\phi_{1},\phi_{2})\in \tilde{\mathfrak{F}}_{r_{\spt{\theta}};1}$}, at the point of Bogomolny \mbox{\emph{(}i.e.\,$\lambda=g^{2}/2$\emph{)}} we have that
\begin{equation}
\label{eqn:ElimitR}
E_{\spt{12}}=\frac{g^{2}\phi^{4}_{0}}{2}a_{1}a_{2}+\frac{\theta^{2}_{\cc}}{2g^{2}r^{2}_{\spt{\theta}}a_{1}a_{2}}.
\end{equation}
\end{thm}
\begin{proof}
The~\eqref{eqn:ElimitR} is obtained rewriting $E_{\spt{12}}$ \textit{à la} Bogomolny (see Lemma~\ref{lem:EalàBogo'}) that, taking into account that here $[\tilde{A}_{i},\tilde{A}_{j}](\hat{x},\hat{y})=0$ and $\phi^{\dagger}_{1}\phi_{2}=\phi_{2}\phi^{\dagger}_{1}$, takes the following form
 {\fontsize{9.5}{12,5}
 \selectfont 
\begin{equation} 
\label{eqn:erg'BRev}
\begin{split}
\hspace{-1.75cm}E_{\spt{12}}=\tr\biggl(\frac{1}{2}(\hat{D}_{i}\phi_{1}-\iu\varepsilon_{ij}\hat{D}_{j}\phi_{1})^{\dagger}(\hat{D}_{i}\phi_{2}-\iu\varepsilon_{ij}\hat{D}_{j}\phi_{2})&+\frac{g^{2}}{2}(\phi_{2}\phi^{\dagger}_{1})^{2}-\Bigl(g^{2}\phi^{2}_{0}-\frac{\theta}{a_{1}a_{2}}\Bigr)\phi_{2}\phi^{\dagger}_{1}+\frac{g^{2}\phi^{4}_{0}}{2}+\frac{\theta^{2}}{2g^{2}a^{2}_{1}a^{2}_{2}}+\\
&+\frac{\iu}{2}\varepsilon_{ij}\bigl(\hat{\partial}_{i}(\phi^{\dagger}_{1}\hat{D}_{j}\phi_{2})-\hat{\partial}_{j}(\phi^{\dagger}_{1}\hat{D}_{i}\phi_{2})\bigr)\biggr).
\end{split} 
\end{equation}}
\,Now, $\tr(\phi_{2}\phi^{\dagger}_{1})$ is null (see Lemma~\ref{lem:p1dp2=p2p1d} and~\ref{lem:Trphid1phi2}) and we can show that \mbox{$\tr\bigl(\hat{\partial}_{i}(\phi^{\dagger}_{1}\hat{D}_{j}\phi_{2})-\hat{\partial}_{j}(\phi^{\dagger}_{1}\hat{D}_{i}\phi_{2})\bigr)$}, \mbox{$\tr\bigl((\hat{D}_{i}\phi_{1})^{\dagger}(\hat{D}_{i}\phi_{2})\bigr)$}, $\tr\bigl(\varepsilon_{ij}(\hat{D}_{i}\phi_{1})^{\dagger}(\hat{D}_{j}\phi_{2} )\bigr)$ and \mbox{$\tr((\phi_{2}\phi^{\dagger}_{1})^{2})$} are null (the detailed calculation of $E_{\spt{12}}$ is given in the Appendix~\ref{sec:AppErg}).
\end{proof}
\begin{remark}
\label{rmk:noBPS}
We can notice how here, differently from~\citep{Manton2004,Arroyo04,Forgacs2005,Lozano:2006xn,EWeinberg2012}, the quantities \mbox{$\tr\bigl((\hat{D}_{i}\phi_{1})^{\dagger}(\hat{D}_{i}\phi_{2})\bigr)$}, $\tr\bigl(\varepsilon_{ij}(\hat{D}_{i}\phi_{1})^{\dagger}(\hat{D}_{j}\phi_{2} )\bigr)$, \mbox{$\tr((\phi_{2}\phi^{\dagger}_{1})^{2})$} and \mbox{$\tr(\phi_{2}\phi^{\dagger}_{1})$} cannot be but null. Therefore in this context, to find minimum-energy configurations, it is not necessary to resolve a system of two equations like the following\footnote{We can notice how the system~\eqref{eqn:Bogomod} is not the usual one of the two BPS eqs because while the~\eqref{eqn:B1} is the I BPS eq, the~\eqref{eqn:B2} is, instead, different from II BPS eq.}(see eq.\,\eqref{eqn:erg'BRev})
\begin{subequations}
\label{eqn:Bogomod}
\begin{align}
&\hat{D}_{i}\phi-\iu\varepsilon_{ij}\hat{D}_{j}\phi=0,\label{eqn:B1}\\
&\frac{g^{2}}{2}(\phi\phi^{\dagger})^{2}-\Bigl(g^{2}\phi^{2}_{0}-\frac{\theta}{a_{1}a_{2}}\Bigr)\phi\phi^{\dagger}=0.\label{eqn:B2}
\end{align} 
\end{subequations}
\end{remark}

\subsection{Influence of $\theta$ on the observables}
\label{ssec:FandM(th)}
\noindent We will proceed now with the analysis of the influence of $\theta$ on the observables.\\
We start noting that no $\theta$-depending scaled gauge charge $\tilde{g}$ is expected here. In the refs.\,\citep{Forgacs2005,Lozano:2006xn}, considering $A^{\spt{\Omega}}_{i}$, i.e.\,a generic gauge transformation $\Omega$ of $A_{i}$, it was noted how such quantity was equal to the sum of $\tilde{A}^{\spt{\Omega}}_{i}$ characterized by $\tilde{g}$ and by an untrasformed linear field in the coordinates. Thus, it was inferred that a gauge theory on a non-commutative torus, with non trivial boundary conditions as the~\eqref{eqn:bordH} and the~\eqref{eqn:bordUA}, was equivalent to a gauge theory on $\tilde{\mathcal{T}}^{2}$ with periodic boundary conditions (in that context $\tilde{A}_{i}$ resulted periodic on $\tilde{\mathcal{T}}^{2}$) and with a gauge charge equal to $\tilde{g}$.\\
Here instead, considering the $A_{i}$ defined by~\eqref{eqn:solA}, we obtain that  
\begin{equation}
A^{\spt{\Omega}}_{i}(\hat{x},\hat{y})=\Omega(\hat{x},\hat{y})\tilde{A}_{i}(\hat{x},\hat{y})\Omega^{-1}(\hat{x},\hat{y})+\varepsilon_{ij}\frac{1}{g a_{j}}\bm{\hat{x}}_{j}.
\end{equation} 
We have thus a $\tilde{A}^{\spt{\Omega}}_{i}$ with a $\tilde{g}\rightarrow\infty$ that is not equal to any $\theta$-depending scaled gauge charge. As concerns the above equivalence (displayed in eq.\,(2.24) of the  ref.\,\citep{Forgacs2005} and in eq.\,(3.18) of the ref.\,\citep{Lozano:2006xn}), we have to admit that it cannot hold since \emph{a gauge theory with periodic boundary conditions implies that} \mbox{$\bigl((\hat{D}_{i}\phi_{1})^{\dagger}(\hat{D}^{i}\phi_{2})\bigr)(\hat{x},\hat{y})\notin\mathcal{A}_{\theta_{\nc}}$}.

We will deal now with the flux of the magnetic field $\mathcal{F}$, i.e.\,\mbox{$\mathcal{F}=\tr F_{12}$}, for generic $n$. In this regard, we consider the following subset of $\mathscr{F}_{r_{\spt{\theta}}}$ 
\begin{defask}
$\!\mathscr{F}_{r_{\spt{\theta}};n}=\{\phi\colon \phi\in\mathscr{F}_{r_{\spt{\theta}}}$ such that \mbox{$\alpha_{i}(x)=n \bar{\alpha}_{i}(x)$} and \mbox{$\beta_{i}(y)=n\bar{\beta}_{i}(y)$} with $\bar{\alpha}_{i}(x),\bar{\beta}_{i}(y)$'s relative to some $\bar{\phi}\in\mathscr{F}_{r_{\spt{\theta}}}\}$.
\end{defask} 
We can easily realize that, calling $\bar{\Omega}_{i}$ the $\Omega_{i}$'s relative to $\bar{\phi}$, the $\Omega'_{1}(\Omega'_{2})$ of $\phi\in\mathscr{F}_{r_{\spt{\theta}};n}$ is equal to $\bar{\Omega}^{n}_{1}(\bar{\Omega}^{n}_{2})$ modulo a gauge transformation. We show that 
\begin{thm}
\label{thm:F(theta)}
Given $(\mathcal{B}^{\nc}_{\spt{\mathfrak{F}_{r_{\spt{\theta}}}}},\rho)$, we have that


\begin{equation}
\label{eqn:F(theta)}
\mathcal{F}=-\frac{1}{r_{\spt{\theta}}}\frac{\theta_{\cc}}{g}.
\end{equation}

\bigskip

\end{thm}
\begin{proof}
It's sufficient to consider the eq.\,\eqref{eqn:Fij} with the fact that by hypothesis \mbox{$[\tilde{A}_{i},\tilde{A}_{j}](\hat{x},\hat{y})=0$} and that $r_{\spt{\theta}}$ is by definition equal to $\theta_{\cc}/\theta_{\nc}$.
\end{proof}
In particular we have that the eq.\,\eqref{eqn:F(theta)} holds even if we consider the $\phi$ of $(\mathcal{B}^{\nc}_{\spt{\mathfrak{F}_{r_{\spt{\theta}}}}},\rho)$ belonging to $\mathscr{F}_{r_{\spt{\theta}};n}$. Thus, we can notice how \emph{here differently from the refs.}\,\citep{Baal82,Manton2004,Arroyo04,Forgacs2005,Lozano:2006xn,EWeinberg2012} \emph{there is no dependence of $\mathcal{F}$ on $n$} and so \emph{the quantization of the color-electric charge $q_{e}$ does not depend on $n$ but on $r_{\spt{\theta}}$} that belongs to $\Q$ (see Definition~\ref{def:B}).
\begin{remark}
It's interesting to notice that, if instead of the~\eqref{eqn:solAB}, we put 
\begin{equation}
\label{eqn:solAn'}
A_{i}(\hat{x},\hat{y})=\tilde{A}_{i}(\hat{x},\hat{y})+\varepsilon_{ij}\frac{n}{g a_{j}}\bm{\hat{x}}_{j}
\end{equation}
with $\tilde{A}_{i}(\hat{x},\hat{y})=c_{i}\tilde{\phi}_{2;i}\tilde{\phi}^{\dagger}_{1;i}$ and $\tilde{\phi}_{j;i}\in\mathscr{F}_{r_{\spt{\theta}};n}$,
though we would obtain that $\mathcal{F}$ is dependent on $n$, i.e.\,\,\mbox{$\mathcal{F}=n(n-2)\theta/g$}, in any case the~\eqref{eqn:bordUA} would not be satisfied and so such ansatz could not be taken into consideration.
\end{remark}
We have thus that
\begin{corth*}
Given $(\mathcal{B}^{\nc}_{\spt{\mathfrak{F}_{r_{\spt{\theta}}}}},\rho)$, we have that
\begin{equation}
\label{eqn:qe=1/rththcc}
q_{e}=\frac{1}{r_{\spt{\theta}}g}\,.
\end{equation}
\end{corth*}
Now, we will deal with $E_{\spt{12}}$ (cf.\,eq.\,\eqref{eqn:ElimitR}). In this regard, we can notice how the Theorem~\ref{thm:E12Bogo} holds in particular also  $\phi_{i}\in\mathscr{F}_{r_{\spt{\theta}};n}$. \emph{We are thus in the presence of a lower limit of energy at the point of Bogomolny whose quantization, as for $q_{e}$, depends on $r_{\spt{\theta}}$ and it is such that  \mbox{$E\rightarrow+\infty$} for \mbox{$a_{1},a_{2} \rightarrow+\infty$}. In the refs}.\,\citep{Baal82,Manton2004,Arroyo04,Forgacs2005,Lozano:2006xn,EWeinberg2012}\emph{, instead, such limit depended exclusively on $n$}.\\
\indent We will analyze now the influence of $\theta$ on the masses of the particle spectrum. In this regard, we find that  
\begin{thm}
\label{thm:Higgsth}
Given $(\mathcal{B}^{\nc}_{\spt{\mathfrak{F}_{r_{\spt{\theta}}}}},\rho)$ with $\mathcal{L}$ defined by~\eqref{eqn:dL}, $\lambda>0$, $c_{i}>0$,
 $A_{i}$ and $\phi^{\dagger}_{1}\phi_{2}$ hermitian and \mbox{$(\phi_{1},\phi_{2})\in \tilde{\mathfrak{F}}_{r_{\spt{\theta}};0}$}, denoting by $H$ and $\Xi$ the fields of the perturbative expansion of $\phi^{\dagger}_{1}\phi_{2}$ around the vacuum, we have that
{\fontsize{8.8}{12,5}
 \selectfont
\begin{subequations}
\label{eqn:m(rth)CS}
\begin{align}
\hspace{-1.25cm}&\hspace{1.25cm}m_{\spt{H}}=\sqrt{\Bigl(\frac{1}{a^{2}_{1}}-2\mu^{2}\Bigr)\cos(\vartheta(r_{\spt{\theta}}))},\,m_{\spt{\Xi}}=\frac{1}{a_{1}}\sqrt{\sin(\vartheta(r_{\spt{\theta}}))},\qquad for\qquad 0 \leq\vartheta(r_{\spt{\theta}})<\frac{\pi}{2}\label{eqn:m(rth)C1S},\\
\hspace{-1.25cm}&\hspace{1.25cm}m_{\spt{H}}=0,\,m_{\spt{\Xi}}=\sqrt{\Bigl(\frac{1}{a^{2}_{1}}-2\mu^{2}\Bigr)},\qquad for\qquad \vartheta(r_{\spt{\theta}})=\frac{\pi}{2}\label{eqn:m(rth)C2S},\\
\hspace{-1.25cm}&\hspace{-0cm}m_{\spt{H}}=\sqrt{-\frac{\cos(\vartheta(r_{\spt{\theta}}))}{a_{1}a_{2}}\Bigl(\frac{a_{1}}{a_{2}}+\frac{a_{2}}{a_{1}}\Bigr)},\,m_{\spt{\Xi}}=\sqrt{\biggl(\frac{1}{a_{1}a_{2}}\Bigl(\frac{a_{1}}{a_{2}}+\frac{a_{2}}{a_{1}}\Bigr)-2\mu^{2}\biggr)\sin(\vartheta(r_{\spt{\theta}}))},\quad for\quad\frac{\pi}{2}<\vartheta(r_{\spt{\theta}})<\pi\label{eqn:m(rth)C3S},\displaybreak\\
\hspace{-1.25cm}&\hspace{1.25cm}m_{\spt{H}}=m_{\spt{\Xi}}=0\qquad for\qquad\vartheta(r_{\spt{\theta}})=\pi\leq\vartheta(r_{\spt{\theta}})\leq\frac{3}{2}\pi\label{eqn:m(rth)C4S},\\
\hspace{-1.25cm}&m_{\spt{H}}=\sqrt{\biggl(\frac{1}{a_{1}a_{2}}\Bigl(\frac{a_{1}}{a_{2}}+\frac{a_{2}}{a_{1}}\Bigr)-2\mu^{2}\biggr)\cos(\vartheta(r_{\spt{\theta}}))},\,m_{\spt{\Xi}}=\sqrt{-\frac{\sin(\vartheta(r_{\spt{\theta}}))}{a_{1}a_{2}}\Bigl(\frac{a_{1}}{a_{2}}+\frac{a_{2}}{a_{1}}\Bigr)},\quad for\quad\frac{3}{2}\pi<\vartheta(r_{\spt{\theta}})< 2\pi ,\label{eqn:m(rth)C5S}
\end{align}
\end{subequations}}
\hspace{-.4em}and
{\fontsize{8,8}{15}
 \selectfont
\begin{subequations}
\label{eqn:m(rth)CSA}
\begin{align}
\hspace{0cm}&m_{\spt{\tilde{A}_{1}}}=0,\,m_{\spt{\tilde{A}_{2}}}=\sqrt{-\frac{2gc_{2}\mu^{2}}{a_{1}\lambda}\cos(\vartheta(r_{\spt{\theta}}))}\qquad for\qquad 0 \leq\vartheta(r_{\spt{\theta}})<\frac{\pi}{2}\label{eqn:m(rth)C1SA},\\
\hspace{0cm}&m_{\spt{\tilde{A}_{1}}}=\sqrt{\frac{2gc_{1}\mu^{2}}{a_{2}\lambda}\cos(\vartheta(r_{\spt{\theta}}))},\,m_{\spt{\tilde{A}_{2}}}=\sqrt{-\frac{2gc_{2}\mu^{2}}{a_{1}\lambda}\sin(\vartheta(r_{\spt{\theta}}))}\qquad for\qquad\frac{\pi}{2}\leq\vartheta(r_{\spt{\theta}})<\pi\label{eqn:m(rth)C3SA},\\
\hspace{0cm}&m_{\spt{\tilde{A}_{i}}}=0\qquad for\qquad\vartheta(r_{\spt{\theta}})=\pi\leq\vartheta(r_{\spt{\theta}})\leq\frac{3}{2}\pi\label{eqn:m(rth)C4SA},\\
\hspace{0cm}&m_{\spt{\tilde{A}_{1}}}=\sqrt{\frac{2gc_{1}\mu^{2}}{a_{2}\lambda}\sin(\vartheta(r_{\spt{\theta}}))},\,m_{\spt{\tilde{A}_{2}}}=\sqrt{-\frac{2gc_{2}\mu^{2}}{a_{1}\lambda}\cos(\vartheta(r_{\spt{\theta}}))}\qquad for\qquad\frac{3}{2}\pi<\vartheta(r_{\spt{\theta}})< 2\pi ,\label{eqn:m(rth)C5SA}
\end{align}
\end{subequations}}
\hspace{-.71em} with $m_{\spt{H}}$ and $m_{\spt{\Xi}}$ the masses of the scalar fields $H$ and $\Xi$, $m_{\spt{\tilde{A}_{i}}}$ the masses of the vectorial fields $\tilde{A}_{i}$, \mbox{$\mu^{2}=-\bigl(\sum_{i}\frac{1}{a^{2}_{i}}+2\phi^{2}_{0}\lambda\bigr)$} and
\begin{equation}
\vartheta(r_{\spt{\theta}})=\frac{\pi}{2}\bigl(k_{\spt{\alpha_{1}}}r_{\spt{\theta}}+2\Re(\omega_{1}(a_{1}))\bigr)+\frac{\pi}{2} \bigl(k_{\spt{\beta_{2}}}r_{\spt{\theta}}+2\Re(\omega_{2}(a_{2}))\bigr).
\end{equation}
\end{thm}
\begin{proof}
Indicating $\phi_{1}$ and $\phi_{2}$ with $\phi$ (see the Lemma~\ref{lem:p2p1d=p1p2d1v}), we start noting that $\rho((\phi^{\dagger}\phi)(\hat{x},\hat{y}))$ can be rewritten as 
{\fontsize{10}{12,5}
 \selectfont
\begin{equation}
\label{eqn:factphi}
\hspace{-.8cm}\rho((\phi^{\dagger}\phi)(\hat{x},\hat{y}))=\cos(\vartheta(r_{\spt{\theta}}))\Phi^{2}_{1}(x,y,r_{\spt{\theta}},\omega_{1}(a_{1}),\omega_{2}(a_{2}))+\sin(\vartheta(r_{\spt{\theta}}))\Phi^{2}_{2}(x,y,r_{\spt{\theta}},\omega_{1}(a_{1}),\omega_{2}(a_{2}))
\end{equation}}
\!\!where\footnote{Considering the Lemma~\ref{lem:e^[x,y]=+-e} and the Corollary~\ref{cor:rth}, we can show that $\Phi^{2}_{i}$ are biperiodic.}  
{\fontsize{8}{15}
 \selectfont
  \begin{equation}
\hspace{-3.15cm}\Phi^{2}_{i}(x,y,r_{\spt{\theta}},\omega_{1}(a_{1}),\omega_{2}(a_{2}))=\hspace{-.45cm}\sum^{+\infty}_{l,l',m,m'=-\infty}\hspace{-.45cm}(-)^{l+m}\e^{\Re(\omega'_{l+1,l'+1}(x/a_{1},\omega_{1}(a_{1}))+\omega'_{m+1,m'+1}(y/a_{2},\omega_{2}(a_{2})))}\cos\bigl(\zeta^{\spt{\omega}}_{l,l',m,m'}(r_{\spt{\theta}})+\varphi^{\spt{\omega}}_{l'+1,m'+1}(x,y,r_{\spt{\theta}})+\delta_{i2}\frac{\pi}{2}\bigr),
\end{equation}}
\!\!with
{\fontsize{9}{13}
 \selectfont
\begin{equation}
\hspace{-.65cm}\zeta^{\spt{\omega}}_{l,l',m,m'}(r_{\spt{\theta}})=\frac{\pi}{2}\Bigl( \bigl(l(l'+1)+l'\bigr)\bigl(k_{\spt{\alpha_{1}}}r_{\spt{\theta}}+2\Re(\omega_{1}(a_{1}))\bigr)+\bigl(m(m'+1)+m'\bigr)\bigl(k_{\spt{\beta_{2}}}r_{\spt{\theta}}+2\Re(\omega_{2}(a_{2}))\bigr)\Bigr),
\end{equation}}
{\fontsize{9}{13}
 \selectfont
\begin{equation}
\varphi^{\spt{\omega}}_{l',m'}(x,y,r_{\spt{\theta}})=K^{\spt{\omega}}_{l',m'}(k_{\spt{\alpha_{1}}},k_{\spt{\alpha_{2}}},k_{\spt{\beta_{2}}},\omega_{1}(a_{1}))\frac{x}{a_{1}}+K^{\spt{\omega}}_{m',l'}(k_{\spt{\beta_{2}}},k_{\spt{\beta_{1}}},-k_{\spt{\alpha_{1}}},\omega_{2}(a_{2}))\frac{y}{a_{2}}.
\end{equation}}
\indent Now, considering that by hypothesis \mbox{$(\phi_{1},\phi_{2})\in \tilde{\mathfrak{F}}_{r_{\spt{\theta}};0}$}, we obtain that
{\fontsize{9}{12,5}
 \selectfont
\begin{equation}
\hspace{-2cm} \rho\bigl(\bigl((\hat{D}_{i}\phi)^{\dagger}(\hat{D}_{i}\phi)\bigr)(\hat{x},\hat{y})\bigr)=-\frac{\varepsilon_{ij}\varepsilon_{ij}}{2a^{2}_{j}}\rho([\bm{\hat{x}}_{j},[\bm{\hat{x}}_{j},\phi^{\dagger}\phi]])+g^{2}c^{2}_{i}\rho(\phi^{\dagger}\phi)(\rho(\tilde{\phi}_{i}\tilde{\phi}^{\dagger}_{i}))^{2}+2gc_{i}\frac{\varepsilon_{ij}}{ a_{j}}\rho(\phi^{\dagger}\phi)\rho(\tilde{\phi}_{i}\tilde{\phi}^{\dagger}_{i})+\frac{\varepsilon_{ij}\varepsilon_{ij}}{a^{2}_{j}}\rho(\phi^{\dagger}\phi),\label{eqn:DidDiF2}
\end{equation}}
\hspace{-.25em}and that 
\begin{equation}
\label{eqn:rhoV0}
\hspace{-2cm}\rho(V(\phi))=-2\lambda\phi^{2}_{0}\rho(\phi^{\dagger}\phi)+\lambda(\rho(\phi^{\dagger}\phi))^{2}+\lambda\phi^{4}_{0}.
\end{equation}
Considering thus the~\eqref{eqn:factphi} and taking into consideration the term \mbox{$\frac{\varepsilon_{ij}\varepsilon_{ij}}{a^{2}_{j}}\rho(\phi^{\dagger}\phi)$} in \mbox{$\rho\bigl(\bigl((\hat{D}_{i}\phi)^{\dagger}(\hat{D}_{i}\phi)\bigr)(\hat{x},\hat{y})\bigr)$}, for Higgs mechanism we will consider the following potential 
{\fontsize{9.8}{12.5}
 \selectfont
\begin{equation}
\label{eqn:rhoV'}
\hspace{-1.25cm}\rho(V'(\phi))=\mu^{2}\bigl(\cos(\vartheta(r_{\spt{\theta}}))\Phi^{2}_{1}+\sin(\vartheta(r_{\spt{\theta}}))\Phi^{2}_{2}\bigr)+\lambda\bigl(\cos(\vartheta(r_{\spt{\theta}}))\Phi^{2}_{1}+\sin(\vartheta(r_{\spt{\theta}}))\Phi^{2}_{2}\bigr)^{2}+\frac{1}{4\lambda}\Bigl(\mu^{2}+\sum_{i}\frac{1}{a^{2}_{i}}\Bigr)^{2}
\end{equation}}
\hspace{-.3em}with \mbox{$\mu^{2}=-\bigl(\sum_{i}\frac{1}{a^{2}_{i}}+2\phi^{2}_{0}\lambda\bigr)$}, whose critical points are obtained from the following equations
{\fontsize{9.8}{12.5}
 \selectfont
\begin{subequations}
\label{eqn:vac}
\begin{align}
\hspace{-0cm}&\frac{\partial \rho(V'(\phi))}{\partial \Phi_{1}}= 2\Phi_{1}\cos(\vartheta(r_{\spt{\theta}}))\bigl(\mu^{2}+2\lambda\bigl(\Phi^{2}_{1}\cos(\vartheta(r_{\spt{\theta}}))+\Phi^{2}_{2}\sin(\vartheta(r_{\spt{\theta}}))\bigr)\bigr),\\
\hspace{-0cm}&\frac{\partial \rho(V'(\phi))}{\partial \Phi_{2}}=2\Phi_{2}\sin(\vartheta(r_{\spt{\theta}}))\bigl(\mu^{2}+2\lambda\bigl(\Phi^{2}_{1}\cos(\vartheta(r_{\spt{\theta}}))+\Phi^{2}_{2}\sin(\vartheta(r_{\spt{\theta}}))\bigr)\bigr).
\end{align}
\end{subequations}}
\hspace{-.3em}For the case \mbox{$0 <\vartheta(r_{\spt{\theta}})<\frac{\pi}{2}$}, taking into consideration the eqs.\,\eqref{eqn:vac} and considering a neighbourhood of the origin $O$ so that \mbox{$\Phi^{2}_{i}(x,y,r_{\spt{\theta}})\geq 0$} for any \mbox{$0 \leq\vartheta(r_{\spt{\theta}})\leq 2\pi$}, we can show that the symmetry breaking typical of Higgs mechanism can happen around the vacuum $\Phi_{1}=\frac{1}{\sqrt{2}}\Phi_{0}$ and $\Phi_{2}=0$ with \mbox{$\Phi_{0}>0$} and \mbox{$\mu^{2}=-\lambda\cos(\vartheta(r_{\spt{\theta}}))\Phi^{2}_{0}$}.\\ 
We will indicate with $\rho((\phi^{\dagger}\phi)'(\hat{x},\hat{y}))$ the expansion of $\rho((\phi^{\dagger}\phi)(\hat{x},\hat{y}))$ around \mbox{$\Phi_{1}=\frac{1}{\sqrt{2}}\Phi_{0}$} and \mbox{$\Phi_{2}=0$} in terms of \mbox{the fields $H$ and $\Xi$, i.e.}
{\fontsize{10}{12.5}
 \selectfont
\begin{equation}
\label{eqn:factphiR}
\hspace{-1.35cm}\rho((\phi^{\dagger}\phi)'(\hat{x},\hat{y}))=\frac{1}{2}\Phi^{2}_{0}\cos(\vartheta(r_{\spt{\theta}}))+\Phi_{0}\cos(\vartheta(r_{\spt{\theta}}))H(x,y)+\frac{1}{2}\cos(\vartheta(r_{\spt{\theta}}))H^{2}(x,y)+\frac{1}{2}\sin(\vartheta(r_{\spt{\theta}}))\Xi^{2}(x,y). 
\end{equation}}
\!\!Let's calculate now the $m_{\spt{\tilde{A}_{i}}}$. In this regard, from the~\eqref{eqn:DidDiF2} let's consider the quantity 
\begin{equation}
g^{2}c^{2}_{i}\rho(\phi^{\dagger}\phi)(\rho(\tilde{\phi}_{i}\tilde{\phi}^{\dagger}_{i}))^{2}+2gc_{i}\frac{\varepsilon_{ij}}{ a_{j}}\rho(\phi^{\dagger}\phi)\rho(\tilde{\phi}_{i}\tilde{\phi}^{\dagger}_{i})\label{eqn:VAi}
\end{equation} 
and instead of $\rho(\phi^{\dagger}\phi(\hat{x},\hat{y}))$ let's insert the expansion $\rho((\phi^{\dagger}\phi)'(\hat{x},\hat{y}))$. We will obtain thus that such quantity is equal to
{\fontsize{8,5}{12.5}
  \selectfont
\begin{equation}
\begin{split}
\hspace{-2.25cm}&\frac{\varepsilon_{ij}}{ a_{j}}gc_{i}\Phi^{2}_{0}\cos(\vartheta(r_{\spt{\theta}}))\rho(\tilde{\phi}_{i}\tilde{\phi}^{\dagger}_{i})+\frac{g^{2}c^{2}_{i}}{2}\Phi^{2}_{0}\cos(\vartheta(r_{\spt{\theta}}))(\rho(\tilde{\phi}_{i}\tilde{\phi}^{\dagger}_{i}))^{2}+2\frac{\varepsilon_{ij}}{ a_{j}}gc_{i}\Phi_{0}\cos(\vartheta(r_{\spt{\theta}}))\rho(\tilde{\phi}_{i}\tilde{\phi}^{\dagger}_{i})H+g^{2}c^{2}_{i}\Phi_{0}\cos(\vartheta(r_{\spt{\theta}}))(\rho(\tilde{\phi}_{i}\tilde{\phi}^{\dagger}_{i}))^{2}H+\\
\hspace{-2.25cm}&+\frac{\varepsilon_{ij}}{ a_{j}}gc_{i}\cos(\vartheta(r_{\spt{\theta}}))\rho(\tilde{\phi}_{i}\tilde{\phi}^{\dagger}_{i})H^{2}+\frac{g^{2}c^{2}_{i}}{2}\cos(\vartheta(r_{\spt{\theta}}))(\rho(\tilde{\phi}_{i}\tilde{\phi}^{\dagger}_{i}))^{2}H^{2}+\frac{\varepsilon_{ij}}{ a_{j}}gc_{i}\sin(\vartheta(r_{\spt{\theta}}))\rho(\tilde{\phi}_{i}\tilde{\phi}^{\dagger}_{i})\Xi^{2}+\frac{g^{2}c^{2}_{i}}{2}\sin(\vartheta(r_{\spt{\theta}}))(\rho(\tilde{\phi}_{i}\tilde{\phi}^{\dagger}_{i}))^{2}\Xi^{2}\label{eqn:SVAi2}.
\end{split}
\end{equation}}
\hspace{-.35em}From the expansion~\eqref{eqn:SVAi2} let's isolate the first two addends defining
\begin{equation}
\rho(\tilde{V}(\tilde{\phi}_{i}))=\frac{\varepsilon_{ij}}{ a_{j}}gc_{i}\Phi^{2}_{0}\cos(\vartheta(r_{\spt{\theta}}))\rho(\tilde{\phi}_{i}\tilde{\phi}^{\dagger}_{i})+\frac{g^{2}c^{2}_{i}}{2}\Phi^{2}_{0}\cos(\vartheta(r_{\spt{\theta}}))(\rho(\tilde{\phi}_{i}\tilde{\phi}^{\dagger}_{i}))^{2}.\label{eqn:rVtphi}
\end{equation} 
The $m_{\spt{\tilde{A}_{i}}}$ will be obtained from the symmetry breaking of the potential $\rho(\tilde{V}(\tilde{\phi}_{i}))$.\pagebreak\\
Considering that $\tilde{\phi}_{i}\tilde{\phi}^{\dagger}_{i}=\tilde{\phi}^{\dagger}_{i}\tilde{\phi}_{i}$ (see  Lemma~\ref{lem:p1dp2=p2p1d}) and taking into consideration the eq.\,\eqref{eqn:factphi}, we have that  
{\fontsize{10}{12.5}
  \selectfont
\begin{equation}
\label{eqn:factTphi}
\hspace{-1.35cm}\rho(\tilde{\phi}_{i}\tilde{\phi}^{\dagger}_{i}(\hat{x},\hat{y}))=\cos(\vartheta(r_{\spt{\theta}}))\tilde{\Phi}^{2}_{1;i}(x,y,r_{\spt{\theta}},\tilde{\omega}_{1;i}(a_{1}),\tilde{\omega}_{2;i}(a_{2}))+\sin(\vartheta(r_{\spt{\theta}}))\tilde{\Phi}^{2}_{2;i}(x,y,r_{\spt{\theta}},\tilde{\omega}_{1;i}(a_{1}),\tilde{\omega}_{2;i}(a_{2}))
\end{equation}}
\!\!with \mbox{$\tilde{\Phi}^{2}_{j;i}(x,y,r_{\spt{\theta}},\tilde{\omega}_{1;i}(a_{1}),\tilde{\omega}_{2;i}(a_{2}))=\Phi^{2}_{j}(x,y,r_{\spt{\theta}},\tilde{\omega}_{1;i}(a_{1}),\tilde{\omega}_{2;i}(a_{2}))$}.
The critical points of $\rho(\tilde{V}(\tilde{\phi}_{i}))$ are obtained from the following equations 
{\fontsize{9.8}{12.5}
 \selectfont
\begin{subequations}
\label{eqn:tvac}
\begin{align}
\hspace{-.75cm}&\frac{\partial \rho(\tilde{V}(\tilde{\phi}_{i}))}{\partial \tilde{\Phi}_{1;i}}= 2c_{i}g\Phi^{2}_{0}\cos(\vartheta(r_{\spt{\theta}}))\tilde{\Phi}_{1;i}\cos(\vartheta(r_{\spt{\theta}}))\Bigl(\frac{\varepsilon_{ij}}{ a_{j}}+c_{i}g\bigl(\tilde{\Phi}^{2}_{1;i}\cos(\vartheta(r_{\spt{\theta}}))+\tilde{\Phi}^{2}_{2;i}\sin(\vartheta(r_{\spt{\theta}}))\bigr)\Bigr),\\
\hspace{-.75cm}&\frac{\partial \rho(\tilde{V}(\tilde{\phi}_{i}))}{\partial \tilde{\Phi}_{2;i}}=2c_{i}g\Phi^{2}_{0}\cos(\vartheta(r_{\spt{\theta}}))\tilde{\Phi}_{2;i}\sin(\vartheta(r_{\spt{\theta}}))\Bigl(\frac{\varepsilon_{ij}}{ a_{j}}+c_{i}g\bigl(\tilde{\Phi}^{2}_{1;i}\cos(\vartheta(r_{\spt{\theta}}))+\tilde{\Phi}^{2}_{2;i}\sin(\vartheta(r_{\spt{\theta}}))\bigr)\Bigr).
\end{align}
\end{subequations}}
\hspace{-.5em}For \mbox{$0 <\vartheta(r_{\spt{\theta}})<\frac{\pi}{2}$}, taking into consideration the eqs.\,\eqref{eqn:tvac} and considering, as for $\rho(V'(\phi))$, a neighbourhood of the origin $O$ so that \mbox{$\tilde{\Phi}^{2}_{j;i}(x,y,r_{\spt{\theta}})\geq 0$} for any \mbox{$0 \leq\vartheta(r_{\spt{\theta}})\leq 2\pi$}, we can show that while for $\rho(\tilde{V}(\tilde{\phi}_{1}))$ a symmetry breaking cannot occur, it is not the same for $\rho(\tilde{V}(\tilde{\phi}_{2}))$. In fact such breaking can occur around \mbox{$\tilde{\Phi}_{1;2}=\frac{1}{\sqrt{2}}\tilde{\Phi}_{0;2}$} and \mbox{$\tilde{\Phi}_{2;2}=0$}, with \mbox{$\tilde{\Phi}_{0;2}>0$} and \mbox{$\frac{2}{a_{1}}=gc_{2}\cos(\vartheta(r_{\spt{\theta}}))\tilde{\Phi}^{2}_{0;2}$}. Inserting thus in $\rho(\tilde{V}(\tilde{\phi}_{2}))$ the expansion $\rho((\tilde{\phi}^{\dagger}_{2}\tilde{\phi}_{2})'(\hat{x},\hat{y}))$ of $\rho((\tilde{\phi}^{\dagger}_{2}\tilde{\phi}_{2})(\hat{x},\hat{y}))$ around \mbox{$\tilde{\Phi}_{1;2}=\frac{1}{\sqrt{2}}\tilde{\Phi}_{0;2}$} and \mbox{$\tilde{\Phi}_{2;2}=0$} in terms of the fields $\tilde{H}_{2}$ and $\tilde{\Xi}_{2}$, i.e.
{\fontsize{9.8}{12.5}
 \selectfont
\begin{equation}
\label{eqn:factphiTR}
\hspace{-2cm}\rho((\tilde{\phi}^{\dagger}_{2}\tilde{\phi}_{2})'(\hat{x},\hat{y}))=\frac{1}{2}\tilde{\Phi}^{2}_{0;2}\cos(\vartheta(r_{\spt{\theta}}))+\tilde{\Phi}_{0;2}\cos(\vartheta(r_{\spt{\theta}}))\tilde{H}_{2}(x,y)+\frac{1}{2}\cos(\vartheta(r_{\spt{\theta}}))\tilde{H}^{2}_{2}(x,y)+\frac{1}{2}\sin(\vartheta(r_{\spt{\theta}}))\tilde{\Xi}^{2}_{2}(x,y), 
\end{equation}}
\hspace{-.3em}we obtain
\begin{equation}
-\frac{gc_{2}\mu^{2}}{a_{1}\lambda}\cos(\vartheta(r_{\spt{\theta}}))=\frac{1}{2}m^{2}_{\spt{\tilde{A}_{2}}}.
\end{equation}
In the end, we obtain the $m_{\spt{H}}$ and $m_{\spt{\Xi}}$ inserting $\rho((\phi^{\dagger}\phi)'(\hat{x},\hat{y}))$ in $\rho(V'(\phi))$ and $\rho((\tilde{\phi}^{\dagger}_{2}\tilde{\phi}_{2})'(\hat{x},\hat{y}))$ in the expansion~\eqref{eqn:SVAi2}. Thus we obtain that
\begin{subequations}
\begin{align}
\Bigl(\mu^{2}-\frac{1}{2a^{2}_{1}}\Bigr)\cos(\vartheta(r_{\spt{\theta}}))&=-\frac{1}{2}m^{2}_{\spt{H}},\\
\frac{1}{a^{2}_{1}}\sin(\vartheta(r_{\spt{\theta}}))&=m^{2}_{\spt{\Xi}}.
\end{align}
\end{subequations}
\indent The other cases are obtained in a similar way. In the figure~\ref{fig:V(rth)}, $\rho(V'(\phi))$ is shown for some values of $\vartheta(r_{\spt{\theta}})$.
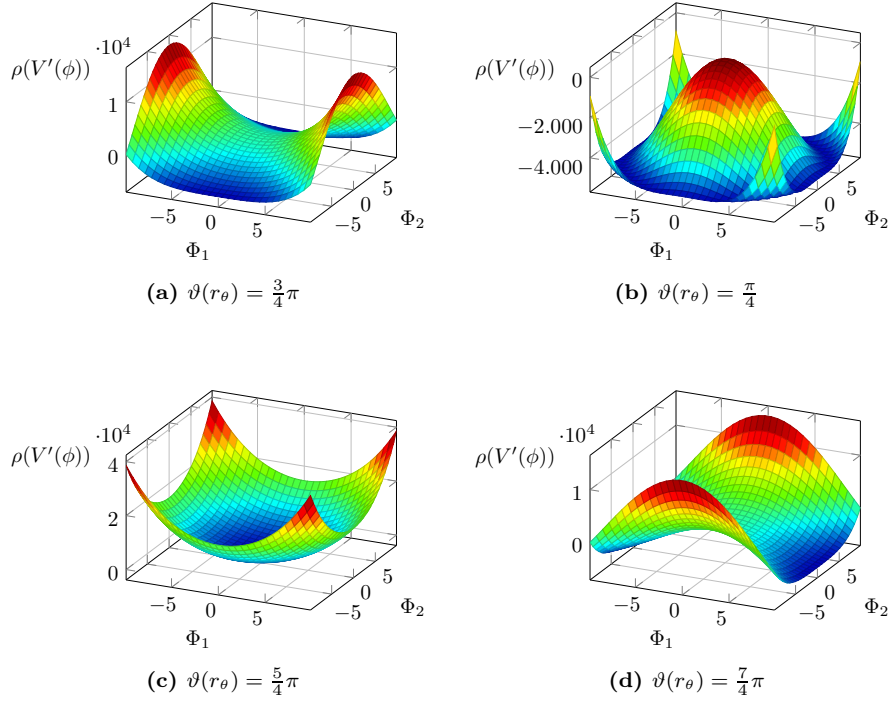
\begin{figure}[t!]
\vspace{-.75cm}
\centering
\subfloat[][\footnotesize{\mbox{$\vartheta(r_{\spt{\theta}})=\frac{3}{4}\pi$}}]{\begin{tikzpicture}[scale=0.8]
\begin{axis}[x label style={at={(axis description cs:0.3,-0.05)}},
    y label style={at={(axis description cs:1,.1)}},
z label style={at={(axis description cs:-0.1,.75)}},
xlabel=$\Phi_{1}$,
ylabel=$\Phi_{2}$,
zlabel=$\rho(V'(\phi))$,
zlabel style={rotate=-90},width=0.5\textwidth, grid=major
]
\addplot3 [domain=-9.85:9.85,
samples=30,
surf,
colormap/bluered]
{72*sqrt(2)*x^2+x^4/2-72*sqrt(2)*y^2-x^2*y^2+y^4/2};
\end{axis}
\end{tikzpicture}
}\quad
\subfloat[][\footnotesize{\mbox{$\vartheta(r_{\spt{\theta}})=\frac{\pi}{4}$}}]{\begin{tikzpicture}[scale=0.8]
\begin{axis}[x label style={at={(axis description cs:0.3,-0.05)}},
    y label style={at={(axis description cs:1,.1)}},
z label style={at={(axis description cs:-0.1,.75)}},
xlabel=$\Phi_{1}$,
ylabel=$\Phi_{2}$,
zlabel=$\rho(V'(\phi))$,
zlabel style={rotate=-90},width=0.5\textwidth, grid=major
]
\addplot3 [domain=-9.85:9.85,
samples=30,
surf,
colormap/bluered]
{-72*sqrt(2)*x^2+x^4/2-72*sqrt(2)*y^2+x^2*y^2+y^4/2};
\end{axis}
\end{tikzpicture}
}

\vspace{2em}

\subfloat[][\footnotesize{\mbox{$\vartheta(r_{\spt{\theta}})=\frac{5}{4}\pi$}}]{\begin{tikzpicture}[scale=0.8] 
\begin{axis}[x label style={at={(axis description cs:0.3,-0.05)}},
    y label style={at={(axis description cs:1,.1)}},
z label style={at={(axis description cs:-0.1,.75)}},
xlabel=$\Phi_{1}$,
ylabel=$\Phi_{2}$,
zlabel=$\rho(V'(\phi))$,
zlabel style={rotate=-90},width=0.5\textwidth, grid=major
]
\addplot3 [domain=-9.85:9.85,
samples=30,
surf,
colormap/bluered]
{72*sqrt(2)*x^2+x^4/2+72*sqrt(2)*y^2+x^2*y^2+y^4/2};
\end{axis}
\end{tikzpicture}
}\quad
\subfloat[][\footnotesize{\mbox{$\vartheta(r_{\spt{\theta}})=\frac{7}{4}\pi$}}]{\begin{tikzpicture}[scale=0.8]
\begin{axis}[x label style={at={(axis description cs:0.3,-0.05)}},
    y label style={at={(axis description cs:1,.1)}},
z label style={at={(axis description cs:-0.1,.75)}},
xlabel=$\Phi_{1}$,
ylabel=$\Phi_{2}$,
zlabel=$\rho(V'(\phi))$,
zlabel style={rotate=-90},width=0.5\textwidth, grid=major
]
\addplot3 [domain=-9.85:9.85,
samples=30,
surf,
colormap/bluered]
{-72*sqrt(2)*x^2+x^4/2+72*sqrt(2)*y^2-x^2*y^2+y^4/2};
\end{axis}
\end{tikzpicture}
}
\caption{$\rho(V'(\phi))$ for some values of $\vartheta(r_{\spt{\theta}})$.}
\label{fig:V(rth)}
\end{figure}
\end{proof}

Now, considering the eq.\,\eqref{eqn:F(theta)}, we obtain the following  

\begin{corth}
\label{cor:FL}
Given $(\mathcal{B}^{\nc}_{\spt{\mathfrak{F}_{r_{\spt{\theta}}}}},\rho)$  with $\mathcal{L}$ defined by~\eqref{eqn:dL}, $\theta_{\nc}>0$, $\lambda>0$, $c_{i}>0$, $A_{i}$ and \mbox{$\phi^{\dagger}_{1}\phi_{2}$} hermitian and \mbox{$(\phi_{1},\phi_{2})\in \tilde{\mathfrak{F}}_{r_{\spt{\theta}};0}$}, considering $g>0$ and the quantities \mbox{$(k_{\spt{\alpha_{1}}}+k_{\spt{\beta_{2}}})\theta_{\cc}>0$} and \mbox{$\omega_{\spt{\mathfrak{R}}}\leq -1/2$}, with \mbox{$\omega_{\spt{\mathfrak{R}}}=\sum^{2}_{i=1}\Re(\omega_{i}(a_{i}))$}, in the confined phase we have that 
\begin{subequations}
\label{eqn:FL}
\begin{align}
\hspace{-0.5cm}\frac{\theta_{\cc}(k_{\spt{\alpha_{1}}}+k_{\spt{\beta_{2}}})}{g(1+2\omega_{\spt{\mathfrak{R}}})}<\mathcal{F}<\frac{\theta_{\cc}(k_{\spt{\alpha_{1}}}+k_{\spt{\beta_{2}}})}{2g(\omega_{\spt{\mathfrak{R}}}-1)}\qquad &for\qquad \omega_{\spt{\mathfrak{R}}}<-1/2\,\, ,\label{eqn:FL1}\\
\hspace{-0.5cm}\mathcal{F}<-\frac{\theta_{\cc}(k_{\spt{\alpha_{1}}}+k_{\spt{\beta_{2}}})}{3g}\qquad &for\qquad\omega_{\spt{\mathfrak{R}}}= -1/2\,\, .\label{eqn:FL3}
\end{align}
\end{subequations}
\end{corth}
\begin{proof}
Considering what already said in the Introduction, the eqs.\,\eqref{eqn:FL} are obtained from~\eqref{eqn:F(theta)} and by the fact that the confined phase is characterized by \mbox{$-\frac{\pi}{2}<\vartheta(r_{\spt{\theta}})<\pi$} (see eqs.\,\eqref{eqn:m(rth)CS}).
\end{proof}

\begin{remark}
\label{rmk:m(tau)}
Differently from~\citep{Baal82,Manton2004,Arroyo04,Forgacs2005,Lozano:2006xn,EWeinberg2012}, the possibility to have a superior and an inferior limit is thus predicted for $\mathcal{F}$.
This is particularly interesting in the phenomenology of Meissner effect. From our point of view, the espulsion of an external magnetic field $\mathbf{B}_{\spt{\textsc{ext}}}$ from inside a superconductor, typical of such effect, ceases when $\vartheta(r_{\spt{\theta}})$ assumes values corresponding to a restoration of the $U(1)$ symmetry (see eqs.\,\eqref{eqn:m(rth)C4S}). In this regard, a novelty is noted: the symmetry breaking of Higgs mechanism is thus determined by the belonging or not of $\vartheta(r_{\spt{\theta}})$ to a specific interval of values and not, as usual, by the sign of the constant $\mu^{2}$. And this, as $\vartheta(r_{\spt{\theta}})$ can vary in time (on the basis of what we have just said about Meissner effect), implies particle spectrum whose \emph{masses can vary in time in a quantized way and they are null when $\vartheta(r_{\spt{\theta}})$ takes values outside a specific interval}. \pagebreak\\
\indent In the end, a further novelty is the absence, in general, of a massless Goldstone boson: the expansion of $\rho(\tilde{V}(\tilde{\phi}_{i}))$ around a minimum, together with the usual one $\rho(V'(\phi))$  around one of its minima, leads to predict in general a non null mass for such boson (see proof of the Theorem~\ref{thm:Higgsth}).   
\end{remark}

We have thus shown how, differently from refs.\,\citep{Forgacs2005,Lozano:2006xn}, thanks to our approach to non-commutativity, $\mathcal{E}$ has a domain of definition. The analysis on the domain of integration has led to define the pair $(\mathcal{B}^{\nc}_{\spt{\mathfrak{F}_{r_{\spt{\theta}}}}},\rho)$ (see Definition~\ref{def:B}) such that \mbox{$\mathcal{E}\in\mathcal{A}_{\theta_{\nc}}$} and, inter alia, to characterize  $e^{\iu\pi \theta_{\nc}\hat{\mathbbm{1}}}$ as the generator of a cyclic subgroup $\tau_{r_{\spt{\theta}}}$ of $U(1)$ (see Corollary~\ref{corth:thetag}).\\
Successively, for an appropriate choice of fields $\phi$ and $A_{\mu}$, we have evaluated $E$ for $\lambda=g^{2}/2$ rewriting it \textit{à la} Bogomolny (see Theorem~\ref{thm:E12Bogo}). In this regard, two novelties have emerged: the first is that $E$ depends only on $r_{\spt{\theta}}$ and on $(a_{1},a_{2})$ so that \mbox{$E\rightarrow+\infty$} for \mbox{$a_{1},a_{2} \rightarrow+\infty$} (in the refs.\,\citep{Baal82,Manton2004,Arroyo04,Forgacs2005,Lozano:2006xn,EWeinberg2012}, instead, it depended only on the index $n$ that labelled the homotopic classes of the fields $\phi$ and $A_{\mu}$) and the second is that the fields, in the hypothesis of the Theorem~\ref{thm:E12Bogo}, already appear like minimum energy configuration: \emph{therefore it is not necessary to solve any equation system, even less a BPS one} (see Remark~\ref{rmk:noBPS}) as it happens for the refs.\,\citep{Manton2004,Arroyo04,Forgacs2005,Lozano:2006xn,EWeinberg2012} instead. In particular, in such regard, we have noticed how we cannot get back the II BPS eq as \mbox{$[\tilde{A}_{i},\tilde{A}_{j}](\hat{x},\hat{y})=0$}.\\
\indent Then we have moved on to deal with the influence of $\theta$ on the other observables of the model.
No dependency of $g$ from $\theta$ has been predicted as in the refs.\,\citep{Forgacs2005,Lozano:2006xn}.\\
We have found that $\mathcal{F}$, and so the quantization and the sign of $q_{e}$, depends exclusively on $r_{\spt{\theta}}$ (see eq.\,\eqref{eqn:F(theta)} and eq.\,\eqref{eqn:qe=1/rththcc}) and not on $n$ as in refs.\,\citep{Baal82,Manton2004,Arroyo04,Forgacs2005,Lozano:2006xn,EWeinberg2012}.
In the end, we have obtained a relation between the masses of the particle spectrum, relative to the fields $\phi$ and $A_{\mu}$, and $\vartheta(r_{\spt{\theta}})$ (see Theorem~\ref{thm:Higgsth}). According to this relation, these masses can vary in time in a quantized way just because $\vartheta(r_{\spt{\theta}})$ can vary in time (see Remark~\ref{rmk:m(tau)}) and are null when $\vartheta(r_{\spt{\theta}})$ is outside a specific interval: considering the\emph{~\ref{item:(ir1)}} of the Corollary~\ref{corth:thetag}, we can also say that \emph{the masses vary in time as the order of $\tau_{r_{\spt{\theta}}}$ varies in time and are null when such order is outside a specific interval}.\\
As for the particle spectrum, we have found that the double $U(1)$ symmetry breaking regulated by $r_{\spt{\theta}}$, one relative to the minima of $\rho(\tilde{V}(\tilde{\phi}_{i}))$ and the other due to the minima of $\rho(V'(\phi))$, has implied, in general, that also the usual massless Goldstone boson acquired mass (see Theorem~\ref{thm:Higgsth}).\\
\indent Let's finish with a question: what gives stability to the vacuum of the theory as this is not granted anymore by the belonging of such vacuum to a homotopic class labelled by $n$, as we have already seen? From this point of view, we will be able to answer after investigating the nature of $r_{\spt{\theta}}$ and its dynamics.


\section{Conclusions}
\label{sec:concl}
\noindent The Tables~\ref{tab:FSvsEpp1} and~\ref{tab:FSvsEpp} show the main differences between ``Fock space approach'', as it is dealt with in~\citep{Forgacs2005,Lozano:2006xn}, and the approach proposed here to non-commutativity.\\
We conclude with five considerations. The first concerns the fact that, without any loss of generality, we could replace the usual commutative treatment of coordinates with ours: in the ordinary calculations we can replace $(x,y)\in\R^{2}$ with $(x\bm{\hat{x}},y\bm{\hat{y}})$, with $[\bm{\hat{x}},\bm{\hat{y}}]=\iu\theta\hat{\mathbbm{1}}$, and consider then the commutative case setting $\theta$ not necessarily on zero (as in~\citep{Forgacs2005,Lozano:2006xn}) but equal to an element of \mbox{$[0]\0$} with \mbox{$[0]\in\R/2\Z$}. In this way it will be interesting to evaluate the possible influence of $\theta_{\cc}$, and then also of $\theta_{\nc}$, on the observables in different models from the one examined so far.\\
\indent The second consideration is that here it is possible to predict, for the color-electric charge of a quark, also fractional charges and not only therefore integer ones as in refs.\,\citep{Baal82,Manton2004,Arroyo04,Forgacs2005,Lozano:2006xn,EWeinberg2012}.\\
\indent The third consideration concerns energy at the point of Bogomolny for minimum configurations. Unlike the aforementioned references, here such energy depends on global parameters of the manifold on which $\mathcal{E}$ is defined, i.e.\,the periods $(a_{1},a_{2})$, and it is not such that \mbox{$E<+\infty$} for \mbox{$a_{1},a_{2} \rightarrow+\infty$}. In other words, in our case, it is not possible to consider the model on a plane.\\
\begin{table}[h!]
\vspace{-2cm}
\hspace{-2.25cm}
\begin{tikzpicture}[x= 13 em, y=-.5cm,  node distance=0 cm,outer sep = 0pt]

\node[dayF,fill=blue!10] (lundi) at (1,8)   {\textsc{fock space approach}};
\node[day,hoursF2,fill=green!10]  (mardi) [right = of lundi]    {\textsc{approach proposed}};
\node[hourF] (8-9) at (1,8)              {\large{Coordinates Map}};
\node[hourF] (9-10)  [below = of  8-9 ] {\fontsize{11}{15}\selectfont{Translation $t_{a}$ in the direction $\hat{x}$ of a quantity $a$}};
\node[3hoursM,fill= yellow!30,minimum height = 3.3cm] (10-11) [below = of  9-10] {\Large{$[\hat{x},\hat{y}] \propto\hat{\mathbbm{1}}$}};
\node[3hoursM,fill= yellow!30,minimum height = 3.8cm] (11) [below = of  10-11] {\large{Twist matrices}};
\node[3hoursM,fill= yellow!30,minimum height = 3.3cm] (12) [below = of 11] {\Large{$\theta$}};
\node[hourF,minimum height = 2.243cm] (13) [below = of 12] {\large{Domain of $\mathcal{E}$}};
\node[title, text centered, above = of  8-9]   {\fontsize{8.5}{9}\selectfont{\textsc{dual superconductivity model}}};

\node[6BF,3hours] at (1,8)  {\emph{Weyl map}\\$(x,y)\mapsto (\hat{x},\hat{y})$};
\node[6BF,3hours]     at (1,12) {\fontsize{11}{15}\selectfont{$t_{a}\colon \hat{x}\mapsto\hat{x}+a\hat{\mathbbm{1}}$}};
\node[6BF,3hoursM,text width = 14 em]     at (1,16) {\large{Arbitrary position}};
\node[6BF,3hoursM,text width = 14 em,minimum height = 3.8cm]     at (1,22.6) {\hspace{-.1em}\fontsize{6}{12}\selectfont{$U_{1}(\hat{x},\hat{y})=\e^{\iu\pi\omega L_{1}\hat{y}}$, $U_{2}(\hat{x},\hat{y})=\e^{-\iu\pi\omega L_{2}\hat{x}}$} \\\vspace{1em} \fontsize{8}{12}\selectfont{with\,\, $\omega =\frac{1}{\theta\pi}(1-s(\theta))$}\\\vspace{1em} \fontsize{8}{12}\selectfont{$s(\theta)=\sqrt{1-2\pi\theta n/L_{1}L_{2}}$,\hspace{.4em}$n\in\Z$}};
\node[6BF,fill=blue!17.5, minimum height=1.65cm]     at (1,30.2) {For\, $\mathcal{T}^{2}_{\cc}$,\,\, $\theta=0$ with $\theta\in\R$};
\node[6BF, minimum height=1.65cm]     at (1,33.5)   {For\, $\mathcal{T}^{2}_{\nc}$,\,\,$\theta>0$};
\node[6BF,3hours]     at (1,36.8) {\fontsize{10}{15}\selectfont{No domain:}\\\vspace{0em} \hspace{-.132em}\fontsize{5.9}{12}\selectfont{\!$(\hat{D}_{\mu}\phi)^{\dagger}(\hat{D}^{\mu}\phi)$\! and\! $(\phi^{\dagger}\phi-\phi^{2}_{0})^{2}$\!\! defined\!\! on\!\! $\mathcal{T}^{2}$\!},\\\vspace{0em}\fontsize{7}{15}\selectfont{\mbox{$\frac{1}{4}F_{\mu\nu}F^{\mu\nu}$} defined on $\tilde{\mathcal{T}}^{2}$},\\\vspace{.2em}\fontsize{7}{15}\selectfont{with $\tilde{L}_{i}=s(\theta)L_{i}$,\,$\tilde{L}_{i}$ periods of $\tilde{\mathcal{T}}^{2}$}};

\node[2AF2]        at (2.1285,8)  {\emph{Exponential mapping} \\ $p\in\mathcal{T}^{2}\mapsto (\hat{x},\hat{y})\equiv (x\bm{\hat{x}},y\bm{\hat{y}})$};
\node[2AF2]     at (2.1285,12)  {\fontsize{11}{15}\selectfont{$t_{a}\colon x\bm{\hat{x}}\mapsto x\bm{\hat{x}}+a\bm{\hat{x}}=(x+a)\bm{\hat{x}}$}};
\node[2AF2, fill = green!20, minimum height=1.8cm]       at (2.1285,16) {\fontsize{6.5}{9}\selectfont{A position that, for $\mathcal{T}^{2}_{\cc}$, can derive from the property of cyclicity of $Z_{N}$: if $(\Omega^{(ab)}_{\mu},\Omega^{(ab)}_{\nu})\in\mathscr{D}^{\ell}_{\spt{\beta\alpha}}$\\\vspace{.5em} \hspace{-.075em}\fontsize{6.8}{12}\selectfont{$Z_{\mu\nu}=[[\Omega_{\mu},\Omega_{\nu}]]\! \Leftrightarrow\! [\hat{\mu},\hat{\nu}]\!\propto\!\hat{\mathbbm{1}}$\! and\! $(\Omega^{(ab)}_{\mu},\Omega^{(ab)}_{\nu})\!\in\!\mathscr{D}_{g}$}}};
\node[2AF2, minimum height=1.5cm]     at (2.1285,19.603) {\fontsize{6.5}{9}\selectfont{\mbox{A\! position\! that can derive\! from\! the\! homogeneity\! of $\mathcal{T}^{2}_{\nc}$}\\\fontsize{7.25}{9}\selectfont{For $\mathcal{T}^{2}_{\nc}$,  $\mathscr{C}\cap\mathscr{D}_{\spt{\beta\alpha}}\neq\emptyset\Leftrightarrow[\hat{\mu},\hat{\nu}]\propto\hat{\mathbbm{1}}$ and $\mathscr{D}^{*}_{\spt{\beta\alpha}}\neq\emptyset$}}};
\node[2AF2,minimum height=1.9cm,fill = green!20]     at (2.1285,22.61) {\hspace{-0em}\fontsize{6.5}{9}\selectfont{For $\mathcal{T}^{2}_{\cc}$, $(U_{1},U_{2})\equiv (\Omega_{\mu},\Omega_{\nu})\in\mathscr{D}^{\ell}_{\spt{\beta\alpha,U(1)}}\!\sim\!\mathscr{D}^{*}_{\spt{\beta\alpha,U(1)}}$}\\\vspace{.7em}\fontsize{8}{9}\selectfont{\mbox{$\Omega_{1}(\hat{x},\hat{y})=\e^{2\pi\iu n\beta_{1}(y)\bm{\hat{y}}}$, $\Omega_{2}(\hat{x},\hat{y})=\e^{2\pi\iu n\alpha_{2}(x)\bm{\hat{x}}}$}\\  with $n\!\!\!\!\!\in\!\!\!\!\!\!\Z$}};
\node[2AF2, minimum height=1.9cm]     at  (2.1285,26.4) {\fontsize{8.5}{9}\selectfont{For $\mathcal{T}^{2}_{\nc}$,\,\, $(\Omega_{\mu},\Omega_{\nu})\in\mathscr{D}^{*}_{\spt{\beta\alpha,U(1)}}$}\\\vspace{.2em}\fontsize{8}{9}\selectfont{$\Omega_{1}(\hat{x},\hat{y})=\e^{2\pi\iu n\beta_{1}(y)\bm{\hat{y}}}\e^{2\pi\iu n\alpha_{1}(x)\bm{\hat{x}}}$, $\Omega_{2}(\hat{x},\hat{y})=\e^{2\pi\iu n\alpha_{2}(x)\bm{\hat{x}}}\e^{2\pi\iu n\beta_{2}(y)\bm{\hat{y}}}$}};
\node[2AF2, fill = green!20, minimum height=1.65cm]     at (2.1285,30.2) {\fontsize{9}{9}\selectfont{Given $\mathcal{B}^{\cc }_{\spt{U(1)}}$,\, $\e^{\iu\pi\theta_{\cc}\hat{\mathbbm{1}}}=e$\\  where \mbox{$\theta_{\cc}\in [0]\0$} with \mbox{$[0]\in\R/2\Z$}}};
\node[2AF2, minimum height=1.65cm]     at (2.1285,33.5)  {\hspace{-.2em}\fontsize{6.5}{9}\selectfont{Given $(\mathcal{B}^{*\nc }_{\spt{U(1)}},\rho)$, $e^{\iu\pi \theta_{\nc}\hat{\mathbbm{1}}}$ is the generator of a cyclic subgroup $\tau_{r_{\spt{\theta}}}$ of $U(1)$ where \mbox{$r_{\spt{\theta}}=\theta_{\cc}/\theta_{\nc}\in\Q$} and \mbox{$\theta_{\nc}\in\R\setminus [0]$} (see Corollary~\ref{corth:thetag})}};
\node[2AF2,minimum height = 2.243cm]     at (2.1285,36.8) {\large{$\mathcal{T}^{2}_{\nc}$}};
\end{tikzpicture}
\caption{Differences between ``Fock space approach'', as it is dealt with in~\citep{Forgacs2005,Lozano:2006xn}, and the approach proposed here relatively to the non-commutative treatment of space coordinates. The sets $\mathscr{D}^{\ell}_{\spt{\beta\alpha,U(1)}}$ and $\mathscr{D}^{*}_{\spt{\beta\alpha,U(1)}}$ are the analogous for $U(1)$ respectively of $\mathscr{D}^{\ell}_{\spt{\beta\alpha}}$ and $\mathscr{D}^{*}_{\spt{\beta\alpha}}$.}
\label{tab:FSvsEpp1}
\end{table}
\indent The fourth concerns the masses of the particle spectrum. In fact, these masses depend on $r_{\spt{\theta}}$ which in its turn may be non constant on time (still in a non specified way). Thus, in general, we may conjecture that \emph{the dark matter is not made up necessarily of unknown particles (here, among other things, it is expected that even the usual Goldstone boson acquires mass). But it could also be the ordinary matter whose particles have however mass values different from the current ones.}\\ 
\indent Finally, the fifth and ultimate consideration concerns the open question on the stability of the vacuum that, as we have seen, is no longer guaranteed by topological instances (belonging or not of the vacuum to a homotopic class). From our point of view, the question is: considering the eqs.\,\eqref{eqn:m(rth)CS} and~\eqref{eqn:m(rth)CSA}, how come is $\vartheta(r_{\spt{\theta}})$ ``confined'' in an interval of values that does not imply the annulment of all the masses of the particle spectrum?\\
\begin{table}[!]
\vspace{-2cm}
\hspace{-1cm}
\begin{tikzpicture}[x= 13 em, y=-.5cm,  node distance=0 cm,outer sep = 0pt]
\node[day,fill=blue!10] (lundi) at (1,8)   {\textsc{fock space approach}};
\node[day,hours,fill=green!10]  (mardi) [right = of lundi]    {\textsc{approach proposed}};

\node[hour] (8-9) at (1,8) {\fontsize{10}{15}\selectfont{\large{Gauge charge $g$}}};
\node[2hours,minimum height = 2.243cm,fill= yellow!30] (9-10) [below = of 8-9] {\Large{$\mathcal{F}$}};

\node[2hours,minimum height = 2.97cm,fill= yellow!30] (10-11) [below = of 9-10] {\fontsize{8.5}{15}\selectfont{\hspace{-0.25cm}{\large{$E$}}, rewritten \textit{à la} Bogomolny, for $\lambda=g^{2}/2$ and for \\\vspace{.25em} minimum configurations }};
\node[2hours,fill= yellow!30] (11-12) [below = of 10-11] {\fontsize{9.5}{15}\selectfont{Masses of the particle spectrum relative to $\phi$ and $A_{\mu}$}};

\node[title, text centered, above = of  8-9]   {\fontsize{8.5}{9}\selectfont{\textsc{dual superconductivity model}}};

\node[6B,2hours] at (1,8)  {\fontsize{13}{15}\selectfont{$\tilde{g}=s(\theta)g$}};
\node[6B,1hour,minimum height = 2.243cm] at (1,11)  {\fontsize{14.5}{15}\selectfont{$\frac{2\pi n}{g}$}};
\node[6B,1hour,minimum height = 3cm]     at (1,15.49) {\fontsize{13.5}{15}\selectfont{$\pi\phi^{2}_{0}\abs{n}$}\\\vspace{1em}\fontsize{10}{15}\selectfont{if the BPS system holds}};
\node[6B,1hour]     at (1,21.43) {\fontsize{10}{15}\selectfont{No relation with $\theta$}};

\node[2A,2hours]        at (2,8)   {\fontsize{13}{15}\selectfont{$g$}};
\node[2A,1hour,minimum height = 2.243cm]        at (2,11)   {\fontsize{13}{9}\selectfont{$-\frac{1}{r_{\spt{\theta}}}\frac{\theta_{\cc}}{g}$} \\\vspace{.4em}\fontsize{7.5}{5}\selectfont{It can have an upper and a lower limit (see Corollary~\ref{cor:FL})}};
\node[2A,1hour,minimum height = 3cm]     at (2,15.49) {\fontsize{13}{15}\selectfont{$\frac{g^{2}\phi^{4}_{0}}{2}a_{1}a_{2}+\frac{\theta^{2}_{\cc}}{2g^{2}r^{2}_{\spt{\theta}}a_{1}a_{2}}$}\\\vspace{1em}\hspace{-.2em}\fontsize{7.5}{15}\selectfont{in\! the\! hypotheses\! of\! the\! Theorem~\ref{thm:E12Bogo}}};
\node[2A,1hour]     at (2,21.43) {\fontsize{9.5}{15}\selectfont{Relation with $\theta$ according to the eqs.\,\eqref{eqn:m(rth)CS} and~\eqref{eqn:m(rth)CSA} }};
\end{tikzpicture}
\caption{Differences between ``Fock space approach'', as it is dealt with in~\citep{Forgacs2005,Lozano:2006xn}, and the approach proposed here relatively to the influence of $\theta$ on the observables of the \emph{dual superconductivity} model.} 
\label{tab:FSvsEpp}
\end{table}
In the future, we plan to answer this question, by first studying the model statics for energies with $\lambda\neq g^{2}/2$ and then the multivortex dynamics (already discussed in the commutative context on the $\mathcal{T}^{2}$ in ref.\,\citep{Gonzalez07}). In this regard, it is interesting to check if also for $\lambda\neq g^{2}/2$ it is possible, as here, to find the minimum configurations without solving any equation system but the TBC only, and so to develop a multivortex dynamics starting from them (usually this dynamics is developed starting from GL or BPS eqs (see ref.\,\citep{Manton2004})). At the moment, the only useful indication about the dynamics comes from the fact that only the parameter $r_{\spt{\theta}}$ can be considered time-dependent (see Remark~\ref{rmk:m(tau)}). \\
On this aspect, it will be interesting to find the relation between $r_{\spt{\theta}}$ and the metric that regulates the aforementioned multivortex dynamics. 
So, considering the eqs.\,\eqref{eqn:m(rth)CS} and~\eqref{eqn:m(rth)CSA}, we will obtain relations between the masses and such metric which could represent, more generally, the first step in the search for a relation between $m_{\ssp{H^{0}}}$ (mass of Higgs) and $c$ (speed of light). \\
\indent Finally, from a mathematical point of view, it will be interesting to test the validity of the Conjecture~\ref{conj:nnlin} at higher orders than the fifth one, by using the computer aid and a calculation algorithm of (C-B-H-D) like the one illustrated in~\citep{Casas2009}. 
If, at a certain order, such a conjecture does not hold, a condition can be ``captured'' on $[\hat{x},\hat{y}]$ and we can already get from it the proportionality to $\hat{\mathbbm{1}}$ without having to consider the~\eqref{eqn:cOmunuab} for $\mathcal{T}^{2}_{\nc}$ or the property of cyclicity of $Z_{N}$ for $\mathcal{T}^{2}_{\cc}$ (see Theorem~\ref{thm:rev3Tc}).

 %

\pdfbookmark[1]{Acknowledgments}{acknowledgments}

\section*{Acknowledgments}

\noindent We thank L.\,Martina for having proposed this topic and for many helpful discussions.
This work is financially partially supported by the ``Regione Puglia: Ritorno al Futuro - Ricerca''.

\appendix


\section{Some calculations about {\boldmath{$\phi(\hat{x}+\hat{a}_{1},\hat{y})$}} and \mbox{{\boldmath{$\phi(\hat{x},\hat{y}+\hat{a}_{2})$}}}}
\label{sec:AppFintT}
\noindent We will report here the detail of the calculations of the equations~\eqref{eqn:phitrasl}.\\

\vspace{-1.15em}

\indent Given a $\phi\in\mathscr{F}$, we have that 
{\fontsize{11}{10}
 \selectfont
\begin{equation}
\hspace{-2.25cm}\phi(\hat{x}+\hat{a}_{1},\hat{y})=\!\!\!\sum^{\infty}_{q,k=-\infty}\!\!(-)^{q+k}\e^{I_{q,k}(\hat{z}_{1}+\pi\hat{\tau}_{1},\hat{z}_{2}+\pi\hat{\tau}'_{2},\hat{\tau}_{1},\hat{\tau}_{2})}\e^{\omega_{q,k}(x+a_{1},y)}\e^{-\iu\pi q^{2}\hat{\tau}_{1}-\iu\pi k^{2}\hat{\tau}_{2}-2\iu q\hat{z}_{1}-2\iu q\pi\hat{\tau}_{1}-2\iu k\hat{z}_{2}-2\iu k\pi\hat{\tau}'_{2}},
\end{equation}}
\hspace{-.25em}where
{\fontsize{9.5}{20}
 \selectfont
\begin{equation}
\begin{split}
\hspace{-2.5cm}I_{q,k}(\hat{z}_{1}+\pi\hat{\tau}_{1},\hat{z}_{2}+\pi\hat{\tau}'_{2},\hat{\tau}_{1},\hat{\tau}_{2})=&-2q(k+1)[\hat{z}_{1},\hat{z}_{2}]-\frac{\pi^{2}}{2}(q^{2}+2q)k^{2}[\hat{\tau}_{1},\hat{\tau}_{2}]+\frac{\pi}{3} q^{3}[\hat{\tau}_{1},\hat{z}_{1}]-\pi(q^{2}+2q)(k+1)[\hat{\tau}_{1},\hat{z}_{2}]+\notag\\
&+\pi qk^{2}[\hat{\tau}_{2},\hat{z}_{1}]+\frac{\pi}{3}k^{3}[\hat{\tau}_{2},\hat{z}_{2}]-2\pi q(k+1)[\hat{z}_{1},\hat{\tau}'_{2}]+\frac{\pi^{2}}{3}k^{3}[\hat{\tau}_{2},\hat{\tau}'_{2}].
\end{split}
\end{equation}}
\hspace{-.25em}However
{\fontsize{9.5}{20}
 \selectfont
\begin{equation}
\hspace{-2.5cm}\phi(\hat{x}+\hat{a}_{1},\hat{y})=\!\!\!\sum^{\infty}_{q,k=-\infty}\!\!(-)^{q+k}\e^{I_{q,k}(\hat{z}_{1}+\pi\hat{\tau}_{1},\hat{z}_{2}+\pi\hat{\tau}'_{2},\hat{\tau}_{1},\hat{\tau}_{2})}\e^{\omega_{q,k}(x+a_{1},y)}\e^{-\iu\pi (q+1)^{2}\hat{\tau}_{1}-\iu\pi k^{2}\hat{\tau}_{2}-2\iu (q+1)\hat{z}_{1}-2\iu k\hat{z}_{2}+2\iu\hat{z}_{1}+\iu\pi\hat{\tau}_{1}-2\iu k\pi\hat{\tau}'_{2}},
\end{equation}}
\hspace{-.25em}so that 
{\fontsize{10}{15}
 \selectfont
\begin{equation}
\begin{split}
\hspace{-2.5cm}\phi(\hat{x}+\hat{a}_{1},\hat{y})=\e^{2\iu\hat{z}_{1}+\iu\pi\hat{\tau}_{1}}\e^{2\iu\pi\Re(\omega_{1}(x))+\iu\pi\Re(\omega_{1}(a_{1}))}\!\!\!\sum^{\infty}_{q,k=-\infty}\!\!(-)^{q+k}\e&^{I_{q,k}(\hat{z}_{1}+\pi\hat{\tau}_{1},\hat{z}_{2}+\pi\hat{\tau}'_{2},\hat{\tau}_{1},\hat{\tau}_{2})}\e^{I'_{q,k}(\hat{z}_{1},\hat{z}_{2},\hat{\tau}_{1},\hat{\tau}_{2})}\e^{\omega_{q+1,k}(x,y)}\\
&\times\e^{-\iu\pi (q+1)^{2}\hat{\tau}_{1}-\iu\pi k^{2}\hat{\tau}_{2}-2\iu (q+1)\hat{z}_{1}-2\iu k\hat{z}_{2}-2\iu k\pi\hat{\tau}'_{2}}
\end{split}
\end{equation}}
\hspace{-.25em}with
{\fontsize{11}{15}
 \selectfont
\begin{equation}
\hspace{-2cm}I'_{q,k}(\hat{z}_{1},\hat{z}_{2},\hat{\tau}_{1},\hat{\tau}_{2})=-2k[\hat{z}_{1},\hat{z}_{2}]-\frac{\pi^{2}}{2}k^{2}[\hat{\tau}_{1},\hat{\tau}_{2}]+\pi q(q+1)[\hat{\tau}_{1},\hat{z}_{1}]-\pi k [\hat{\tau}_{1},\hat{z}_{2}]+\pi k^{2}[\hat{\tau}_{2},\hat{z}_{1}]-2\pi k[\hat{z}_{1},\hat{\tau}'_{2}].
\end{equation}}
\hspace{-.25em}Now, considering that 
{\fontsize{8,5}{15}
 \selectfont
\begin{equation}
\hspace{-3.1cm}I_{q,k}(\hat{z}_{1}+\pi\hat{\tau}_{1},\hat{z}_{2}+\pi\hat{\tau}'_{2},\hat{\tau}_{1},\hat{\tau}_{2})+I'_{q,k}(\hat{z}_{1},\hat{z}_{2},\hat{\tau}_{1},\hat{\tau}_{2})=I_{q+1,k}(\hat{z}_{1},\hat{z}_{2},\hat{\tau}_{1},\hat{\tau}_{2})+2[\hat{z}_{1},\hat{z}_{2}]-\frac{\pi}{3}[\hat{\tau}_{1},\hat{z}_{1}]+\pi [\hat{\tau}_{1},\hat{z}_{2}]-2\pi (q(k+1)+k)[\hat{z}_{1},\hat{\tau}'_{2}]+\frac{\pi^{2}}{3}k^{3}[\hat{\tau}_{2},\hat{\tau}'_{2}],
\end{equation}}
\hspace{-.25em}we obtain that
{\fontsize{11}{15}
 \selectfont
\begin{equation}
\begin{split}
\label{eqn:trasla1}
\hspace{-1.5cm}\phi(\hat{x}+\hat{a}_{1},\hat{y})=\Omega'_{1}(\hat{x},\hat{y})\hspace{-.25cm}\sum^{\infty}_{q,k=-\infty}\!\!\e&^{-2\pi (q(k+1)+k)[\hat{z}_{1},\hat{\tau}'_{2}]+\frac{\pi^{2}}{3}k^{3}[\hat{\tau}_{2},\hat{\tau}'_{2}]}\e^{I_{q+1,k}(\hat{z}_{1},\hat{z}_{2},\hat{\tau}_{1},\hat{\tau}_{2})}\e^{\omega_{q+1,k}(x,y)}\\
&\times\e^{-\iu\pi (q+1)^{2}\hat{\tau}_{1}-\iu\pi k^{2}\hat{\tau}_{2}-2\iu (q+1)\hat{z}_{1}-2\iu k\hat{z}_{2}-2\iu k\pi\hat{\tau}'_{2}}
\end{split}
\end{equation}}
\hspace{-.25em}and so, taking into account the Lemma~\ref{lem:e^[x,y]=+-e}, we have that 
\begin{equation}
\phi(\hat{x}+\hat{a}_{1},\hat{y})=\Omega'_{1}(\hat{x},\hat{y})\sum^{\infty}_{q,k=-\infty}f^{\spt{\omega}}_{q+1,k}(\hat{x},\hat{y},\theta\hat{\mathbbm{1}})=\Omega'_{1}(\hat{x},\hat{y})\phi(\hat{x},\hat{y})
\end{equation}
\,if{}f $\e^{\iu \frac{\pi}{3}  \hat{\tau}'_{2}}=\pm e$. \\
\indent As concerns $\phi(\hat{x},\hat{y}+\hat{a}_{2})$, proceeding in the same way, we obtain that
{\fontsize{11}{15}
 \selectfont
\begin{equation}
\begin{split}
\label{eqn:trasla2}
\hspace{-1.75cm}\phi(\hat{x},\hat{y}+\hat{a}_{2})=\Omega'_{2}(\hat{x},\hat{y})\!\!\!\sum^{\infty}_{q,k=-\infty}\!\!\e&^{4q[\hat{z}_{1},\hat{z}_{2}]+2\pi q^{2}[\hat{\tau}_{1},\hat{z}_{2}]+2\pi qk[\hat{z}_{2},\hat{\tau}'_{1}]+\frac{\pi^{2}}{3}q^{3}[\hat{\tau}_{1},\hat{\tau}'_{1}]}\e^{I_{q,k+1}(\hat{z}_{1},\hat{z}_{2},\hat{\tau}_{1},\hat{\tau}_{2})}\e^{\omega_{q,k+1}(x,y)}\\
&\times\e^{-\iu\pi q^{2}\hat{\tau}_{1}-\iu\pi (k+1)^{2}\hat{\tau}_{2}-2\iu q\hat{z}_{1}-2\iu (k+1)\hat{z}_{2}-2\iu q\pi\hat{\tau}'_{1}}
\end{split}
\end{equation}}
\hspace{-.25em}and so, taking into account the Lemma~\ref{lem:e^[x,y]=+-e}, we have that
\begin{equation}
\phi(\hat{x},\hat{y}+\hat{a}_{2})=\Omega'_{2}(\hat{x},\hat{y})\sum^{\infty}_{q,k=-\infty}f^{\spt{\omega}}_{q,k+1}(\hat{x},\hat{y},\theta\hat{\mathbbm{1}})=\Omega'_{2}(\hat{x},\hat{y})\phi(\hat{x},\hat{y})
\end{equation}
\,if{}f\footnote{It is not necessary to impose \mbox{$\e^{\iu 2\pi \hat{\tau}_{1}}=\pm e$} as a further condition as if $[\hat{z}_{1},\hat{z}_{2}]=0$ then \mbox{$[\hat{\tau}_{1},\hat{z}_{2}]=[\hat{\tau}'_{2},\hat{z}_{1}]$} and so, considering the Lemma~\ref{lem:e^[x,y]=+-e} and the condition \mbox{$\e^{\iu \frac{\pi}{3}  \hat{\tau}'_{2}}=\pm e$} already obtained from~\eqref{eqn:trasla1}, we find that the quantity \mbox{$(\e^{\pi[\hat{\tau}_{1},\hat{z}_{2}]})^{2q^{2}}$} in~\eqref{eqn:trasla2} is equal to $e$.}  $\e^{\iu \frac{\pi}{3}  \hat{\tau}'_{1}}$ is equal to $\pm e$ and $[\hat{z}_{1},\hat{z}_{2}]=0$ .


\section{Some details about the terms in $E$ rewritten \textit{à la} Bogomolny}
\label{sec:AppErg}
\noindent Here is the detailed calculation of $E_{\spt{12}}$ rewritten \textit{à la} Bogomolny in the hypotheses of the Theorem~\ref{thm:E12Bogo}.\\
\indent Considering the~\eqref{eqn:Fij}, the Lemma~\eqref{lem:p1dp2=p2p1d} and that\footnote{\mbox{According to\! the\! hypotheses\! $(\tilde{\phi}_{k;i},\tilde{\phi}_{l;j})\!\in\mathfrak{F}_{r_{\spt{\theta}}}$\! and\! so $\tilde{A}^{i,j}_{\spt{\substack{q,q',k,k'\\r,r',s,s'}}}(\hat{x},\hat{y},\theta\hat{\mathbbm{1}})$\! is\! null\!
(see eq.\,\eqref{eqn:[Ai,Aj]p}).}} $[\tilde{A}_{i},\tilde{A}_{j}](\hat{x},\hat{y})=0$, we can easily see how the sum of the quantities~\eqref{eqn:erg2gen} and~\eqref{eqn:erg3gen} is null.\\
The $\tr(\phi_{2}\phi^{\dagger}_{1})$ is null (see Lemma~\ref{lem:p1dp2=p2p1d} and~\ref{lem:Trphid1phi2}) and as concerns the remaining quantities (see eq.\,\eqref{eqn:erg'BRev}), we show here their representation by $\rho$. The demonstration then that {\fontsize{10}{28}\selectfont\mbox{$\tr\bigl(\hat{\partial}_{i}(\phi^{\dagger}_{1}\hat{D}_{j}\phi_{2})-\hat{\partial}_{j}(\phi^{\dagger}_{1}\hat{D}_{i}\phi_{2})\bigr)$}}, \mbox{$\tr\bigl((\hat{D}_{i}\phi_{1})^{\dagger}(\hat{D}_{i}\phi_{2})\bigr)$}, $\tr\bigl(\varepsilon_{ij}(\hat{D}_{i}\phi_{1})^{\dagger}(\hat{D}_{j}\phi_{2} )\bigr)$ and \mbox{$\tr((\phi_{2}\phi^{\dagger}_{1})^{2})$} are null, is analogous to the one of the Lemma~\ref{lem:Trphid1phi2}.\\

For \mbox{$\frac{\iu}{2}\varepsilon_{ij}\bigl(\hat{\partial}_{i}(\phi^{\dagger}_{1}\hat{D}_{j}\phi_{2})-\hat{\partial}_{j}(\phi^{\dagger}_{1}\hat{D}_{i}\phi_{2})\bigr)$}, we have that
\begin{equation}
\label{eqn:phidDphi}
\hspace{-.65cm}\hat{\partial}_{1}(\phi^{\dagger}\hat{D}_{2}\phi)-\hat{\partial}_{2}(\phi^{\dagger}\hat{D}_{1}\phi)=\frac{\iu}{a_{1}}\hat{\partial}_{1}(\phi^{\dagger}\phi\,\bm{\hat{x}}) +\frac{\iu}{a_{2}}\hat{\partial}_{2}(\phi^{\dagger}\phi\,\bm{\hat{y}})-\iu gc_{2}\hat{\partial}_{1}(\phi^{\dagger}\tilde{\phi}_{2}\tilde{\phi}^{\dagger}_{2}\phi)+\iu gc_{1}\hat{\partial}_{2}(\phi^{\dagger}\tilde{\phi}_{1}\tilde{\phi}^{\dagger}_{1}\phi ).
\end{equation}
We obtain that (see Remark~\ref{rmk:rhox2})
{\fontsize{8,5}{14}
 \selectfont
\begin{equation}
\begin{split}
\hspace{-2.9cm}\rho\bigl(\frac{\iu}{a_{1}}\hat{\partial}_{1}(\phi^{\dagger}\phi\,\bm{\hat{x}}) +\frac{\iu}{a_{2}}\hat{\partial}_{2}(\phi^{\dagger}\phi\,\bm{\hat{y}})\bigr)=\frac{\iu\theta}{a_{1}a_{2}}\sum^{+\infty}_{l,l',m,m'=-\infty}&\biggl(2+\iu\pi r_{\spt{\theta}}\Bigl(ll'\frac{k_{\spt{\alpha_{1}}}}{2}+(l'k_{\spt{\alpha_{1}}}+6m'k_{\spt{\alpha_{2}}})\frac{x}{a_{1}}\Bigr)+\iu\pi r_{\spt{\theta}}\Bigl(mm'\frac{k_{\spt{\beta_{2}}}}{2}+(6l'k_{\spt{\beta_{1}}}+m'k_{\spt{\beta_{2}}})\frac{y}{a_{2}}\Bigr)\biggr)\\
&\times G'^{\spt{\omega}}_{l,l',m'}(x/a_{1},\theta,k_{\spt{\alpha_{1}}},k_{\spt{\alpha_{2}}},k_{\spt{\beta_{2}}},\omega_{1}(a_{1}))G'^{\spt{\omega}}_{m,m',l'}(y/a_{2},\theta,k_{\spt{\beta_{2}}},k_{\spt{\beta_{1}}},-k_{\spt{\alpha_{1}}},\omega_{2}(a_{2}))\label{eqn:trphidphix}
\end{split}
\end{equation}}
\hspace{-.3em}and that 
{\fontsize{7.8}{18}
 \selectfont
\begin{equation}
\begin{split}
\hspace{-3.1cm}\rho\bigl(\iu g\hat{\partial}_{2}(\phi^{\dagger}_{1}\tilde{\phi}_{1}\tilde{\phi}^{\dagger}_{1}\phi_{2}(\hat{x},&\hat{y}))-\iu g\hat{\partial}_{1}(\phi^{\dagger}_{1}\tilde{\phi}_{2}\tilde{\phi}^{\dagger}_{2}\phi_{2}(\hat{x},\hat{y}))\bigr)=\hspace{-.7cm}\sum^{+\infty}_{l,l',m,m',r,r',s,s'=-\infty}\hspace{-.05cm}\biggl(-\frac{\pi gc_{1}\theta_{\cc}}{a_{1}}\Bigl((mm'+ss')\frac{k_{\spt{\beta_{2}}}}{2}+\bigl(6(l'+r')k_{\spt{\beta_{1}}}+(m'+s')k_{\spt{\beta_{2}}}\bigr)\frac{y}{a_{2}}\Bigr) \\
\hspace{-3.1cm}&\times\tilde{G}'^{\spt{\omega}}_{m,m',l,l',s,s',r,r'}(y/a_{2},\theta,k_{\spt{\beta_{2}}},k_{\spt{\beta_{1}}},-k_{\spt{\alpha_{1}}},\omega_{2}(a_{2}),\tilde{\omega}_{2;1}(a_{2}))\tilde{G}'^{\spt{\omega}}_{l,l',m,m',r,r',s,'s'}(x/a_{1},\theta,k_{\spt{\alpha_{1}}},k_{\spt{\alpha_{2}}},k_{\spt{\beta_{2}}},\omega_{1}(a_{1}),\tilde{\omega}_{1;1}(a_{1}))+\\
\hspace{-3.1cm}&+\frac{\pi gc_{2}\theta_{\cc}}{a_{2}}\Bigl((ll'+rr')\frac{k_{\spt{\alpha_{1}}}}{2}+\bigl((l'+r')k_{\spt{\alpha_{1}}}+6(m'+s')k_{\spt{\alpha_{2}}}\bigr)\frac{x}{a_{1}}\Bigr) \tilde{G}'^{\spt{\omega}}_{l,l',m,m',r,r',s,'s'}(x/a_{1},\theta,k_{\spt{\alpha_{1}}},k_{\spt{\alpha_{2}}},k_{\spt{\beta_{2}}},\omega_{1}(a_{1}),\tilde{\omega}_{1;2}(a_{1}))\\
\hspace{-3.1cm}&\times\tilde{G}'^{\spt{\omega}}_{m,m',l,l',s,s',r,r'}(y/a_{2},\theta,k_{\spt{\beta_{2}}},k_{\spt{\beta_{1}}},-k_{\spt{\alpha_{1}}},\omega_{2}(a_{2}),\tilde{\omega}_{2;2}(a_{2}))\biggr)
\end{split}
\end{equation}}
\hspace{-.3em}with
{\fontsize{9}{15}
 \selectfont
\begin{equation}
\begin{split}
\hspace{-3cm}\tilde{G}'^{\spt{\omega}}_{l,l',m,m',r,r',s,'s'}(x/a_{1},\theta,k_{\spt{\alpha_{1}}},k_{\spt{\alpha_{2}}},k_{\spt{\beta_{2}}},\omega_{1}(a_{1}),\tilde{\omega}_{1;1}(a_{1}))=(-)^{l+r}&\e^{\Re\bigl(\omega'_{r,r'}(x/a_{1},\tilde{\omega}_{1;1}(a_{1}))+\omega'_{l,l'}(x/a_{1},\omega_{1}(a_{1}))\bigr)}\e^{\iu\frac{\pi}{2}\bigl(ll'(k_{\spt{\alpha_{1}}}r_{\spt{\theta}}+2\Re(\omega_{1}(a_{1})))\bigr)}\\
&\times\e^{\iu\frac{\pi}{2}\bigl(rr'(k_{\spt{\alpha_{1}}}r_{\spt{\theta}}+2\Re(\tilde{\omega}_{1;i}(a_{1})))\bigr)}\e^{\iu \tilde{K}^{\spt{\omega}}_{l',m,m',r',s,s'}\frac{x}{a_{1}}}
\end{split}
\end{equation}}
\hspace{-.3em}where
{\fontsize{7,25}{15}
 \selectfont
\begin{equation}
\hspace{-3.65cm}\tilde{K}^{\spt{\omega}}_{l',m,m',r',s,s'}(k_{\spt{\alpha_{1}}},k_{\spt{\alpha_{2}}},k_{\spt{\beta_{2}}},\omega_{1}(a_{1}),\tilde{\omega}_{1;1}(a_{1}))=\pi\biggl(2l'\Re(\omega_{1}(a_{1}))+2r'\Re(\tilde{\omega}_{1;1}(a_{1}))+\frac{\pi}{2}k_{\spt{\alpha_{2}}}k_{\spt{\beta_{2}}}r^{2}_{\spt{\theta}}\theta\bigl(m'^{3}+3m's'(m-s)+s'^{3}\bigr)+\bigl((l'+r')k_{\spt{\alpha_{1}}}+6(m'+s')k_{\spt{\alpha_{2}}}\bigr)r_{\spt{\theta}}\biggr).
\end{equation}}
\hspace{-.3em}For the demonstration that {\fontsize{10}{28}
 \selectfont\mbox{$\tr\bigl(\hat{\partial}_{i}(\phi^{\dagger}_{1}\hat{D}_{j}\phi_{2})-\hat{\partial}_{j}(\phi^{\dagger}_{1}\hat{D}_{i}\phi_{2})\bigr)$}} is null, it's useful to consider the~\eqref{eqn:Inta1G'} and that 
{\fontsize{8}{12.5}
 \selectfont
\begin{equation}
\begin{split}
\label{eqn:IntG8}
\hspace{-3.5cm}\int^{a_{1}}_{0}\tilde{G}'^{\spt{\omega}}_{l,l',m,m',r,r',s,'s'}&(x/a_{1},\theta,k_{\spt{\alpha_{1}}},k_{\spt{\alpha_{2}}},k_{\spt{\beta_{2}}},\omega_{1}(a_{1}),\tilde{\omega}_{1;1}(a_{1}))\,dx=(-1)^{l+r}a_{1}\tilde{A}_{l,l',m,m',r,r',s,s'}(k_{\spt{\alpha_{1}}},k_{\spt{\alpha_{2}}},k_{\spt{\beta_{2}}},\omega_{1}(a_{1}),\tilde{\omega}_{1;1}(a_{1}))\\
&\times\bigl(\erf(b_{l,l',m,m',r,r',s,s'}(k_{\spt{\alpha_{1}}},k_{\spt{\alpha_{2}}},k_{\spt{\beta_{2}}},\omega_{1}(a_{1}),\tilde{\omega}_{1;1}(a_{1}))-\erf(b_{l+2,l',m,m',r+2,r',s,s'}(k_{\spt{\alpha_{1}}},k_{\spt{\alpha_{2}}},k_{\spt{\beta_{2}}},\omega_{1}(a_{1}),\tilde{\omega}_{1;1}(a_{1}))\bigr),
\end{split}
\end{equation}}
\hspace{-.3em}with
{\fontsize{8}{20}
 \selectfont
\begin{equation}
\begin{split}
\hspace{-3.5cm}\tilde{A}_{l,l',m,m',r,r',s,s'}(k_{\spt{\alpha_{1}}},k_{\spt{\alpha_{2}}},k_{\spt{\beta_{2}}},\omega_{1}(a_{1}),\tilde{\omega}_{1;1}(a_{1}))=&\frac{\e^{\frac{\pi}{2}\bigl(l'^{2}\Im(\omega_{1}(a_{1}))+r'^{2}\Im(\tilde{\omega}_{1;1}(a_{1}))+(l-r)^{2}\frac{\Im(\omega_{1}(a_{1}))\Im(\tilde{\omega}_{1;1}(a_{1}))}{\Im(\omega_{1}(a_{1}))+\Im(\tilde{\omega}_{1;1}(a_{1}))}+\frac{(\tilde{K}^{\spt{\omega}}_{l',m,m',r',s,s'})^{2}}{4\pi^{2} (\Im(\omega_{1}(a_{1}))+\Im(\tilde{\omega}_{1;1}(a_{1})))}\bigr)}}{2\sqrt{2(\abs{\Im(\omega_{1}(a_{1}))}+\abs{\Im(\tilde{\omega}_{1;1}(a_{1}))})}}\\
\hspace{-.5cm}&\times\e^{\frac{\iu }{2}\Bigl(\pi k_{\spt{\alpha_{1}}}r_{\spt{\theta}}(ll'+rr')+2\pi (ll'\Re(\omega_{1}(a_{1}))+rr'\Re(\tilde{\omega}_{1;1}(a_{1})))-\tilde{K}^{\spt{\omega }}_{l',m,m',r',s,s'}\frac{l\Im(\omega_{1}(a_{1}))+r\Im(\tilde{\omega}_{1;1}(a_{1}))}{\Im(\omega_{1}(a_{1}))+\Im(\tilde{\omega}_{1;1}(a_{1}))}\Bigr)}
\end{split}
\end{equation}}
\hspace{-.3em}and
{\fontsize{10}{15}
 \selectfont
\begin{equation}
\hspace{-1.5cm}b_{l,l',m,m',r,r',s,s'}(k_{\spt{\alpha_{1}}},k_{\spt{\alpha_{2}}},k_{\spt{\beta_{2}}},\omega_{1}(a_{1}),\tilde{\omega}_{1;1}(a_{1}))=\frac{2\pi\bigl( l\Im(\omega_{1}(a_{1}))+ r\Im(\tilde{\omega}_{1;1}(a_{1}))\bigr)+\iu \tilde{K}^{\spt{\omega }}_{l',m,m',r',s,s'}}{2\sqrt{2\pi\bigl(\abs{\Im(\omega_{1}(a_{1}))}+\abs{\Im(\tilde{\omega}_{1;1}(a_{1}))}\bigr)}}.
\end{equation}}
\indent As concerns \mbox{$(\hat{D}_{i}\phi_{1})^{\dagger}(\hat{D}_{i}\phi_{2})$} and $\varepsilon_{ij}(\hat{D}_{i}\phi_{1})^{\dagger}(\hat{D}_{j}\phi_{2} )$, taking into account the~\eqref{eqn:rho(x,x2)} and defining  
 \begin{subequations}
\label{eqn:ab}
\begin{align}
&a_{l,l',m'}(x,k_{\spt{\alpha_{1}}},k_{\spt{\alpha_{2}}})=r_{\spt{\theta}}\Bigl(ll'\frac{k_{\spt{\alpha_{1}}}}{2}+\bigl(l'k_{\spt{\alpha_{1}}}+6m'k_{\spt{\alpha_{2}}}\bigr)\frac{x}{a_{1}}\Bigr),\\
&b_{l',m,m'}(y,k_{\spt{\beta_{1}}},k_{\spt{\beta_{2}}})= r_{\spt{\theta}} \Bigl(mm'\frac{k_{\spt{\beta_{2}}}}{2}+\bigl(6 l'k_{\spt{\beta_{1}}}+m' k_{\spt{\beta_{2}}}\bigr) \frac{y}{a_{2}}\Bigr),
\end{align} 
\end{subequations}
we obtain that 
{\fontsize{7.8}{15}
 \selectfont
\begin{equation}
\label{eqn:rho(DiDi)}
\begin{split}
\hspace{-3.3cm}&\rho\bigl((\hat{D}_{i}\phi)^{\dagger}(\hat{D}_{i}\phi)\bigr)(x,y) =\hspace{-1em}\sum^{+\infty}_{l,l',m,m'=-\infty}\hspace{-.5em}\biggl(\frac{\pi\theta}{ a^{2}_{1}}\Bigl(1+(1-\pi\theta)b_{l',m,m'}(y,k_{\spt{\beta_{1}}},k_{\spt{\beta_{2}}})-\frac{\pi}{2}b^{2}_{l',m,m'}(y,k_{\spt{\beta_{1}}},k_{\spt{\beta_{2}}})\theta\bigl(1-\frac{\pi}{2}\theta\bigr)\Bigr)+\\
\hspace{-3.3cm}&\hspace{-.75cm}+\frac{\pi\theta}{ a^{2}_{2}}\Bigl(1-(1-\pi\theta)a_{l,l',m'}(x,k_{\spt{\alpha_{1}}},k_{\spt{\alpha_{2}}})-\frac{\pi}{2} a^{2}_{l,l',m'}(x,k_{\spt{\alpha_{1}}},k_{\spt{\alpha_{2}}})\theta\bigl(1-\frac{\pi}{2}\theta\bigr)\Bigr)\biggr)G'^{\spt{\omega}}_{l,l',m'}(x/a_{1},\theta,k_{\spt{\alpha_{1}}},k_{\spt{\alpha_{2}}},k_{\spt{\beta_{2}}},\omega_{1}(a_{1}))G'^{\spt{\omega}}_{m,m',l'}(y/a_{2},\theta,k_{\spt{\beta_{2}}},k_{\spt{\beta_{1}}},-k_{\spt{\alpha_{1}}},\omega_{2}(a_{2}))+\\
\hspace{-3.3cm}&+2g\hspace{-1em}\sum^{+\infty}_{l,l',m,m',r,r',s,s'=-\infty}\hspace{-.5em}\biggl(\frac{c_{1}}{a_{2}}\check{G}'^{\spt{\omega}}_{l,l',m,m',r,r',s,'s'}(x/a_{1},\theta,k_{\spt{\alpha_{1}}},k_{\spt{\alpha_{2}}},k_{\spt{\beta_{2}}},\omega_{1}(a_{1}),\tilde{\omega}_{1;1}(a_{1}))\check{G}'^{\spt{\omega}}_{m,m',l,l',s,s',r,r'}(y/a_{2},\theta,k_{\spt{\beta_{2}}},k_{\spt{\beta_{1}}},-k_{\spt{\alpha_{1}}},\omega_{2}(a_{2}),\tilde{\omega}_{2;1}(a_{2}))+\\
\hspace{-3.3cm}&\hspace{2.5cm}-\frac{c_{2}}{a_{1}}\check{G}'^{\spt{\omega}}_{l,l',m,m',r,r',s,'s'}(x/a_{1},\theta,k_{\spt{\alpha_{1}}},k_{\spt{\alpha_{2}}},k_{\spt{\beta_{2}}},\omega_{1}(a_{1}),\tilde{\omega}_{1;2}(a_{1}))\check{G}'^{\spt{\omega}}_{m,m',l,l',s,s',r,r'}(y/a_{2},\theta,k_{\spt{\beta_{2}}},k_{\spt{\beta_{1}}},-k_{\spt{\alpha_{1}}},\omega_{2}(a_{2}),\tilde{\omega}_{2;2}(a_{2}))\biggr)+\\
\hspace{-3.3cm}&\hspace{0.25cm}+g^{2}c^{2}_{i}\hspace{-5em}\sum^{+\infty}_{l,l',m,m',r,r',s,s',u,u',v,v'=-\infty}\hspace{-5em}\mathring{G}'^{\spt{\omega}}_{l,l',m,m',r,r',s,'s',u,u',v,v'}(x/a_{1},\theta,k_{\spt{\alpha_{1}}},k_{\spt{\alpha_{2}}},k_{\spt{\beta_{2}}},\omega_{1}(a_{1}),\tilde{\omega}_{1;i}(a_{1}))\mathring{G}'^{\spt{\omega}}_{m,m',l,l',s,s',r,r',v,v',u,u'}(y/a_{2},\theta,k_{\spt{\beta_{2}}},k_{\spt{\beta_{1}}},-k_{\spt{\alpha_{1}}},\omega_{2}(a_{2}),\tilde{\omega}_{2;i}(a_{2})),
\end{split} 
\end{equation}}
\hspace{-.3em}and
{\fontsize{7.8}{15}
 \selectfont
\begin{equation}
\begin{split}
\hspace{-3.2cm}\varepsilon_{ij}\rho((\hat{D}_{i}\phi)^{\dagger}(\hat{D}_{j}\phi )(\hat{x},\hat{y}))=&\frac{\iu\theta}{a_{1}a_{2}}\sum^{+\infty}_{q,q',k,k'=-\infty}\hspace{-1em}\Bigl(1+\iu\pi a_{l,l',m'}(x,k_{\spt{\alpha_{1}}},k_{\spt{\alpha_{2}}})+\iu\pi b_{l',m,m'}(y,k_{\spt{\beta_{1}}},k_{\spt{\beta_{2}}})\Bigr)G'^{\spt{\omega}}_{l,l',m'}(x/a_{1},\theta,k_{\spt{\alpha_{1}}},k_{\spt{\alpha_{2}}},k_{\spt{\beta_{2}}},\omega_{1}(a_{1}))\\
\hspace{-3.2cm}&\hspace{-4.15cm}\times G'^{\spt{\omega}}_{m,m',l'}(y/a_{2},\theta,k_{\spt{\beta_{2}}},k_{\spt{\beta_{1}}},-k_{\spt{\alpha_{1}}},\omega_{2}(a_{2}))
 +\pi g\theta\hspace{-3.5em}\sum^{+\infty}_{l,l',m,m',r,r',s,s'=-\infty}\hspace{-.5em}\biggl(\frac{c_{2}}{a_{2}}(a_{l,l',m'}+a_{r,r',s'})(x,k_{\spt{\alpha_{1}}},k_{\spt{\alpha_{2}}})\check{G}'^{\spt{\omega}}_{l,l',m,m',r,r',s,'s'}(x/a_{1},\theta,k_{\spt{\alpha_{1}}},k_{\spt{\alpha_{2}}},k_{\spt{\beta_{2}}},\omega_{1}(a_{1}),\tilde{\omega}_{1;2}(a_{1}))\\
\hspace{-3.2cm}&\hspace{-4.15cm}\times\check{G}'^{\spt{\omega}}_{m,m',l,l',s,s',r,r'}(y/a_{2},\theta,k_{\spt{\beta_{2}}},k_{\spt{\beta_{1}}},-k_{\spt{\alpha_{1}}},\omega_{2}(a_{2}),\tilde{\omega}_{2;2}(a_{2}))-\frac{c_{1}}{a_{1}}(b_{l',m,m'}+b_{r,r',s'})(y,k_{\spt{\beta_{1}}},k_{\spt{\beta_{2}}})\check{G}'^{\spt{\omega}}_{l,l',m,m',r,r',s,'s'}(x/a_{1},\theta,k_{\spt{\alpha_{1}}},k_{\spt{\alpha_{2}}},k_{\spt{\beta_{2}}},\omega_{1}(a_{1}),\tilde{\omega}_{1;1}(a_{1}))\\
\hspace{-3.2cm}&\hspace{-4.15cm}\times\check{G}'^{\spt{\omega}}_{m,m',l,l',s,s',r,r'}(y/a_{2},\theta,k_{\spt{\beta_{2}}},k_{\spt{\beta_{1}}},-k_{\spt{\alpha_{1}}},\omega_{2}(a_{2}),\tilde{\omega}_{2;1}(a_{2}))\biggr)
\end{split}
\end{equation}}
\hspace{-.3em}where
{\fontsize{8}{20}
 \selectfont
\begin{subequations}
\begin{align}
\hspace{-3cm}&\check{G}'^{\spt{\omega}}_{l,l',m,m',r,r',s,'s'}(x/a_{1},\theta,k_{\spt{\alpha_{1}}},k_{\spt{\alpha_{2}}},k_{\spt{\beta_{2}}},\omega_{1}(a_{1}),\tilde{\omega}_{1;1}(a_{1}))=(-)^{l+r}\e^{\Re\bigl(\omega'_{r,r'}(x/a_{1},\tilde{\omega}_{1;1}(a_{1}))+\omega'_{l,l'}(x/a_{1},\omega_{1}(a_{1}))\bigr)}\e^{\iu\frac{\pi}{2}\bigl(ll'(k_{\spt{\alpha_{1}}}r_{\spt{\theta}}+2\Re(\omega_{1}(a_{1})))\bigr)}\notag\\
\hspace{-3cm}&\hspace{8cm}\times\e^{\iu\frac{\pi}{2}\bigl(rr'(k_{\spt{\alpha_{1}}}r_{\spt{\theta}}+2\Re(\tilde{\omega}_{1;1}(a_{1})))\bigr)}\e^{\iu \check{K}^{\spt{\omega}}_{l',m,m',r',s,s'}\frac{x}{a_{1}}},\\
\hspace{-3cm}&\mathring{G}'^{\spt{\omega}}_{l,l',m,m',r,r',s,'s',u,u',v,v'}(x/a_{1},\theta,k_{\spt{\alpha_{1}}},k_{\spt{\alpha_{2}}},k_{\spt{\beta_{2}}},\omega_{1}(a_{1}),\tilde{\omega}_{1;1}(a_{1}))=(-)^{l+r+u}\e^{\Re\bigl(\omega'_{l,l'}(x/a_{1},\omega_{1}(a_{1}))+\omega'_{r,r'}(x/a_{1},\tilde{\omega}_{1;1}(a_{1}))+\omega'_{u,u'}(x/a_{1},\tilde{\omega}_{1;1}(a_{1}))\bigr)}\notag\\
\hspace{-3cm}&\hspace{7cm}\times\e^{\iu\frac{\pi}{2}\bigl(ll'(k_{\spt{\alpha_{1}}}r_{\spt{\theta}}+2\Re(\omega_{1}(a_{1})))+(rr'+uu')(k_{\spt{\alpha_{1}}}r_{\spt{\theta}}+2\Re(\tilde{\omega}_{1;1}(a_{1})))\bigr)}\e^{\iu \mathring{K}^{\spt{\omega}}_{l',m,m',r',s,s',u',v,v'}\frac{x}{a_{1}}},
\end{align} 
\end{subequations}}
\hspace{-.3em}with 
{\fontsize{8}{15}
 \selectfont
\begin{subequations}
\begin{align}
\hspace{-3.35cm}&\check{K}^{\spt{\omega}}_{l',m,m',r',s,s'}(k_{\spt{\alpha_{1}}},k_{\spt{\alpha_{2}}},k_{\spt{\beta_{2}}},\omega_{1}(a_{1}),\tilde{\omega}_{1;1}(a_{1}))=\pi\biggl(2l'\Re(\omega_{1}(a_{1}))+2r'\Re(\tilde{\omega}_{1;1}(a_{1}))+\frac{\pi}{4}k_{\spt{\alpha_{2}}}k_{\spt{\beta_{2}}}r^{2}_{\spt{\theta}}\theta\bigl(m'^{2}(2m'+3s')-s'(s'^{2}-3(m-s)^{2})\bigr)+\notag\\
\hspace{-3.35cm}&\hspace{7cm}+\bigl((l'+r')k_{\spt{\alpha_{1}}}+6(m'+s')k_{\spt{\alpha_{2}}}\bigr)r_{\spt{\theta}}\biggr),\\
\hspace{-3.35cm}&\mathring{K}^{\spt{\omega}}_{l',m,m',r',s,s',u',v,v'}(k_{\spt{\alpha_{1}}},k_{\spt{\alpha_{2}}},k_{\spt{\beta_{2}}},\omega_{1}(a_{1}),\tilde{\omega}_{1;1}(a_{1}))=\pi\biggl(2l'\Re(\omega_{1}(a_{1}))+2(r'+u')\Re(\tilde{\omega}_{1;1}(a_{1}))+\frac{\pi}{2}k_{\spt{\alpha_{2}}}k_{\spt{\beta_{2}}}r^{2}_{\spt{\theta}}\theta\bigl(m'^{3}+s'^{3}+v'^{3}+3s'v'(s-v)+\notag\\
\hspace{-3.35cm}&\hspace{7cm}+3m'(s'(m-s)+v'(m-v))\bigr)+\bigl((l'+r'+u')k_{\spt{\alpha_{1}}}+6(m'+s'+v')k_{\spt{\alpha_{2}}}\bigr)r_{\spt{\theta}}\biggr).
\end{align} 
\end{subequations}}
\hspace{-.3em}This time to show that \mbox{$\tr\bigl((\hat{D}_{i}\phi_{1})^{\dagger}(\hat{D}_{i}\phi_{2})\bigr)$} and $\tr\bigl(\varepsilon_{ij}(\hat{D}_{i}\phi_{1})^{\dagger}(\hat{D}_{j}\phi_{2} )\bigr)$ are null, apart from considering the eqs.\,\eqref{eqn:Inta1G'} and~\eqref{eqn:IntG8}, we have to consider that
{\fontsize{8}{13}
 \selectfont
\begin{equation}
\begin{split}
\label{eqn:IntG12}
\hspace{-3.45cm}\int^{a_{1}}_{0}&\mathring{G}'^{\spt{\omega}}_{l,l',m,m',r,r',s,'s',u,u',v,v'}(x/a_{1},\theta,k_{\spt{\alpha_{1}}},k_{\spt{\alpha_{2}}},k_{\spt{\beta_{2}}},\omega_{1}(a_{1}),\tilde{\omega}_{1;1}(a_{1}))\,dx=(-)^{l+r+u}a_{1}\mathring{A}_{l,l',m,m',r,r',s,s',u,u',v,v'}(k_{\spt{\alpha_{1}}},k_{\spt{\alpha_{2}}},k_{\spt{\beta_{2}}},\omega_{1}(a_{1}),\tilde{\omega}_{1;1}(a_{1}))\\
\hspace{-3.45cm}&\times\bigl(\erf(\mathring{b}_{l,l',m,m',r,r',s,s',u,u',v,v'}(k_{\spt{\alpha_{1}}},k_{\spt{\alpha_{2}}},k_{\spt{\beta_{2}}},\omega_{1}(a_{1}),\tilde{\omega}_{1;1}(a_{1}))-\erf(\mathring{b}_{l+2,l',m,m',r+2,r',s,s',u+2,u',v,v'}(k_{\spt{\alpha_{1}}},k_{\spt{\alpha_{2}}},k_{\spt{\beta_{2}}},\omega_{1}(a_{1}),\tilde{\omega}_{1;1}(a_{1}))\bigr),
\end{split}
\end{equation}}
\hspace{-.3em}with \mbox{$\mathring{A}_{l,l',m,m',r,r',s,s',u,u',v,v'}(k_{\spt{\alpha_{1}}},k_{\spt{\alpha_{2}}},k_{\spt{\beta_{2}}},\omega_{1}(a_{1}),\tilde{\omega}_{1;1}(a_{1}))$} equal to 
\bigskip
{\fontsize{9}{20}
 \selectfont
\begin{equation}
\begin{split}
\hspace{-2.25cm}&\frac{\e^{\frac{(\mathring{K}^{\spt{\omega}}_{l',m,m',r',s,s',u',v,v'})^{2}+4\pi^{2}\bigl(((r-u)^{2}+2(r'^{2}+u'^{2}))(\Im(\tilde{\omega}_{1;1}(a_{1})))^{2}+((l-r)^{2}+(l-u)^{2}+2l'^{2}+r'^{2}+u'^{2})\Im(\tilde{\omega}_{1;1}(a_{1}))\Im(\omega_{1}(a_{1}))+l'^{2}(\Im(\omega_{1}(a_{1})))^{2}\bigr)}{8\pi\bigl(\Im(\omega_{1}(a_{1}))+2\Im(\tilde{\omega}_{1;1}(a_{1}))\bigr)}}}{2\sqrt{2(\abs{\Im(\omega_{1}(a_{1}))}+2\abs{\Im(\tilde{\omega}_{1;1}(a_{1}))})}}\\
\hspace{-2.25cm}&\hspace{1cm}\times\e^{\frac{\iu\Im(\omega_{1}(a_{1}))}{2(\abs{\Im(\omega_{1}(a_{1}))}+2\abs{\Im(\tilde{\omega}_{1;1}(a_{1}))})}\bigl(\pi k_{\spt{\alpha_{1}}}r_{\spt{\theta}}(ll'+rr'+uu')+2\pi (ll'\Re(\omega_{1}(a_{1}))+(rr'+uu')\Re(\tilde{\omega}_{1;1}(a_{1})))-l\mathring{K}^{\spt{\omega}}_{l',m,m',r',s,s',u',v,v'}\bigr)}\\
\hspace{-2.25cm}&\hspace{1cm}\times\e^{\frac{\iu\Im(\tilde{\omega}_{1;1}(a_{1}))}{2(\abs{\Im(\omega_{1}(a_{1}))}+2\abs{\Im(\tilde{\omega}_{1;1}(a_{1}))})}\bigl(2\pi k_{\spt{\alpha_{1}}}r_{\spt{\theta}}(ll'+rr'+uu')+4\pi (ll'\Re(\omega_{1}(a_{1}))+(rr'+uu')\Re(\tilde{\omega}_{1;1}(a_{1})))-(r+u)\mathring{K}^{\spt{\omega}}_{l',m,m',r',s,s',u',v,v'}\bigr)}
\end{split}
\end{equation}}
\hspace{-.3em}and 
{\fontsize{10}{15}
 \selectfont
\begin{equation}
\hspace{-3cm}\mathring{b}_{l,l',m,m',r,r',s,s',u,u',v,v'}(k_{\spt{\alpha_{1}}},k_{\spt{\alpha_{2}}},k_{\spt{\beta_{2}}},\omega_{1}(a_{1}),\tilde{\omega}_{1;1}(a_{1}))=\frac{2\pi\bigl( l\Im(\omega_{1}(a_{1}))+ (r+u)\Im(\tilde{\omega}_{1;1}(a_{1}))\bigr)+\iu \mathring{K}^{\spt{\omega}}_{l',m,m',r',s,s',u',v,v'}}{2\sqrt{2\pi\bigl(\abs{\Im(\omega_{1}(a_{1}))}+2\abs{\Im(\tilde{\omega}_{1;1}(a_{1}))}\bigr)}}.
\end{equation}}
\indent In the end, we obtain
{\fontsize{9}{15}
 \selectfont
\begin{equation}
\hspace{-3.3cm}\rho((\phi\phi^{\dagger})^{2}(\hat{x},\hat{y}))=\hspace{-3.15em}\sum^{+\infty}_{l,l',m,m',r,r',s,s'=-\infty}\hspace{-3em}\tilde{G}'^{\spt{\omega}}_{l,l',m,m',r,r',s,'s'}(x/a_{1},\theta,k_{\spt{\alpha_{1}}},k_{\spt{\alpha_{2}}},k_{\spt{\beta_{2}}},\omega_{1}(a_{1}),\omega_{1}(a_{1}))\tilde{G}'^{\spt{\omega}}_{m,m',l,l',s,s',r,r'}(y/a_{2},\theta,k_{\spt{\beta_{2}}},k_{\spt{\beta_{1}}},-k_{\spt{\alpha_{1}}},\omega_{2}(a_{2}),\omega_{2}(a_{2})).
\end{equation}}
\hspace{-.3em}For the demonstration that \mbox{$\tr((\phi_{2}\phi^{\dagger}_{1})^{2})$} is null, we can just consider the~\eqref{eqn:IntG8}.


\addcontentsline{toc}{section}{\refname}

\bibliographystyle{utphys}

\bibliography{Bibliography}

\providecommand{\href}[2]{#2}\begingroup\raggedright\begin{thebibliography}{10}

\bibitem{Manton2004}
N.~S. Manton and P.~Sutcliffe, {\em Topological solitons}.
\newblock Cambridge University Press, Cambridge U.K., 2004.

\bibitem{Vilenkin1994}
A.~Vilenkin and E.~P.~S. Shellard, {\em Cosmic strings and other topological
  defects}.
\newblock Cambridge University Press, Cambridge U.K., 1994.

\bibitem{EWeinberg2012}
E.~J. Weinberg, {\em Classical solutions in quantum field theory:\,solitons and
  instantons in high energy physics}.
\newblock Cambridge monographs on mathematical physics. Cambridge University
  Press, 2012.

\bibitem{Hooft2007}
G.~{'t Hooft}, {\it Models for conf{}inement},  {\em Prog.Theor.Phys.Suppl.}
  {\bf 167} (2007) 144.

\bibitem{Mandelstam1975}
S.~Mandelstam, {\it Vortices and quark conf{}inement in non-abelian gauge
  theories},  {\em Physics Letters B} {\bf 53} (1975) 476.

\bibitem{Hooft75}
G.~{'t Hooft}, {\it Gauge theories with unified weak, electromagnetic and
  strong interactions},  in {\em E.P.S. Int. Conf. on High Energy Physics,
  Palermo, 23-28 June 1975}, Editrice Compositori, A. Zichichi Ed., Bologna,
  1976.

\bibitem{Parisi75}
G.~Parisi, {\it Quark imprisonment and vacuum repulsion.},  {\em Phys. Rev. D}
  {\bf 11} (1975) 970.

\bibitem{Hooft1979}
G.~{'t Hooft}, {\it A property of electric and magnetic flux in non-abelian
  gauge theories},  {\em Nucl. Phys. B} {\bf 153} (1979) 141.

\bibitem{Hooft1981}
G.~{'t Hooft}, {\it Topology of the gauge condition and new confinement phases
  in non-abelian gauge theories},  {\em Nucl. Phys. B} {\bf 190} (1981) 455.

\bibitem{Baal82}
P.~van Baal, {\it Some results for {SU(N) Gauge-Fields on the Hypertorus}},
  {\em Commun. Math. Phys.} {\bf 85} (1982) 529.

\bibitem{Bogomolny1976}
E.~B. Bogomolny, {\it Stability of classical solutions},  {\em Sov. J. Nucl.
  Phys.} {\bf 24} (1976) 449.

\bibitem{Bradlow90}
S.~B. Bradlow, {\it Vortices in holomorphic line bundles over closed k{\"a}hler
  manifolds},  {\em Commun. Math. Phys.} {\bf 135} (1990) 1.

\bibitem{Lozano:2006xn}
G.~S. Lozano, D.~Marques, and F.~Schaposnik, {\it Vortex solutions in the
  noncommutative torus},  {\em JHEP} {\bf 0609} (2006) 044,
  [\href{http://xxx.lanl.gov/abs/hep-th/0606099}{{\tt hep-th/0606099}}].

\bibitem{Weyl31}
H.~Weyl, {\em The theory of groups and Quantum Mechanics}.
\newblock Dover, 1931.

\bibitem{Gro46}
H.~Gr$\mathrm{\ddot{o}}$newold, {\it On principles of quantum mechanics},  {\em
  Physica} {\bf 12} (1946) 405.

\bibitem{Moyal49}
J.~E. Moyal, {\it Quantum mechanics as a statistical theory},  {\em Proc.
  Cambridge Phil. Soc.} {\bf 45} (1949) 99.

\bibitem{Gonzalez07}
A.~Gonz{\'a}lez-Arroyo and A.~Ramos, {\it Dynamics of critical vortices on the
  torus and on the plane},  {\em JHEP} {\bf 0701} (2007) 054,
  [\href{http://xxx.lanl.gov/abs/hep-th/0610294}{{\tt hep-th/0610294}}].

\bibitem{Forgacs2005}
P.~Forgacs, G.~S. Lozano, E.~F. Moreno, and F.~A. Schaposnik, {\it Bogomolny
  equations for vortices in the noncommutative torus},  {\em JHEP} {\bf 07}
  (2005) 074, [\href{http://xxx.lanl.gov/abs/hep-th/0503168}{{\tt
  hep-th/0503168}}].

\bibitem{Arroyo98}
A.~Gonz{\'a}lez-Arroyo, {\it {Yang-Mills f{}ields on the 4-dimensional torus.
  Part I: Classical Theory}},  in {\em Proceedings of the Pe{\~n}iscola 1997
  advanced school on non-perturbative quantum field physics}, World Scientific,
  Singapore, 1998.

\bibitem{Helg78}
S.~Helgason, {\em Differential Geometry, Lie Groups, and Symmetric Spaces}.
\newblock Pure and applied mathematics. Academic Press, London, 1978.

\bibitem{Bondia2001}
J.~M. Gracia-Bond{\'i}a, J.~C. V{\'a}rilly, and H.~Figueroa, {\em Elements of
  Noncommutative Geometry}.
\newblock Birkh{\"a}user Basel, 2001.

\bibitem{Arroyo04}
A.~Gonz{\'a}lez-Arroyo and A.~Ramos, {\it Expansion for the solutions of the
  {Bogomolny} equations on the torus},  {\em JHEP} {\bf 0407} (2004) 008,
  [\href{http://xxx.lanl.gov/abs/hep-th/0404022}{{\tt hep-th/0404022}}].

\bibitem{Madore1995}
J.~Madore, {\em An Introduction to Noncommutative Geometry and its Physical
  Applications}.
\newblock LMS Lecture Notes 206, 1995.

\bibitem{Marmo2010}
A.~P. Balachandran, S.~G. Jo, and G.~Marmo, {\em Group Theory and Hopf Algebra:
  Lectures for Physicists}.
\newblock World Scientific, 2010.

\bibitem{Eguchi80}
T.~Eguchi, P.~B. Gilkey, and A.~J. Hanson, {\it Gravitation, gauge theories and
  differential geometry},  {\em Phys. Rep.} {\bf 66} (1980) 213.

\bibitem{Dynkin47}
E.~B. Dynkin, {\it Calculation of the coef{}f{}icients in the
  {Campbell-Hausdorf{}f} formula},  in {\em \emph{Selected papers of E. B.
  Dynkin with commentary}}.
\newblock American Mathematical Society, 2000.

\bibitem{Pachol:2013}
A.~Pacho{\l}, {\it Short review on noncommutative spacetimes},  {\em
  J.Phys.Conf.Ser.} {\bf 442} (2013) 012039.

\bibitem{Watson1969}
G.~N. Watson and E.~T. Whittaker, {\em A Course on Modern Analysis}.
\newblock Cambridge University Press, Cambridge U.K., 1969.

\bibitem{Casas2009}
F.~Casas and A.~Murua, {\it An ef{}f{}icient algorithm for computing the
  {Baker-Campbell-Hausdorf{}f} series and some of its applications},  {\em J.
  Math. Phys.} {\bf 50} (2009) 033513.

\end{thebibliography}\endgroup

\end{document}